\documentclass[%
 aip,
 amsmath,amssymb,
 reprint,%
]{revtex4-1}

\usepackage{graphicx}
\usepackage{dcolumn}
\usepackage{bm}

\usepackage[utf8]{inputenc}
\usepackage[T1]{fontenc}
\usepackage{mathptmx}
\usepackage{etoolbox}

\usepackage{color}
\usepackage{enumerate}
\usepackage{circuitikz}
\usepackage{subcaption}
\usepackage[breaklinks=true]{hyperref}
\usepackage{breakcites}

\DeclareMathOperator*{\ho}{\overset{\perp}{\mathcal{\oplus}}}
\DeclareMathOperator{\ran}{\operatorname{ran}}
\newcommand{\FSpace}[2]{\mathcal{F}(#1,#2)}

\DeclareMathOperator{\Hzero}{\mathcal{F}_\#(\mathbb{Z}^\text{$d$},\mathbb{K})}
\DeclareMathOperator{\Hone}{\mathcal{F}_\#(\text{$E$},\mathbb{K})}
\DeclareMathOperator{\Fzero}{\mathcal{F}(\mathbb{Z}^\text{$d$},\mathbb{K})}
\DeclareMathOperator{\Fone}{\mathcal{F}(\text{$E$},\mathbb{K})}

\newtheorem{theorem}{Theorem}
\newtheorem{definition}[theorem]{Definition}
\newtheorem{corollary}[theorem]{Corollary}
\newtheorem{lemma}[theorem]{Lemma}
\newtheorem{proposition}[theorem]{Proposition}
\newtheorem{remark}[theorem]{Remark}
\newtheorem{example}[theorem]{Example}
\newtheorem{notation}[theorem]{Notation}
\newenvironment{proof}[1][Proof]{\noindent\textbf{#1.} }{\ \rule{0.5em}{0.5em}}

\makeatletter
\def\@email#1#2{%
 \endgroup
 \patchcmd{\titleblock@produce}
  {\frontmatter@RRAPformat}
  {\frontmatter@RRAPformat{\produce@RRAP{*#1\href{mailto:#2}{#2}}}\frontmatter@RRAPformat}
  {}{}
}%
\makeatother
\begin{document}

\preprint{AIP/123-QED}

\title[Effective operators and variational principles for discrete networks]{Effective operators and their variational principles for discrete electrical network problems}
\author{K. Beard}
\affiliation{ 
Louisiana State University, Baton Rouge, LA USA
}%
\author{A. Stefan}%
\affiliation{ 
Florida Institute of Technology, Melbourne, FL USA
}%

\author{R. Viator, Jr.}
\affiliation{%
Swarthmore College, Swarthmore, PA USA
}%

\author{A. Welters}
\email{awelters@fit.edu}
\affiliation{ 
Florida Institute of Technology, Melbourne, FL USA
}%

\date{\today}

\begin{abstract}
Using a Hilbert space framework inspired by the methods of orthogonal projections and Hodge decompositions, we study a general class of problems (called $Z$-problems) that arise in effective media theory, especially within the theory of composites, for defining the effective operator. A new and unified approach is developed, based on block operator methods, for obtaining solutions of the $Z$-problem, formulas for the effective operator in terms of the Schur complement, and associated variational principles (e.g., the Dirichlet and Thomson minimization principles) that lead to upper and lower bounds on the effective operator. In the case of finite-dimensional Hilbert spaces, this allows for a relaxation of the standard hypotheses on positivity and invertibility for the classes of operators usually considered in such problems, by replacing inverses with the Moore-Penrose pseudoinverse. As we develop the theory, we show how it applies to the classical example from the theory of composites on the effective conductivity in the periodic conductivity problem in the continuum ($2d$ and $3d$) under the standard hypotheses. After that, we consider the following three important and diverse examples (increasing in complexity) of discrete electrical network problems in which our theory applies under the relaxed hypotheses. First, an operator-theoretic reformulation of the discrete Dirichlet-to-Neumann (DtN) map for an electrical network on a finite linear graph is given and used to relate the DtN map to the effective operator of an associated $Z$-problem.\ Second, we show how the classical effective conductivity of an electrical network on a finite linear graph is essentially the effective operator of an associated $Z$-problem. Finally, we consider electrical networks on periodic linear graphs and develop a discrete analog to classical example of the periodic conductivity equation and effective conductivity in the continuum.\footnote{This work is based in part on the M.S.\ thesis\cite{22KB} of the first author whose advisor was the fourth author.}
\end{abstract}

\maketitle

\section{\label{sec:intro}Introduction}

Broadly speaking, effective media (or medium) theory is an analytical or theoretical model which describes the effective (apparent, overall, or aggregate) behavior of a complex, multicomponent system, usually having a constitutive relation, in terms of a simpler system with an effective constitutive relation that defines the effective parameter, moduli, coefficient, tensor, matrix, operator, etc.\cite{78RL, 02GM, 13ST, 16TC, 19KM}, which we will denote here by $\sigma_*$. The effective properties of the original system are then often described in terms of this $\sigma_*$.

For example, in a continuum electrical conductivity problem with a conductivity tensor field $\sigma=\sigma(x)$ in a periodic medium, the constitutive relation governing the relation between the electric field $E=E(x)$ and the current density $J=J(x)$ is Ohm's law $J=\sigma E$ and the effective conductivity tensor $\sigma_*$ is defined by the relation $\langle J \rangle = \sigma_*\langle E \rangle$, where $\langle \cdot \rangle$ denotes the periodic average \cite{02GM, 16YG, 16GM}. 

Similar to this example, there are a common set of fundamental problems one wants to solve such as: Find sufficient conditions on the constitutive relations that guarantee $\sigma_*$ is uniquely defined by the effective constitutive relation and then give an explicit formula for $\sigma_*$. In addition, if possible, give variational principles for $\sigma_*$ and derive upper and lower bounds on it. 

To treat these fundamental problems, a mathematical structure has emerged from within the theory of composites \cite{02GM,16YG,16GM}, that has provided a common systematic method for defining the effective operator $\sigma_*$ starting from a ``$Z$-problem" (aka, a generalized ``field equation" or ``cell problem"), which represents the constitutive relation, and based on a Hilbert space framework in which orthogonal projections associated with a Hodge decomposition plays a key role. Originally, part of the framework arose out of homogenization theory \cite{78GP,81PV, 82KP, 83GP, 86GD} and then was further developed within the theory of composites\cite{87aGM, 87bGM, 88MK, 90GM, 02GM,16YG} (see also Chap.\ 2 in Ref.\ \onlinecite{16GM}). But now the abstract framework has been shown to be useful more generally\cite{16GM, 18YG, 17GM, 18GM, 19MO} and leads to several important open problems\cite{21GM} that is fueling research \cite{21AS, 21aAW, 21bAW}. As such, this framework deserves further study and should include new compelling examples that fit within the abstract framework. This provides one motivation for our paper.

In order to better understand the structure of a $Z$-problem, let us consider the example above in more detail.
\begin{example}\label{ExContinuumPeriodicCondZProb}
Consider the Hilbert space of periodic square-integrable vector-valued functions $\left[  L_{\#}^{2}\left(\Omega\right)\right]^{d}$ ($d=2$ or $d=3$, over the field $\mathbb{K}=\mathbb{R}$ or $\mathbb{K}=\mathbb{C}$) with unit cell $\Omega$ (e.g., $\Omega = [0,2\pi]^d$) and the Hodge decomposition: 
\begin{gather}
    \left[  L_{\#}^{2}\left(\Omega\right)\right]^{d}=\mathcal{H}=\mathcal{U}\ho\mathcal{E}\ho\mathcal{J},\label{HodgeDecompPeriodContinuumCond}\\
	\mathcal{U}  =\{U\in\mathcal{H}:\langle U \rangle=U\}, \label{HodgeDecompPeriodContinuumCond1}\\
	\mathcal{E}  =\{E\in\mathcal{H}:\nabla\times E=0,\; \langle E\rangle=0\}, \label{HodgeDecompPeriodContinuumCond2}\\
	\mathcal{J}  =\{J\in\mathcal{H}:\nabla\cdot J=0,\; \langle J\rangle=0\},\label{HodgeDecompPeriodContinuumCond3}
\end{gather}
with the inner product and (cell) average
\begin{align}
	\left(E,F\right)_{\mathcal{H}}=\frac{1}{\left\vert \Omega\right\vert }%
	{\textstyle\int\limits_{\Omega}}\overline{E\left(x\right)}^{T}F(x)dx,\;\; 
	\left\langle F\right\rangle                 
	=\frac{1}{\left\vert \Omega\right\vert }%
	{\textstyle\int\limits_{\Omega}}            
	F\left(  x\right)  dx,    \label{ExContinuumPeriodicCondZProbRealInnerProdAndAvg}          
\end{align}
respectively, for all $E,F\in\mathcal{H}$, where $(\cdot)^T$ and $\overline{(\cdot)}$ denote the transpose and complex conjugation, respectively. In particular, $\mathcal{U}$ is the $d$-dimensional space of uniform (constant) vector functions; $\mathcal{E}$ is the infinite-dimensional space of all $\Omega$-periodic fields $E$ characterized by $E=\nabla u$ for some $\Omega$-periodic function $u$; $\mathcal{U}\ho \mathcal{E}$ and $\mathcal{U}\ho \mathcal{J}$ are the spaces of periodic vector fields which are the gradient $\nabla$ of a potential and divergence ($\nabla\cdot $)-free, respectively. 

The $Z$-problem in this example is the constitutive relation, namely, Ohm's law (formulated as the field equation or cell problem \cite{02GM, 16YG, 16GM}) with a measurable, bounded, periodic conductivity tensor $\sigma=\sigma(x)$ [more precisely, it can be scalar-, tensor-, or $d\times d$ matrix-valued as a function of the spatial variable $x$, but the key point is that as a left multiplication operator, $\sigma$ is a bounded linear operator on $\mathcal{H}$, i.e., $\sigma\in \mathcal{L}(\mathcal{H})$] or, more generally, with any bounded linear operator $\sigma\in \mathcal{L}(\mathcal{H})$ (in particular, $\sigma$ need not be a local operator\cite{88MK, 90GM}):
\begin{align}
    J_0+J=\sigma (E_0+E),\label{DefZProbContConductivity}
\end{align}
where $E_0,J_0\in \mathcal{U},$ $E\in \mathcal{E}$, and $J\in \mathcal{J}$. Then effective conductivity $\sigma_*$ is defined in terms of this $Z$-problem as the bounded linear operator on $\mathcal{U}$ [i.e., $\sigma_*\in \mathcal{L}(\mathcal{U})$] satisfying the effective constitutive relation
\begin{align}
    J_0=\sigma_*E_0.\label{DefZProbContEffConductivity}
\end{align}
Equivalently, in terms of averages between the periodic electric field $E_0+E\in \mathcal{U}\ho \mathcal{E}$ and the periodic current density $J_0+J\in \mathcal{U}\ho \mathcal{J}$ related through Ohm's law (\ref{DefZProbContConductivity}), the effective constitutive relation (\ref{DefZProbContEffConductivity}) is the usual relation defining the effective conductivity:
\begin{align}
    \langle J_0+J \rangle = \sigma_* \langle E_0+E \rangle,
\end{align}
since $\langle E_0+E \rangle=E_0$ and $\langle J_0+J \rangle=J_0$. 
\end{example}

\begin{remark}
In this paper we will treat Hilbert spaces $\mathcal{H}$ over a field $\mathbb{K}$, where $\mathbb{K}=\mathbb{R}$, the real numbers, or $\mathbb{K}=\mathbb{C}$, the complex numbers. In either case, the abstract framework and our results are formulated in a general form regardless of which field is used. One motivation for including both real and complex fields comes from the theory of composites\cite{90GM, 02GM, 16YG, 16GM}: In static problems, usually one deals with real fields and hence a real Hilbert space $\mathcal{H}$. But for quasistatic problems, one also deals with complex fields and hence a complex Hilbert space $\mathcal{H}$.
\end{remark}

\begin{remark}\label{rem:DefPreciseSpacesContinuum}
For the Hodge decomposition above, there is an alternative characterization (see pp.\ 1487--1488 and Appendix A in Ref.\ \onlinecite{22LV} on the Helmholtz decomposition) of these spaces (which can be proved using Fourier series analysis, for instance) that will be useful when we consider discrete examples in Sec.\ \ref{sec:DiscreteNetworkExamples}:
$\left[  L_{\#}^{2}\left(\Omega\right)\right]^{d}$ denotes the Hilbert space (over the field $\mathbb{K}$) of all periodic (on $\Omega$) vector-valued functions belonging to $L^2_{loc}(\mathbb{R}^d,\mathbb{K}^d)$ with inner product (\ref{ExContinuumPeriodicCondZProbRealInnerProdAndAvg}),
\begin{gather}
    \mathcal{U}=\{U\in \left[  L_{\#}^{2}\left(\Omega\right)\right]^{d}:U \text{ is a constant function}\},\label{DefPreciseUSpaceContinuum}\\
    \mathcal{E}=\{\nabla u:u\in L_{\#}^{2}\left(\Omega\right) \text{ and } u\in H^1_{loc}(\mathbb{R}^d,\mathbb{K})\},\label{DefPreciseESpaceContinuum}\\
    \mathcal{J}=\{J\in \left[  L_{\#}^{2}\left(\Omega\right)\right]^{d}:J\in H_{loc}(\operatorname{div},\mathbb{R}^d,\mathbb{K}^d), \nabla\cdot J=0,\langle J\rangle=0\},\label{DefPreciseJSpaceContinuum}
\end{gather}
where $\nabla:H^1_{loc}(\mathbb{R}^d,\mathbb{K})\rightarrow L^2_{loc}(\mathbb{R}^d,\mathbb{K}^d)$ is the gradient operator and $\nabla\cdot:H_{loc}(\operatorname{div},\mathbb{R}^d,\mathbb{K}^d)\rightarrow L^2_{loc}(\mathbb{R}^d,\mathbb{K}^d)$ is the divergence operator.  Then
\begin{align*}
   \mathcal{J}=\{\nabla\times F:F\in \left[  L_{\#}^{2}\left(\Omega\right)\right]^{d},\nabla\times F\in \left[  L_{\#}^{2}\left(\Omega\right)\right]^{d}\},
\end{align*}
where $\nabla\times:H_{loc}(\operatorname{curl},\mathbb{R}^d,\mathbb{K}^d)\rightarrow L^2_{loc}(\mathbb{R}^d,\mathbb{K}^d)$ is the curl operator, and 
\begin{gather}
    \left[  L_{\#}^{2}\left(\Omega\right)\right]^{d}=\mathcal{U}\ho\mathcal{E}\ho\mathcal{J},\label{PreciseHodgeDecompContinuum}\\
    \operatorname{ker}(\nabla\cdot)\cap \left[  L_{\#}^{2}\left(\Omega\right)\right]^{d}=\mathcal{U}\ho\mathcal{J},\;\operatorname{ran}(\nabla)\cap \left[  L_{\#}^{2}\left(\Omega\right)\right]^{d}=\mathcal{U}\ho\mathcal{E}.\label{PreciseUEAndUJSpaceContinuum}
\end{gather}
It is these representations, specifically, (\ref{DefPreciseUSpaceContinuum})-- (\ref{PreciseUEAndUJSpaceContinuum}), that we will consider for an analogy in the periodic lattice example in Sec.\ \ref{SectCondInf}, but we will find that there are nuances which lead us to conclude that the analogy is not a perfect one (see Sec.\ \ref{SecEffCond}).
\end{remark}

\subsection{\label{sec:intro:subsec:ZprobEffOp}What is the \textit{Z}-problem and effective operator?}

The above discussion motivations the following precise definition of the $Z$-problem (a term coined in Ref.\ \onlinecite{16GM}, see Chaps.\ 7 \& 10, that also may be called a generalized ``cell problem" as coined in Ref.\ \onlinecite{18YG} or ``field equation" following \onlinecite{87aGM, 87bGM, 88MK, 90GM, 02GM, 16GM}) associated with a Hilbert space having an orthogonal triple decomposition [or, in light of (\ref{HodgeDecompPeriodContinuumCond})-(\ref{HodgeDecompPeriodContinuumCond3}) and Ref.\ \onlinecite{86GD}, a ``generalized Hodge decomposition" as coined in Ref.\ \onlinecite{18YG}] and effective operator (or ``$Z$-operator," see Chaps.\ 7 \& 10  in Ref.\ \onlinecite{16GM}) that comes from the abstract theory of composites \cite{02GM,16GM}.
\begin{definition}[$Z$-problem and effective operator]\label{DefZProbMain}
	The $Z$-problem
	\begin{equation}
		(\mathcal{H},\mathcal{U},\mathcal{E},\mathcal{J},\sigma), \label{DefZProb}
	\end{equation}
	is the following problem associated with a Hilbert space $\mathcal{H}$, an orthogonal triple decomposition of $\mathcal{H}$ as 
	\begin{equation}
		\mathcal{H=U}\overset{\bot}{\mathcal{\oplus}}\mathcal{E}\overset{\bot
			}{\mathcal{\oplus}}\mathcal{J},\label{DefZProbHOrthTri}
	\end{equation}
	and a (bounded) linear operator $\sigma\in\mathcal{L}(\mathcal{H})$:
	given $E_{0}\in\mathcal{U}$, find triples $\left(J_{0},E,J\right)\in\mathcal{U}\times\mathcal{E}\times\mathcal{J}$ satisfying 
	\begin{equation}
		J_{0}+J=\sigma \left(  E_{0}+E\right),
		\label{DefZProbEq} 
	\end{equation}
	such a triple $\left(  J_{0},E,J\right)$ is called a solution of the $Z$-problem at $E_{0}$.
	If there exists a (bounded) linear operator $\sigma_*\in\mathcal{L}(\mathcal{U}$) such that 
	\begin{equation}
		J_{0}=\sigma_{\ast}E_{0}, \label{DefZProbEffOp}
	\end{equation}
	whenever $E_0\in\mathcal{U}$ and $\left(  J_{0},E,J\right)$ is a solution of the $Z$-problem at $E_0$, then $\sigma_*$ is called an effective operator of the $Z$-problem.
\end{definition}

For instance, it follows from this definition that Example \ref{ExContinuumPeriodicCondZProb} defines a $Z$-problem $(\mathcal{H},\mathcal{U},\mathcal{E},\mathcal{J},\sigma)$, where $\mathcal{H}=\left[L_{\#}^{2}\left(\Omega\right)\right]^{d}$ has the orthogonal triple decomposition (\ref{HodgeDecompPeriodContinuumCond}) with $\mathcal{U}, \mathcal{E}, \mathcal{J}$ given by (\ref{HodgeDecompPeriodContinuumCond1}), (\ref{HodgeDecompPeriodContinuumCond2}), (\ref{HodgeDecompPeriodContinuumCond3}) (which is a special case of a Hodge decomposition) and the effective conductivity $\sigma_*$ is an effective operator of that $Z$-problem according to our definition above.

\subsection{\label{sec:intro:subsec:MainProbs}Overview of the paper}

To every $Z$-problem $(\mathcal{H},\mathcal{U},\mathcal{E},\mathcal{J},\sigma)$ (e.g., Example \ref{ExContinuumPeriodicCondZProb}), there are three important sets of problems that naturally arise: 
\begin{enumerate}
    \item[\hypertarget{(i)}{(i)}] (Solvability of the $Z$-problem) Under what conditions does the $Z$-problem (\ref{DefZProbEq}) have a solution; a unique solution? If possible, find a formula for all solutions in terms of $\sigma$ and parameterized by those $E_0\in \mathcal{U}$ for which solutions exist.
    \item[\hypertarget{(ii)}{(ii)}] (Existence, uniqueness, and representation formulas for the effective operator) Under what conditions does the effective operator $\sigma_*$ exist; is unique?  If it exists, find representation formulas for it in terms of $\sigma$.
    \item[\hypertarget{(iii)}{(iii)}] (Variational principles \& bounds) If $\sigma^*=\sigma$ (i.e., self-adjoint) and $\sigma\geq 0$ (i.e., positive semidefinite), are there variational principles: $(1)$ for the solutions of the $Z$-problem? $(2)$ that define the effective operator $\sigma_*$? $(3)$ that can be used to derive upper and lower bounds on an effective operator $\sigma_*$ in terms of $\sigma$?
\end{enumerate}

In this paper, we will focus on answering these questions. To do so, we will first consider classical results in Sec.\ \ref{sec:ClassicalResults}. Here, certain ``strong hypotheses" will be used such as $\sigma^*=\sigma\geq 0$ and $\sigma$ invertible (the ``classical" hypotheses). Often the latter will be too restrictive though and the invertibility of the subblock $\sigma_{11}$ of $\sigma$ [see (\ref{3b3BlockOpReprOfSigma})-(\ref{DefAltSigma11Subblock})] will be enough. Second, in Sec.\ \ref{sec:RelaxingHyps}, we discuss how we plan to relax or weaken the hypotheses and then give insight into the reason these ``weaker hypotheses" [see \hyperlink{(H1)}{(H1)-(H4)}] will naturally occur. Next, in order to extend the results in Sec.\ \ref{sec:ClassicalResults} under the weaker hypotheses, we introduce a unified framework for treating the variational principles associated to constrained linear equations in Sec.\ \ref{sec:VarPrincConstLinearEqsUnifiedFramework}. Then, in Sec.\ \ref{sec:MainResults}, we use this framework to prove the main results of this paper on the $Z$-problem and effective operator which answer the sets of questions \hyperlink{(i)}{(i)}, \hyperlink{(i)}{(ii)}, and \hyperlink{(i)}{(iii)} above under the weaker hypotheses. After this, we will develop in Sec. \ref{sec:DiscreteNetworkExamples}, three important examples of discrete electrical network problems for which our abstract framework and results apply. These examples provide far more than just evidence that our abstract framework can be used to solve for their associated effective operators but also, equally as important, there exists a connecting relationship between all of them, which becomes apparent from the results in this paper. In Appendices \ref{SectAbsTheoryCompositesVecSpFramework} and \ref{SectAbsHodgDecomp}, we provide some fundamental results on the abstract theory of composites, effective operators, and Hodge decompositions that are needed in this paper and interesting in their own right, but are best placed in these appendices.

In regard to the approach we use in this paper to treat the problems \hyperlink{(i)}{(i)-(iii)}, it is new, even in the classical setting. It is based on block operator methods and Schur complement theory. As such, we will review the classical results and include proofs using this new approach, before generalizing them later in the paper.

Finally, although our focus is on weakening the hypotheses on the operator $\sigma$ to achieve comparable results in the classical setting, we will still assume, in general, some form of self-adjoint hypotheses, e.g., $\sigma^*=\sigma$, when we state and prove our results. The reason for this is two fold. The first is that $\sigma^*=\sigma$ is the most common hypothesis in the examples we consider in this paper. The second reason is that the other typical hypotheses on $\sigma$, that it has positive imaginary or positive real part (see, for instance, Refs.\ \onlinecite{90GM, 94CG, 02GM, 16GM}), is important enough to deserve a separate study of which the current paper will be useful.

\section{\label{sec:ClassicalResults}Classical results}

In the classical setting, i.e., with the strong assumptions that $\sigma^*=\sigma\geq 0$ and $\sigma$ invertible, all these problems can solved\cite{88MK, 90GM, 02GM, 16GM}. We will review this now and include proofs using a new approach to solving these problems which emphasizes block operator methods and the use of Schur complements\cite{05FZ} (although the importance of it in the theory of composites was originally recognized in Ref.\ \onlinecite{16CW}). Later in the paper we generalize our approach when we relax the invertibility hypothesis.   

Let $(\mathcal{H},\mathcal{U},\mathcal{E},\mathcal{J},\sigma)$ be a $Z$-problem (as defined in Def.\ \ref{DefZProbMain}). Then we can write the operator
\begin{align}
 \sigma=[\sigma_{ij}]_{i,j=0,1,2}\in\mathcal{L}(\mathcal{H}) \label{3b3BlockOpReprOfSigma}  
\end{align}
as a $3\times 3$ block operator matrix with respect to the orthogonal triple decomposition (\ref{DefZProbHOrthTri}) of the Hilbert space $\mathcal{H}=\mathcal{U}\ho\mathcal{E}\ho\mathcal{J}$. More precisely, we introduce the orthogonal projections $\Gamma_0,\Gamma_1,\Gamma_2$ of $\mathcal{H}$ onto $H_0=\mathcal{U}, H_1=\mathcal{E}, H_2=\mathcal{J},$ respectively, and define
\begin{align}
    \sigma_{ij}\in \mathcal{L}(H_j,H_i),\;\sigma_{ij}=\Gamma_i\sigma\Gamma_j:H_j\rightarrow H_i,\label{DefOfSigmaSubblocks}
\end{align}
for $i,j=0,1,2$. In particular, for $i=j=1,$ $\sigma_{11}$ is the compression of $\sigma$ to $\mathcal{E}$, that is,
\begin{align}
   \sigma_{11}= \Gamma_1\sigma\Gamma_1|_{\mathcal{E}},\label{DefAltSigma11Subblock}
\end{align}
i.e., the restriction of the operator $\Gamma_1\sigma\Gamma_1$ on $\mathcal{H}$ to the closed subspace $\mathcal{E}$. Then the $Z$-problem (\ref{DefZProbEq}) is equivalent to the system
\begin{gather}
    \sigma_{00}E_0+\sigma_{01}E=J_0,\label{ZProbEquivFormPart1}\\
        \sigma_{10}E_0+\sigma_{11}E=0,\label{ZProbEquivFormPart2}\\
        \sigma_{20}E_0+\sigma_{21}E=J.\label{ZProbEquivFormPart3}
\end{gather}
Finally, from this and assuming $\sigma_{11}$ is invertible, we get the classical formulas for the solution of the $Z$-problem and the effective operator as a Schur complement:
\begin{gather}
    J_0=\sigma_*E_0,\;E=-\sigma_{11}^{-1}\sigma_{10}E_0,\; J=\sigma_{20}E_0+\sigma_{21}E,\label{ClassicSolnZProb}\\
    \sigma_*=\begin{bmatrix}
        \sigma_{00}&\sigma_{01}\\
        \sigma_{10}&\sigma_{11}
    \end{bmatrix}/\sigma_{11}=\sigma_{00}-\sigma_{01}\sigma_{11}^{-1}\sigma_{10}.\label{ClassicEffOperFormula}
\end{gather}

This proves the following theorem which answers the questions \hyperlink{(i)}{(i)} and \hyperlink{(ii)}{(ii)} above in the case $\sigma_{11}$ is invertible.
\begin{theorem}\label{ThmMainClassicalZProbEffOp}
If $(\mathcal{H},\mathcal{U},\mathcal{E},\mathcal{J},\sigma)$ is a $Z$-problem (as in Def.\ \ref{DefZProbMain}) and $\sigma_{11}$ [as defined by (\ref{DefAltSigma11Subblock})] is invertible then the $Z$-problem (\ref{DefZProbEq}) has a unique solution for each $E_0\in \mathcal{U}$ and it is given by the formulas (\ref{ClassicSolnZProb}), (\ref{ClassicEffOperFormula}). Moreover, the effective operator of the $Z$-problem exists, is unique, and is given by the Schur complement formula (\ref{ClassicEffOperFormula}).
\end{theorem}

The Schur complement formula (\ref{ClassicEffOperFormula}) for the effective operator $\sigma_*$ is useful for deriving some of it's basic properties as we shall now see.
\begin{corollary}\label{CorClassicalBasicPropEffOp}
If $\sigma\in \mathcal{L}(\mathcal{H})$ and $\sigma_{11}$ is invertible then
\begin{align}
    (\sigma_*)^*=(\sigma^*)_*.
\end{align}
In particular, if $\sigma^*=\sigma$ then 
\begin{align}
    (\sigma_*)^*=\sigma_*.
\end{align}
\end{corollary}
\begin{proof}
Assume the hypotheses. Then since $(\sigma_{ij})^*=(\sigma^*)_{ji}$ for each $i,j=0,1,2$ it follows that $(\sigma_{11})^*=(\sigma^*)_{11}$ is invertible so by Theorem \ref{ThmMainClassicalZProbEffOp} applied to the adjoint $Z$-problem, i.e., to the $Z$-problem $(\mathcal{H},\mathcal{U},\mathcal{E},\mathcal{J},\sigma^*)$, we have following formula for the associated effective operator $(\sigma^*)_*$:
\begin{align*}
    (\sigma^*)_*&=(\sigma^*)_{00}-(\sigma^*)_{01}(\sigma^*)_{11}^{-1}(\sigma^*)_{10}\\
    &=(\sigma_{00})^*-(\sigma_{10})^*(\sigma_{11}^{-1})^*(\sigma_{01})^*\\
    &=[\sigma_{00}-\sigma_{01}(\sigma_{11})^{-1}\sigma_{10}]^*\\
    &=(\sigma_*)^*.
\end{align*}
This proves the corollary.
\end{proof}

Now under the self-adjoint assumptions that $\sigma\in \mathcal{L}(\mathcal{H})$ and $\sigma^*=\sigma$, the next two theorems answer the questions \hyperlink{(iii)}{(iii)}.$(2)$ under slightly weaker hypotheses, namely, in the case $\sigma_{11}\geq 0$ and $\sigma_{11}$ is invertible and all of \hyperlink{(iii)}{(iii)} in the case $\sigma\geq 0$ and $\sigma$ is invertible in which the notion of duality plays a key role.
\begin{theorem}[Dirichlet minimization principle]\label{ThmClassicalDiriMinPrin}
If $\sigma\in \mathcal{L}(\mathcal{H}),$ $\sigma^*=\sigma$, $\sigma_{11}\geq 0$, and $\sigma_{11}$ is invertible then the effective operator $\sigma_*$ is the unique self-adjoint operator satisfying the minimization principle:
\begin{equation}
		( E_0,\sigma_*E_0 )=\min_{E\in\mathcal{E}}( E_0+E,\sigma(E_0+E) ),\;\forall E_0\in\mathcal{U},\label{ClassicalDirMinPrincEffOp}
	\end{equation}
	and, for each $E_0\in\mathcal{U}$, the minimizer is unique and given by
	\begin{equation}
		E=-\sigma_{11}^{-1}\sigma_{10}E_0.
	\end{equation}
	Moreover, we have the following upper bound on the effective operator:
	\begin{equation}
		\sigma_*\leq \sigma_{00},
	\end{equation}
	where
	\begin{align}
	    \sigma_{00}=\Gamma_0\sigma\Gamma_0|_{\mathcal{U}}.
	\end{align}
\end{theorem}
In fact, as our proof of this theorem below shows, the theorem is still true if we drop the self-adjoint hypothesis $\sigma^*=\sigma$ and replace it with the weaker hypothesis:
\begin{align}
    \begin{bmatrix}
        \sigma_{00}&\sigma_{10}\\
        \sigma_{10}&\sigma_{11}
    \end{bmatrix}^*=\begin{bmatrix}
        \sigma_{00}&\sigma_{10}\\
        \sigma_{10}&\sigma_{11}
    \end{bmatrix},\label{Weaker2x2BlockSelfAdjCondOnSigma}
\end{align}
where
\begin{align}
    \begin{bmatrix}
        \sigma_{00}&\sigma_{10}\\
        \sigma_{10}&\sigma_{11}
    \end{bmatrix}=(\Gamma_0+\Gamma_1)\sigma(\Gamma_0+\Gamma_1)|_{\mathcal{U}\ho\mathcal{E}}.
\end{align}
In other words, the weaker hypothesis (\ref{Weaker2x2BlockSelfAdjCondOnSigma}) is equivalent to the hypothesis that the compression of $\sigma$ to $\mathcal{U}\ho\mathcal{E}$ is self-adjoint.

We can use Theorem \ref{ThmClassicalDiriMinPrin} to derive some additional properties of the effective operator $\sigma_*$.
\begin{corollary}\label{CorClassicalBasicInvPosPropEffOp}
Suppose $\sigma\in \mathcal{L}(\mathcal{H}),$ $\sigma^*=\sigma\geq 0$. Then the following are true: $(a)$ If $\sigma_{11}$ is invertible then $(\sigma_*)^*=\sigma_*\geq 0$. $(b)$ If $\sigma$ is invertible then $0\leq (\sigma_{jj})^*=\sigma_{jj}$ are invertible for each $j=0,1,2$ and $\sigma_*$ is invertible.
\end{corollary}
\begin{proof}
Suppose $\sigma\in \mathcal{L}(\mathcal{H}),$ $\sigma^*=\sigma\geq 0$. $(a)$: If $\sigma_{11}$ is invertible then by Corollary \ref{CorClassicalBasicPropEffOp} we know $(\sigma_*)^*=\sigma_*$ and by (\ref{ClassicalDirMinPrincEffOp}) in Theorem \ref{ThmClassicalDiriMinPrin} together with the hypothesis $\sigma\geq 0$, it follows immediately that $\sigma_*\geq 0$. $(b)$: Suppose $\sigma$ is invertible. Then it follows from the hypotheses that $\sigma\geq \delta I_{\mathcal{H}}$ for some scalar $\delta>0,$ where $I_{\mathcal{H}}$ denotes the identity operator on $\mathcal{H}$. From this it follows that $(\sigma_{jj})^*=\sigma_{jj}=\Gamma_j\sigma\Gamma_{j}|_{H_j}\geq \delta I_{H_j}$, where $I_{H_j}$ is the identity operator on the Hilbert space $H_j$ for each $j=0,1,2$ with $H_0=\mathcal{U}, H_1= \mathcal{E}, H_2=\mathcal{J}$. And this implies $\sigma_{jj}$ is invertible for each $j=0,1,2$. Now it follows immediately from (\ref{ClassicalDirMinPrincEffOp}) in Theorem \ref{ThmClassicalDiriMinPrin} that $(\sigma_*)^*=\sigma_*\geq \delta I_{\mathcal{U}}$ which implies $\sigma_*$ is invertible.
\end{proof}
\begin{theorem}[Thomson minimization principle]\label{ThmClassicalThomMinPrin}
If $\sigma\in \mathcal{L}(\mathcal{H}),$ $\sigma^*=\sigma\geq 0,$ and $\sigma$ is invertible then $(\sigma_*)^{-1}$ is the unique self-adjoint operator satisfying the minimization principle:
	\begin{gather}
		(J_0,(\sigma_*)^{-1}J_0 )=\min_{J\in\mathcal{J}}( J_0+J,\sigma^{-1}(J_0+J)),\; \forall J_0\in\mathcal{U},
	\end{gather}
	and, for each $J_0\in\mathcal{U}$, the minimizer is unique and given by
	\begin{equation}
		J=-\sigma_{22}^{-1}\sigma_{20}J_0.
	\end{equation}
	Moreover, we have the upper and lower bounds on the effective operator:
	\begin{equation}
		0\leq [(\sigma^{-1})_{00}]^{-1}]\leq\sigma_*\leq \sigma_{00},\label{ClassicalUpperLowerBddsEffOp}
	\end{equation}
	where 
	\begin{align}
	    (\sigma^{-1})_{00}=\Gamma_0\sigma^{-1}\Gamma_0|_{\mathcal{U}}.
	\end{align}
\end{theorem}

We will give a proof of these two theorems which is surprisingly simple now and based on the following well-known minimization principle for Schur complements\cite{05FZ, 00LM} and the notion of duality between the direct and dual $Z$-problems as defined next. In this regard, the Dirichlet variational principle, i.e., Theorem \ref{ThmClassicalDiriMinPrin}, is called the direct variational principle and the Thomson (or Thompson) variational principle, i.e., Theorem \ref{ThmClassicalThomMinPrin}, is called the dual (or complementary) variational principle (see, for instance, Refs.\ \onlinecite{90GM, 02GM, 16GM}).

We will need the following well-known result\cite{05FZ} (the proof is omitted as it is trivial).
\begin{lemma}[Aitken block-diagonalization formula]\label{LemAitkenBlockDiag}
If $H=H_0\ho H_1$ is a Hilbert space, $A=[A_{i,j}]_{i,j=0,1}\in\mathcal{L}\left(H\right)$, and $A_{11}$ is invertible then
\begin{gather}
\begin{bmatrix}
A_{00} & A_{01}\\
A_{10} & A_{11}%
\end{bmatrix}
=%
\begin{bmatrix}
I_{H_0}  & A_{01}A_{11}^{-1}\\
0 & I_{H_1} 
\end{bmatrix}%
\begin{bmatrix}
A/A_{11} & 0\\
0 & A_{11}%
\end{bmatrix}%
\begin{bmatrix}
I_{H_0}  & 0\\
A_{11}^{-1}A_{10} & I_{H_1} 
\end{bmatrix},\label{LemAitkenBlockDiagTheFormula}\\
\begin{bmatrix}
I_{H_0}  & A_{01}A_{11}^{-1}\\
0 & I_{H_1} 
\end{bmatrix}^{-1}=\begin{bmatrix}
I_{H_0}  & -A_{01}A_{11}^{-1}\\
0 & I_{H_1} 
\end{bmatrix},\\
\begin{bmatrix}
I_{H_0}  & 0\\
A_{11}^{-1}A_{10} & I_{H_1} 
\end{bmatrix}^{-1}=\begin{bmatrix}
I_{H_0}  & 0\\
-A_{11}^{-1}A_{10} & I_{H_1} 
\end{bmatrix}.
\end{gather}
\end{lemma}
\begin{theorem}
[Schur complement: minimization principle]%
\label{ThmSchurComplConstrMinPrinc}If $H=H_0\ho H_1$ is a Hilbert space, $A=[A_{i,j}]_{i,j=0,1}\in\mathcal{L}\left(H\right)$, $A^*=A,$ $A_{11}$ is
invertible and $A_{11}\geq0$ then the Schur complement of $A$ with respect to $A_{11}$, i.e., $A/A_{11}=A_{00}-A_{01}A_{11}^{-1}A_{10},$ is the unique self-adjoint operator satisfying the minimization principle:%
\begin{align}
\left(  x,A/A_{11}x\right)  =\min_{y\in H_{1}}\left(  x+y,A\left(
x+y\right)  \right)  ,\text{ }\forall x\in H_{0},\label{ClassicalSchurComplMinPrinc}
\end{align}
and, for each $x\in H_{0}$, the minimizer is unique and given by
\begin{align}
y=-A_{11}^{-1}A_{10}x\in H_{1}.   \label{ClassicalSchurComplMinPrincMinimizer} 
\end{align}
Moreover,
\begin{align}
    A/A_{11}\leq A_{00}.\label{ClassicalSchurComplUpperBdd} 
\end{align}
\end{theorem}
\begin{proof}
Assume the hypotheses. Then $(A_{11})^*=A_{11}\geq\delta I,$ for some scalar $\delta>0$. Next, it follows by Lemma \ref{LemAitkenBlockDiag} that $A$ has the block
factorization (\ref{LemAitkenBlockDiagTheFormula}) with
\begin{gather*}
\begin{bmatrix}
I_{H_0} & A_{01}A_{11}^{-1}\\
0 & I_{H_1}
\end{bmatrix}
^{\ast}=%
\begin{bmatrix}
I_{H_0} & 0\\
A_{11}^{-1}A_{10} & I_{H_1}
\end{bmatrix}.
\end{gather*}
Hence, for any $x\in H_{0}$ and
$y\in H_{1}$,%
\begin{gather*}
\left(  x+y,A\left(  x+y\right)  \right)    =\left(
\begin{bmatrix}
x\\
y
\end{bmatrix}
,%
\begin{bmatrix}
A_{00} & A_{01}\\
A_{10} & A_{11}%
\end{bmatrix}%
\begin{bmatrix}
x\\
y
\end{bmatrix}
\right) \\
  =\left(
\begin{bmatrix}
I & 0\\
A_{11}^{-1}A_{10} & I
\end{bmatrix}%
\begin{bmatrix}
x\\
y
\end{bmatrix}
,%
\begin{bmatrix}
A/A_{11} & 0\\
0 & A_{11}%
\end{bmatrix}%
\begin{bmatrix}
I & 0\\
A_{11}^{-1}A_{10} & I
\end{bmatrix}%
\begin{bmatrix}
x\\
y
\end{bmatrix}
\right) \\
  =\left(  x,A/A_{11}x\right)  +\left(  A_{11}^{-1}A_{10}x+y,A_{11}\left(
A_{11}^{-1}A_{10}x+y\right)  \right) \\ \geq\left(  x,A/A_{11}x\right)  ,
\end{gather*}
with equality iff $\left(  A_{11}^{-1}A_{10}x+y,A_{11}\left(
A_{11}^{-1}A_{10}x+y\right)  \right)  =0$ iff $A_{11}^{-1}%
A_{10}x+y=0$ iff $y=-A_{11}^{-1}A_{10}x$. This proves the equality (\ref{ClassicalSchurComplMinPrinc}) and the uniqueness of the minimizer which is given by (\ref{ClassicalSchurComplMinPrincMinimizer}). The proof of the theorem now follows from all this and the fact that $A/A_{11}\in\mathcal{L}\left(H_{0}\right)$ is
self-adjoint so is uniquely determined by its quadratic form (see, for instance, Theorem 4.4 in Ref.\ \onlinecite{80JW}) and the upper bound (\ref{ClassicalSchurComplUpperBdd}) follows from considering the RHS of (\ref{ClassicalSchurComplMinPrinc}) with $y=0\in H_1$.
\end{proof}

\begin{definition}
[Dual $Z$-problem]\label{DefDualZProb}Given a $Z$-problem $(\mathcal{H},\mathcal{U},\mathcal{E},\mathcal{J},\sigma)$ with invertible $\sigma$ (the direct $Z$-problem), the corresponding dual $Z$-problem is the $Z$-problem $(\mathcal{H},\mathcal{U},\mathcal{J},\mathcal{E},\sigma^{-1})$
associated with the orthogonal decomposition
\begin{equation}
\mathcal{H=U}\overset{\bot}{\mathcal{\oplus}}\mathcal{J}\overset{\bot
}{\mathcal{\oplus}}\mathcal{E}
\label{DefDualZProbOrthogonalTripleDecomposition}%
\end{equation}
and operator $\sigma^{-1}\in\mathcal{L}\left(  \mathcal{H}\right)$, i.e., the problem: given $J_{0}\in\mathcal{U}$, find
vectors $\left(  E_{0},J,E\right)  \in\mathcal{U}\times\mathcal{J}%
\times\mathcal{E}$ satisfying%
\begin{equation}
E_{0}+E=\sigma^{-1}\left(  J_{0}+J\right)  . \label{DefDualZProbEq}%
\end{equation}
An effective operator for this $Z$-problem will be denoted by $\left(
\sigma^{-1}\right)_{\ast^{\prime}}  $. In other words, if there exists an operator $\left(
\sigma^{-1}\right)  _{\ast^{\prime}}\in\mathcal{L}\left(  \mathcal{U}\right)  $
such that
\begin{equation}
E_{0}=\left(  \sigma^{-1}\right)  _{\ast^{\prime}}J_{0}
\label{DefDualZProbDefEffOp}%
\end{equation}
whenever $\left(  E_{0},J,E\right)  $ is a solution of the dual $Z$-problem at
$J_{0}$ then $\left(  \sigma^{-1}\right)  _{\ast^{\prime}}$ is an effective
operator for the dual $Z$-problem.
\end{definition}

\begin{proof}[Proof of Theorems \ref{ThmClassicalDiriMinPrin} and \ref{ThmClassicalThomMinPrin}]
Assume the hypotheses of Theorem \ref{ThmClassicalDiriMinPrin}. Then the proof of Theorem \ref{ThmClassicalDiriMinPrin} follows immediately from Theorem \ref{ThmMainClassicalZProbEffOp} with the Schur complement formula (\ref{ClassicEffOperFormula}) for $\sigma_*$ and by Theorem \ref{ThmSchurComplConstrMinPrinc} with 
\begin{gather}
    \sigma_*=A/A_{11},
\end{gather}
where
\begin{gather}
    H=H_0\ho H_1, H_0=\mathcal{U}, H_1=\mathcal{E},\\ A=[A_{ij}]_{i,j=0,1}=[\sigma_{ij}]_{i,j=0,1}.
\end{gather}
Now notice it follows from this proof that we can replace the hypothesis $\sigma^*=\sigma$ in Theorem \ref{ThmClassicalDiriMinPrin} with the weaker hypothesis \ref{Weaker2x2BlockSelfAdjCondOnSigma} and the conclusions of Theorem \ref{ThmClassicalDiriMinPrin} are still true.

Next, assume the hypotheses in Theorem \ref{ThmClassicalThomMinPrin}. Then by Corollary \ref{CorClassicalBasicInvPosPropEffOp} we know that $0\leq (\sigma_{jj})^*=\sigma_{jj}$ are invertible for each $j=0,1,2$ and similarly for $\sigma^{-1}$ since by hypotheses it follows that $(\sigma^{-1})=\sigma^{-1}\geq 0$. Thus, by Theorem \ref{ThmMainClassicalZProbEffOp} (in particular the uniqueness of the effective operator) it follows that the effective operator $(\sigma^{-1})_{*'}$ of the dual problem $(\mathcal{H},\mathcal{U},\mathcal{J},\mathcal{E},\sigma^{-1})$ satisfies
\begin{align}
    (\sigma^{-1})_{*'}=(\sigma_*)^{-1},
\end{align}
and by Corollary \ref{CorClassicalBasicInvPosPropEffOp} we also know that $(\sigma^{-1})_{*'}\geq 0$ and $\sigma_*\geq 0$. The result now follows immediately from this by Theorem \ref{ThmClassicalDiriMinPrin} applied to the dual $Z$-problem. In particular, we have from Theorem \ref{ThmClassicalDiriMinPrin}, applied to both the direct and dual $Z$-problem, that $0\leq \sigma_*\leq \sigma_{00}$ and $0\leq (\sigma_*)^{-1}=(\sigma^{-1})_{*'}\leq (\sigma^{-1})_{00}$ so by taking inverses we prove (\ref{ClassicalUpperLowerBddsEffOp}). This completes the proof.
\end{proof}

We conclude this section on some fundamental results on the dual problem which we will generalize under weaker hypotheses later in the paper. To do this, we will need the following lemma along with a proposition that gives an alternative Schur complement representation formula for the effective operator $\sigma_*$.
\begin{lemma}[Babachiewicz inversion formula]\label{LemBabachiewiczInversionFormula}
Suppose $H=H_0\ho H_1$ is a Hilbert space and $A=[A_{i,j}]_{i,j=0,1}\in\mathcal{L}\left(H\right)$. Then the following statements are true:
\begin{itemize}
    \item[(a)] If $A_{11}$ is invertible then $A$ is invertible if and only if $A/A_{11}$ is invertible.
    \item[(b)] If $A$ is invertible then $A_{11}$ is invertible if and only if $(A^{-1})_{00}$ is invertible.
    \item[(c)] If $A$ and $A_{11}$ are invertible then
\begin{gather}
A^{-1}=
\begin{bmatrix}
\left(  A^{-1}\right)  _{00} & \left(  A^{-1}\right)  _{01}\nonumber\\
\left(  A^{-1}\right)  _{10} & \left(  A^{-1}\right)  _{11}%
\end{bmatrix}
\\
=
\begin{bmatrix}
\left(  A/A_{11}\right)  ^{-1} & -\left(  A/A_{11}\right)  ^{-1}A_{01}
A_{11}^{-1}\\
-A_{11}^{-1}A_{10}\left(  A/A_{11}\right)  ^{-1} & A_{11}^{-1}+A_{11}
^{-1}A_{10}\left(  A/A_{11}\right)  ^{-1}A_{01}A_{11}^{-1}
\end{bmatrix},\\
    A/A_{11}=\left[  \left(  A^{-1}\right)  _{00}\right]  ^{-1},\;
A^{-1}/\left(  A^{-1}\right)  _{00}=A_{11}^{-1}.
\end{gather}
\end{itemize}
\end{lemma}
\begin{proof}
     Although this lemma is well-known\cite{05FZ}, we will give a self-contained proof. Assume the hypotheses. The proof of $(a)$ and $(c)$ follows immediately from Lemma \ref{LemAitkenBlockDiag}. To prove $(b)$, assume $A$ is invertible. If $A_{11}$ is invertible then by (c) we immediately get that $\left(  A^{-1}\right)  _{00}$ is invertible. Conversely, if $\left(  A^{-1}\right)  _{00}$ is invertible, then considering the same statement but with respect to $H=H_1\ho H_0$ and $A^{-1}=[(A^{-1})_{ij}]_{i,j=1,0}$ instead of $A$, we have that $\left(  (A^{-1})^{-1}\right)  _{11}$ is invertible. This proves the converse of $(b)$.
\end{proof}

\begin{proposition}\label{PropAltReprEffOp}
    If $(\mathcal{H},\mathcal{U},\mathcal{E},\mathcal{J},\sigma)$ is a $Z$-problem with $\sigma_{11}$ invertible then
    \begin{align}
        \sigma_*=\Gamma_0\sigma/\sigma_{11}\Gamma_0|_{\mathcal{U}}=(\sigma/\sigma_{11})_{00}=\sigma_{00}-\sigma_{01}\sigma_{11}^{-1}\sigma_{10},
    \end{align}
    where $\sigma/\sigma_{11}$ is the Schur complement of $\sigma$ with respect to $\sigma_{11}$ and the orthogonal decomposition $\mathcal{H}=(\mathcal{U}\ho \mathcal{J})\ho \mathcal{E}$. More precisely, with respect to the decomposition $\mathcal{H}=(\mathcal{U}\ho \mathcal{J})\ho \mathcal{E}$, the following block formulas hold:
    \begin{gather}
        \sigma=\begin{bmatrix}%
\begin{bmatrix}
\sigma_{00} & \sigma_{02}\\
\sigma_{20} & \sigma_{22}%
\end{bmatrix}
&
\begin{bmatrix}
\sigma_{01}\\
\sigma_{21}%
\end{bmatrix}
\\%
\begin{bmatrix}
\sigma_{10} & \sigma_{12}%
\end{bmatrix}
& \sigma_{11}%
\end{bmatrix},\\
\sigma/\sigma_{11}=\begin{bmatrix}
\sigma_{00} & \sigma_{02}\\
\sigma_{20} & \sigma_{22}%
\end{bmatrix}-\begin{bmatrix}
\sigma_{01}\\
\sigma_{21}%
\end{bmatrix}\sigma_{11}^{-1}\begin{bmatrix}
\sigma_{10} & \sigma_{12}%
\end{bmatrix}\\
=\begin{bmatrix}
\sigma_{00}-\sigma_{01}\sigma_{11}^{-1}\sigma_{10} & \sigma_{02}-\sigma_{01}\sigma_{11}^{-1}\sigma_{12}\\
\sigma_{20}-\sigma_{21}\sigma_{11}^{-1}\sigma_{10} & \sigma_{22}-\sigma_{21}\sigma_{11}^{-1}\sigma_{12}%
\end{bmatrix},
    \end{gather}
where the latter block decomposition of $\sigma/\sigma_{11}\in\mathcal{L}(\mathcal{U}\ho \mathcal{J})$ is with respect to the Hilbert space $\mathcal{U}\ho \mathcal{J}$.
\end{proposition}
\begin{proof}
     Assume the hypotheses. Then by Theorem \ref{ThmMainClassicalZProbEffOp} we know that $\sigma_*=\sigma_{00}-\sigma_{01}\sigma_{11}^{-1}\sigma_{10}$ and hence the result follows by block operator multiplication.
\end{proof}

For the purposes of the next corollary and proposition, we make the following observations. Suppose $(\mathcal{H},\mathcal{U},\mathcal{E},\mathcal{J},\sigma)$ is a $Z$-problem with $\sigma$ invertible. Then for the dual $Z$-problem $(\mathcal{H},\mathcal{U},\mathcal{E},\mathcal{J},\sigma^{-1})$, the operator $\sigma^{-1}$ written as a $3\times 3$ block operator matrix with respect to the orthogonal decomposition (\ref{DefDualZProbOrthogonalTripleDecomposition}) is simply
\begin{align}
    \sigma^{-1}=[(\sigma^{-1})_{ij}]_{i,j=0,2,1}=\begin{bmatrix}
        (\sigma^{-1})_{00}&(\sigma^{-1})_{02}&(\sigma^{-1})_{01}\\
        (\sigma^{-1})_{20}&(\sigma^{-1})_{22}&(\sigma^{-1})_{21}\\
        (\sigma^{-1})_{10}&(\sigma^{-1})_{12}&(\sigma^{-1})_{11}
    \end{bmatrix},
\end{align}
where, as in (\ref{DefOfSigmaSubblocks}), we define
\begin{align}
    (\sigma^{-1})_{ij}\in \mathcal{L}(H_j,H_i),\;(\sigma^{-1})_{ij}=\Gamma_i\sigma^{-1}\Gamma_j:H_j\rightarrow H_i,\label{DefOfSigmaInvSubblocks}
\end{align}
for $i,j=0,1,2$ and $\Gamma_0,\Gamma_1,\Gamma_2$ are just the same orthogonal projections of $\mathcal{H}$ onto $H_0=\mathcal{U}, H_1=\mathcal{E}, H_2=\mathcal{J},$ respectively. In particular, by applying Theorem \ref{ThmMainClassicalZProbEffOp} to the dual $Z$-problem $(\mathcal{H},\mathcal{U},\mathcal{J},\mathcal{E},\sigma^{-1})$, we have proven the following corollary.
\begin{corollary}\label{CorDualThmMainClassicalZProbEffOp}
If $(\mathcal{H},\mathcal{U},\mathcal{E},\mathcal{J},\sigma)$ is a $Z$-problem with both $\sigma$ and $(\sigma^{-1})_{22}$ invertible then Theorem \ref{ThmMainClassicalZProbEffOp} is true for the dual $Z$-problem $(\mathcal{H},\mathcal{U},\mathcal{J},\mathcal{E},\sigma^{-1})$ and
\begin{align}
\left(  \sigma^{-1}\right)  _{\ast^{\prime}}&=[(\sigma^{-1})_{ij}]_{i,j=0,2}/(\sigma^{-1})_{22}\\
    &=(\sigma^{-1})_{00}-(\sigma^{-1})_{02}(\sigma^{-1})_{22}^{-1}(\sigma^{-1})_{20}. \label{LemEffDualOpsFundamPropertiesDualEffOpAsSchurComplFormula}%
\end{align}
\end{corollary}

The following proposition on the dual problem and dual effective operator can also be considered a generalization of the notion of matrically coupled $2\times 2$ block operator matrices\cite{90GG, 92BT, 13HR} to $3\times 3$ block operator matrices.
\begin{proposition}
\label{PropEffDualOpsFundamProperties}Suppose $(\mathcal{H},\mathcal{U},\mathcal{E},\mathcal{J},\sigma)$ is a $Z$-problem and $\sigma$ is invertible. Then the following statements are true:

\begin{itemize}
\item[(a)] The operator $\left(  \sigma^{-1}\right)  _{22}$ is invertible if and only if $[\sigma_{ij}]_{i,j=0,1}\in \mathcal{L}(\mathcal{U}\ho\mathcal{E})$ is invertible, in which case
\begin{gather}
\sigma^{-1}/(\sigma^{-1})_{22}=[\sigma_{ij}]_{i,j=0,1}^{-1},\;
\sigma/[\sigma_{ij}]_{i,j=0,1}
=\left(  \sigma^{-1}\right)  _{22}^{-1},\\
\left(  \sigma^{-1}\right)  _{\ast^{\prime}}=\Gamma_{0}[\sigma_{ij}]_{i,j=0,1}
^{-1}\Gamma_{0}|_{\mathcal{U}}=\left(
[\sigma_{ij}]_{i,j=0,1}
^{-1}\right)  _{00},
\end{gather}
where the Schur complement $\sigma/[\sigma_{ij}]_{i,j=0,1}
$ and $\sigma^{-1}/\left(  \sigma^{-1}\right)  _{22}$ of $\sigma$ and $\sigma^{-1}$, respectively,
are with respect to the orthogonal decomposition $\mathcal{H}=(\mathcal{U}\ho\mathcal{E})\ho\mathcal{J}$.
\item[(b)] The operator $\sigma_{11}$ is invertible if and only if $[(\sigma^{-1})_{ij}]_{i,j=0,2}\in\mathcal{L}(\mathcal{U}\ho\mathcal{J})$ is
invertible, in which case
\begin{gather}
\sigma/\sigma_{11}=[(\sigma^{-1})_{ij}]_{i,j=0,2}
^{-1},\text{ }\sigma^{-1}/[(\sigma^{-1})_{ij}]_{i,j=0,2}
=\sigma_{11}^{-1},\\
\sigma_{\ast}=\Gamma_{0}%
[(\sigma^{-1})_{ij}]_{i,j=0,2}^{-1}\Gamma_{0}|_{\mathcal{U}}=\left(
[(\sigma^{-1})_{ij}]_{i,j=0,2}^{-1}\right)  _{00},
\end{gather}
where the Schur complement $\sigma^{-1}/[(\sigma^{-1})_{ij}]_{i,j=0,2}
$ and $\sigma/\sigma_{11}$ of $\sigma^{-1}$ and $\sigma$, respectively, are with respect to the orthogonal decomposition $\mathcal{H}=(\mathcal{U}\ho\mathcal{J})\ho\mathcal{E}$.
\item[(c)] The operators $\sigma_{11}$ and $\sigma_{\ast}$ are invertible if and only if
$\left(  \sigma^{-1}\right)  _{22}$ and $\left(  \sigma^{-1}\right)  _{\ast^{\prime}}$
are invertible if and only if $\sigma_{11}$ and $\left(  \sigma^{-1}\right)  _{22}$ are
invertible, in which case
\begin{equation}
\left(  \sigma^{-1}\right)  _{\ast^{\prime}}=\left(  \sigma_{\ast}\right)  ^{-1}.
\label{LemEffDualOpsFundamPropertiesPartd}%
\end{equation}
\end{itemize}
\end{proposition}
\begin{proof}
Assume the hypotheses. First, with respect to the orthogonal
decomposition $\mathcal{H}=(\mathcal{U}\ho\mathcal{E})\ho \mathcal{J}$, the operators $\sigma,\sigma^{-1}$ are the $2\times 2$ block operator matrices
\begin{gather}
\sigma=%
\begin{bmatrix}%
\begin{bmatrix}
\sigma_{00} & \sigma_{01}\\
\sigma_{10} & \sigma_{11}%
\end{bmatrix}
&
\begin{bmatrix}
\sigma_{02}\\
\sigma_{12}%
\end{bmatrix}
\\%
\begin{bmatrix}
\sigma_{20} & \sigma_{21}%
\end{bmatrix}
& \sigma_{22}%
\end{bmatrix}
,\\
\sigma^{-1}=%
\begin{bmatrix}%
\begin{bmatrix}
\left(  \sigma^{-1}\right)  _{00} & \left(  \sigma^{-1}\right)  _{01}\\
\left(  \sigma^{-1}\right)  _{10} & \left(  \sigma^{-1}\right)  _{11}%
\end{bmatrix}
&
\begin{bmatrix}
\left(  \sigma^{-1}\right)  _{02}\\
\left(  \sigma^{-1}\right)  _{12}%
\end{bmatrix}
\\%
\begin{bmatrix}
\left(  \sigma^{-1}\right)  _{20} & \left(  \sigma^{-1}\right)  _{21}%
\end{bmatrix}
& \left(  \sigma^{-1}\right)  _{22}%
\end{bmatrix}
.
\end{gather}
Then the proof of statement $(a)$ follows immediately from this by Lemma \ref{LemBabachiewiczInversionFormula} and by Proposition \ref{PropAltReprEffOp} applied to the dual $Z$-problem. 

Now part $(b)$ follows immediately from part $(a)$ by
duality, i.e., applying part $(a)$ to $\sigma^{-1}$ instead of $\sigma$ and the dual
decomposition $\mathcal{H=U}\overset{\bot}{\mathcal{\oplus}}\mathcal{J}%
\overset{\bot}{\mathcal{\oplus}}\mathcal{E}$ instead of the decomposition
$\mathcal{H=U}\overset{\bot}{\mathcal{\oplus}}\mathcal{E}\overset{\bot
}{\mathcal{\oplus}}\mathcal{J}$. 

We will now prove part $(c)$. Suppose the
operators $\sigma_{11}$ and $\sigma_{\ast}$ are invertible. Then by Theorem 2 we have $\sigma_*=[\sigma_{ij}]_{i,j=0,1}/\sigma_{11}$ and so by Lemma \ref{LemBabachiewiczInversionFormula}.$(a)$ we have that $[\sigma_{ij}]_{i,j=0,1}$ is invertible. Hence, by statement $(a)$ of this proposition it follows that $(\sigma^{-1})_{22}$ is invertible and so by the formula (\ref{LemEffDualOpsFundamPropertiesDualEffOpAsSchurComplFormula}) for $(\sigma^{-1})_{*^\prime}$ in Corollary \ref{CorDualThmMainClassicalZProbEffOp}, we have $(\sigma^{-1})_{*^\prime}=[(\sigma^{-1})_{ij}]_{i,j=0,2}/(\sigma^{-1})_{22}$. Also, since $\sigma_{11}$ is invertible it follows by statement $(b)$ that $[(\sigma^{-1})_{ij}]_{i,j=0,2}/(\sigma^{-1})_{22}$ is also invertible. This proves that $\left(  \sigma^{-1}\right)  _{22}$ and $\left(  \sigma^{-1}\right)  _{\ast^{\prime}}$
are invertible. Conversely, if $\left(  \sigma^{-1}\right)  _{22}$ and $\left(  \sigma^{-1}\right)  _{\ast^{\prime}}$ are invertible then by duality [i.e., by applying the statement we just proved to the dual $Z$-problem of the $Z$-problem $(\mathcal{H},\mathcal{U},\mathcal{J},\mathcal{E},\sigma^{-1})$, which is just the original $Z$-problem $(\mathcal{H},\mathcal{U},\mathcal{E},\mathcal{J},\sigma)$], it follows immediately that $\sigma_{11}$ and $\sigma_{\ast}$ are invertible.

Now suppose that $\sigma_{11}$ and $\left(  \sigma^{-1}\right)  _{22}$ are both invertible. Then, by Theorem \ref{ThmMainClassicalZProbEffOp}, we have $\sigma_*=[\sigma_{ij}]_{i,j=0,1}/\sigma_{11}$ and, by Corollary \ref{CorDualThmMainClassicalZProbEffOp}, we have $(\sigma^{-1})_{*^\prime}=[(\sigma^{-1})_{ij}]_{i,j=0,2}/(\sigma^{-1})_{22}$. Also, it follows from our hypotheses and by part $(a)$ and $(b)$ of this proposition, that both $[\sigma_{ij}]_{i,j=0,1}, [(\sigma^{-1})_{ij}]_{i,j=0,2}$ are invertible. Hence, by Lemma \ref{LemBabachiewiczInversionFormula}.$(c)$, it follows that both $[\sigma_{ij}]_{i,j=0,1}/\sigma_{11}, [(\sigma^{-1})_{ij}]_{i,j=0,2}/(\sigma^{-1})_{22}$ are invertible. This proves that $\sigma_*, (\sigma^{-1})_{*^\prime}$ are both invertible. Conversely, if $\sigma_{11}$ and $\sigma_*$ are invertible then as we have proven, $\left(  \sigma^{-1}\right)  _{22}$ and $(\sigma^{-1})_{*^\prime}$ are invertible so that, in particular, $\sigma_{11}$ and $\left(  \sigma^{-1}\right)  _{22}$ are both invertible.

To complete the proof, suppose that $\sigma_{11}$ and $\left(  \sigma^{-1}\right)  _{22}$ are both invertible. Then as we have shown, $\sigma_*, (\sigma^{-1})_{*^\prime}$ are both invertible. By part $(b)$ of this statement we know that $[(\sigma^{-1})_{ij}]_{i,j=0,2}$ is invertible and
\begin{align*}
    (\sigma_*)^{-1}=\left(
[(\sigma^{-1})_{ij}]_{i,j=0,2}^{-1}\right)  _{00}^{-1}.
\end{align*}
It now follows from Lemma \ref{LemBabachiewiczInversionFormula}.$(c)$ with the identification $H=\mathcal{U}\ho\mathcal{J}, H_0=\mathcal{U},H_1=\mathcal{J}, A=[(\sigma^{-1})_{ij}]_{i,j=0,2}$ that
\begin{gather*}
    \left(
[(\sigma^{-1})_{ij}]_{i,j=0,2}^{-1}\right)  _{00}^{-1}=[(A^{-1})_{00}]^{-1}=A/A_{11}\\
=[(\sigma^{-1})_{ij}]_{i,j=0,2}/(\sigma^{-1})_{22}.
\end{gather*}
By Corollary \ref{CorDualThmMainClassicalZProbEffOp}, we know that
\begin{align*}
    (\sigma^{-1})_{*^\prime}=[(\sigma^{-1})_{ij}]_{i,j=0,2}/(\sigma^{-1})_{22}.
\end{align*}
Therefore, we have proven that $(\sigma^{-1})_{*^\prime}=(\sigma_*)^{-1}$. This completes the proof.
\end{proof}
\begin{remark}
We give here a simpler proof, using the $Z$-problem directly, that  $(\sigma^{-1})_{*^\prime}=(\sigma_*)^{-1}$ under the hypotheses $\sigma, \sigma_{11}, (\sigma_*)^{-1}$ are all invertible. Let $J_0\in \mathcal{U}$. Then by Theorem \ref{ThmMainClassicalZProbEffOp} there is a unique solution $(J_0,E,J)$ to the direct $Z$-problem $(\mathcal{H},\mathcal{U},\mathcal{E},\mathcal{J},\sigma)$ at $E_0=\sigma_*^{-1}(J_0)$. In particular, $\sigma(E_0+E)=J_0+J$ so that $\sigma^{-1}(J_0+J)=E_0+E$ and hence $(J_0,J,E)$ is a solution to the dual $Z$-problem $(\mathcal{H},\mathcal{U},\mathcal{J},\mathcal{E},\sigma^{-1})$ at $J_0$ and $(\sigma_*)^{-1}(J_0)=E_0$. Now suppose that $(E_0',J',E')$ is a solution to the dual $Z$-problem $(\mathcal{H},\mathcal{U},\mathcal{J},\mathcal{E},\sigma^{-1})$ at $J_0$ so that $\sigma^{-1}(J_0+J')=E_0'+E'$. Then $(J_0,E',J')$ is a solution to the direct $Z$-problem $(\mathcal{H},\mathcal{U},\mathcal{E},\mathcal{J},\sigma)$ at $E_0'$ so that $\sigma_*(E_0')=J_0$ implying $E_0'=(\sigma_*)^{-1}(J_0)$. This proves that the dual effective operator $(\sigma^{-1})_{*^\prime}$ of the dual $Z$-problem $(\mathcal{H},\mathcal{U},\mathcal{J},\mathcal{E},\sigma^{-1})$ exists, is unique, and is given by the formula $(\sigma^{-1})_{*^\prime}=(\sigma_*)^{-1}$. Although this proof is more straightforward, it seems harder to generalize when weaker hypotheses are used.
\end{remark}

\section{\label{sec:RelaxingHyps}Relaxing the hypotheses}

The primary objective of this paper is to study the three quintessential problems \hyperlink{(i)}{(i)}, \hyperlink{(ii)}{(ii)}, and \hyperlink{(iii)}{(iii)} above that arise in every $Z$-problem, but under weaker (or relaxed) hypotheses.

More precisely, we consider the following sets of hypotheses:
\begin{itemize}
	\item[\hypertarget{(H1)}{(H1)}] $\sigma^*=\sigma\geq 0$, $\sigma$ is invertible.
	\item[\hypertarget{(H2)}{(H2)}] $\sigma^*=\sigma$, $\sigma_{11}\geq 0$, $\sigma_{11}$ is invertible.
	\item[\hypertarget{(H3)}{(H3)}] $\begin{bmatrix}
        \sigma_{00}&\sigma_{10}\\
        \sigma_{10}&\sigma_{11}
    \end{bmatrix}^*=\begin{bmatrix}
        \sigma_{00}&\sigma_{10}\\
        \sigma_{10}&\sigma_{11}
    \end{bmatrix}$, $\sigma_{11}\geq 0$, $\ker \sigma_{11}\subseteq \ker \sigma_{01}$.
	\item[\hypertarget{(H4)}{(H4)}]$\dim \mathcal{H}<\infty.$
\end{itemize} 
In particular, we have already consider the three problems \hyperlink{(i)}{(i)-(iii)} under hypotheses \hyperlink{(H1)}{(H1)}  and \hyperlink{(H2)}{(H2)}  and derived from these the classical results in section \ref{sec:ClassicalResults}. In this paper, we are primarily interested in a relaxation of hypotheses \hyperlink{(H1)}{(H1)}  and \hyperlink{(H2)}{(H2)}  with a focus on hypotheses \hyperlink{(H3)}{(H3)}  and \hyperlink{(H4)}{(H4)}.

To understand the relationship between the hypotheses \hyperlink{(H1)}{(H1)-(H3)}, we have the following results.
\begin{lemma}\label{LemFundPosSemiDefImpliesKerRanBlockOpInclusion}
If $H=H_0\ho H_1$ is a Hilbert space, $A=[A_{i,j}]_{i,j=0,1}\in\mathcal{L}\left(H\right)$, and $A^*=A\geq 0$ then
\begin{align}
    \ker A_{11}&\subseteq \ker A_{01},\\
    \overline{\ran A_{10}}&\subseteq \overline{\ran A_{11}}.
\end{align}
\end{lemma}
\begin{proof}
Assume the hypotheses. Let $y\in \ker A_{11}$. Then, 
\begin{gather}
		\left\Vert
		A^{1/2}\begin{bmatrix}
		0\\y
		\end{bmatrix}\right\Vert^2=\left(\begin{bmatrix}
		A_{00}&A_{01}\\
		A_{10}&A_{11}
		\end{bmatrix}\begin{bmatrix}
		0\\y
		\end{bmatrix},\begin{bmatrix}
		0\\y
		\end{bmatrix}\right)=\left(\begin{bmatrix}
		A_{01}y\\A_{11}y
		\end{bmatrix},\begin{bmatrix}
		0\\y
		\end{bmatrix}\right)\\
		=(A_{01}y,0)+(A_{11}y,y)=0.
	\end{gather}
	Hence, follows that
	\begin{gather}
		0=A^{1/2}\left(A^{1/2}\begin{bmatrix}
		0\\y
		\end{bmatrix}\right)=\begin{bmatrix}
		A_{00}&A_{01}\\
		A_{10}&A_{11}
		\end{bmatrix}\begin{bmatrix}
		0\\y
		\end{bmatrix}=\begin{bmatrix}
		A_{01}y\\A_{11}y
		\end{bmatrix}=\begin{bmatrix}
		A_{01}y\\0
		\end{bmatrix},
	\end{gather}
	which implies $A_{01}y=0$. Thus, $y\in\ker A_{01}$, which proves $\ker A_{11}\subseteq \ker A_{01}$. Finally, by the hypotheses along with the basic properties of adjoints and orthogonal complements, it follows
	\begin{gather}
		\overline{\ran A_{10}}=\overline{\ran (A_{01})^*}=(\ker A_{01})^\perp\subseteq(\ker A_{11})^\perp\\
		=\overline{\ran (A_{11})^*}=\overline{\ran A_{11}}.
	\end{gather}
	This completes the proof.
\end{proof}
\begin{corollary}
Consider the hypotheses \hyperlink{(H1)}{(H1)}, \hyperlink{(H2)}{(H2)}, and \hyperlink{(H3)}{(H3)} above. Then
\begin{align}
    \hyperlink{(H1)}{(H1)} \Rightarrow \hyperlink{(H2)}{(H2)} \; \&\; \hyperlink{(H3)}{(H3)},
\end{align}
and, if $\sigma^*=\sigma \geq 0$ then
\begin{align}
    \hyperlink{(H1)}{(H1)} \Rightarrow \hyperlink{(H2)}{(H2)} \Rightarrow \hyperlink{(H3)}{(H3)}.
\end{align}
\end{corollary}
\begin{proof}
We have already proven in Corollary \ref{CorClassicalBasicInvPosPropEffOp} that $\hyperlink{(H1)}{(H1)} \Rightarrow \hyperlink{(H2)}{(H2)} $. Hence, to prove the corollary it suffices to prove that $\hyperlink{(H2)}{(H2)} \Rightarrow \hyperlink{(H3)}{(H3)}$ under the hypothesis $\sigma^*=\sigma \geq 0$. Thus, assume $\sigma^*=\sigma \geq 0$ and the hypotheses \hyperlink{(H2)}{(H2)} . Then the result follows immediately from Lemma \ref{LemFundPosSemiDefImpliesKerRanBlockOpInclusion} with the identifications
\begin{gather}
    H=H_0\ho H_1, H_0=\mathcal{U}, H_1=\mathcal{E},\\ A=[A_{ij}]_{i,j=0,1}=[\sigma_{ij}]_{i,j=0,1}.
\end{gather}
This completes the proof.
\end{proof}

\begin{remark}
Appendix \ref{SectAbsTheoryCompositesVecSpFramework} gives the most compelling reason why hypothesis \hyperlink{(H3)}{(H3)}, specifically, $\ker \sigma_{11}\subseteq \ker \sigma_{01}$, comes into play in this paper when treating problem \hyperlink{(ii)}{(ii)}. It is based on some very general results (Theorem \ref{thm:FundThmExistenceUniquenessSolvabilityZProbEffOpVecSp} and Corollary \ref{cor:KeyResultAppendixComposites}) which are more fitting for the appendix then in this section.
\end{remark}

Our secondary goal of this paper is to give a class of examples for which our results apply and hypotheses \hyperlink{(H4)}{(H4)} is satisfied. This objective is important for two reasons. First, there are currently very few such physical models for which the abstract framework of the $Z$-problem and effective operator applies under hypothesis \hyperlink{(H4)}{(H4)}, whereas there are many (e.g., within the theory of composites) when hypothesis \hyperlink{(H4)}{(H4)} is false such as the above example (Example \ref{ExContinuumPeriodicCondZProb}) of continuum electrical conductivity in periodic media with effective conductivity. Second, the mathematical theory described in this paper is already rich enough when hypothesis \hyperlink{(H4)}{(H4)} is true and becomes much more involved when hypothesis \hyperlink{(H4)}{(H4)} is false. It is mainly for these reasons that we consider the problems we do in Section \ref{sec:DiscreteNetworkExamples}.

\section{\label{sec:VarPrincConstLinearEqsUnifiedFramework}A unified framework for the variational principles}
In this section, we introduce a unified framework for solving constrained linear equations using variational principles which is based on $2\times 2$ block operator matrices that admit a block UDL factorization of a special form (U - upper triangular, D - diagonal, L - lower triangular). This will be use to derive solutions to $Z$-problems and their associated effective operators under weaker hypotheses. It can also be useful in deriving variational principles for other problems in the abstract theory of composites and related problems in matrix/operator theory (for more on this see Ref.\ \onlinecite{22KB}). Hence, we believe this unified framework will have applications outside this paper.

We begin by recalling\cite{03BG, 19FIS} the definition of the Moore-Penrose pseudoinverse and some of it's fundamental properties.
\begin{definition}[Moore-Penrose pseudoinverse]\label{DefMP}
	Let $H,\mathcal{H}$ be two finite-dimensional Hilbert spaces and $X\in\mathcal{L}(H, \mathcal{H})$. Then the Moore-Penrose pseudoinverse $X^+\in\mathcal{L}(\mathcal{H}, H)$ is the unique linear operator that satisfies the four Penrose equations:
		\begin{enumerate}[(1)]
			\item $X^+XX^+=X^+,$
			\item $XX^+X=X,$
			\item $(X^+X)^*=X^+X,$
			\item $(XX^+)^*=XX^+.$
		\end{enumerate}
\end{definition}
\begin{notation}
     For any Hilbert space $H$ with $\dim H<\infty$ and for any subspace $S$ of $H$, we denote the orthogonal projection of $H$ onto $S$ by $\Gamma_S$.
\end{notation}
\begin{lemma}\label{LemMPProp}
    Let $H,\mathcal{H}$ be two finite-dimensional Hilbert spaces and $X\in\mathcal{L}(H, \mathcal{H})$. Then
		\begin{enumerate}[(1)]
			\item $(X^+)^*=(X^*)^+$,
			\item $XX^+=\Gamma_{\ran X},\;I-XX^+=\Gamma_{\ker X^*}$,
			\item
			      $X^+X=\Gamma_{\ran X^*},\;I-X^+X=\Gamma_{\ker X}$,
			\item If $X=X^*$ then $X^+X=XX^+$.
			\item If $X$ is invertible then $X^{-1}=X^+$.
			\item $(X^+)^+=X.$
		\end{enumerate}
\end{lemma}

Next, we consider the solutions of a class of constrained linear equations involving an operator having a special block factorization and then describe a fundamental variational principle for solving the equations.
\begin{lemma}\label{LemGenConstEq}
	Let $H=H_0\ho H_1$ be a Hilbert space with $\dim H<\infty$.
	If $X\in\mathcal{L}(H)$ has the $2\times 2$ block factorization
	\begin{gather}\label{LemL2GenMinVarPrincBlockXDef}
		X=\begin{bmatrix}
		I_{H_0} & Y^*\\
		0 & I_{H_1}
		\end{bmatrix}\begin{bmatrix}
		W & 0\\
		0 & Z
		\end{bmatrix}\begin{bmatrix}
		I_{H_0} & 0\\
		Y & I_{H_1}
		\end{bmatrix},\\
		W\in \mathcal{L}(H_0),\; Y\in \mathcal{L}(H_0,H_1),\; Z\in\mathcal{L}(H_1),\label{ConstrOne}\\
		W^*=W\text{ and } Z^*=Z\label{ConstrTwo}
	\end{gather}
	then, for each $u\in H_0$, the set of all solutions of the constrained linear equation
	\begin{equation}
		X(u+v)=w,\;\;(w,v)\in H_0\times H_1,\label{LabelConstrainedEqsForMatrixA}
	\end{equation}
	is given by
	\begin{align}
			\{(w,v)\in H_0\times H_1:w=Wu,\;v=-Z^+ZYu+v_1,\;v_1\in \ker Z\}. 
	\end{align}
\end{lemma}
\begin{proof}
	Assume the hypotheses. Then, for each $u\in H_0$,
	\begin{gather*}
		X(u+v)=w,\;\;(w,v)\in H_0\times H_1 \\
		\Leftrightarrow
		\begin{bmatrix}
			W & 0 \\
			0 & Z 
		\end{bmatrix}\begin{bmatrix}
		I_{H_0} & 0\\
		Y & I_{H_1}
		\end{bmatrix}\begin{bmatrix}
		u \\ v
		\end{bmatrix}=\begin{bmatrix}
		I_{H_0} & -Y^*\\
		0 & I_{H_1}
		\end{bmatrix}\begin{bmatrix}
		w \\ 0
		\end{bmatrix}\\
		\Leftrightarrow
		\begin{bmatrix}
			Wu \\ Z(Yu+v)
		\end{bmatrix}=\begin{bmatrix}
		w \\ 0
		\end{bmatrix}
		\Leftrightarrow
		w=Wu,\;\;-ZYu=Zv\\
		\underset{\text{Lemma}\;\ref{LemMPProp}.(3)}{\Leftrightarrow}
		v=v_0+v_1,\;\;v_1\in\ker Z,\;\;v_0=-Z^+ZYu,\;\;w=Wu.
	\end{gather*}
	This completes the proof.
\end{proof}

\begin{theorem}[Fundamental minimization principle]\label{ThmGenMinVarPrinc}
	Let $H=H_0\ho H_1$ be a Hilbert space with $\dim H<\infty$.
	If $X\in\mathcal{L}(H)$ has the $2\times 2$ block factorization (\ref{LemL2GenMinVarPrincBlockXDef})-(\ref{ConstrTwo}) and $Z\geq 0$,
	then $W$ is the unique self-adjoint operator satisfying the minimization principle
	\begin{equation}
		\left(u,Wu\right)=\min_{v\in H_1}\left( u+v, X(u+v) \right), \text{ for each } u\in H_0.
	\end{equation}
	Furthermore, for each $u\in H_0$, the set of minimizers is given by 
	\begin{eqnarray}
	\{v\in H_1:v=-(Z^+ZY)u+v_1,\;v_1\in \ker{Z}\}.
	\end{eqnarray}
\end{theorem}
\begin{proof}  
	Assume the hypotheses. Then, for any $(u,v)\in H_0\times H_1$, it follows that
	\begin{gather*}
		( u+v, X(u+v) ) = 
		\left(\begin{bmatrix}
		u\\v
		\end{bmatrix},
		\begin{bmatrix}
			I_{H_0} & Y^*     \\
			0       & I_{H_1} 
		\end{bmatrix}
		\begin{bmatrix}
			W & 0 \\
			0 & Z 
		\end{bmatrix}
		\begin{bmatrix}
			I_{H_0} & 0       \\
			Y       & I_{H_1} 
		\end{bmatrix}
		\begin{bmatrix}
			u \\v
		\end{bmatrix}\right)\\
		= \left(\begin{bmatrix}
		I_{H_0} & Y^*\\
		0 & I_{H_1}
		\end{bmatrix}^*
		\begin{bmatrix}
			u \\v
		\end{bmatrix},
		\begin{bmatrix}
			W & 0 \\
			0 & Z 
		\end{bmatrix}
		\begin{bmatrix}
			I_{H_0} & 0       \\
			Y       & I_{H_1} 
		\end{bmatrix}
		\begin{bmatrix}
			u \\v
		\end{bmatrix}\right)\\
		= \left( \begin{bmatrix}
		I_{H_0} & 0\\
		Y & I_{H_1}
		\end{bmatrix}\begin{bmatrix}u\\v\end{bmatrix},\begin{bmatrix}
		W & 0\\
		ZY & Z
		\end{bmatrix}\begin{bmatrix}u\\v\end{bmatrix}\right)\\
		=\left( \begin{bmatrix}
		u\\
		Yu+v
		\end{bmatrix},\begin{bmatrix}
		W(u)\\
		Z(Yu+v)
		\end{bmatrix}\right)\\
		=(u,Wu)+(Yu+v,Z(Yu+v))\geq (u,Wu),
	\end{gather*}
	with equality if and only if $(Yu+v,Z(Yu+v))=0$. By hypothesis $Z^*=Z\geq 0$ and we have
	\begin{gather*}
		Z(Yu+v)=0\Leftrightarrow Z(Yu)+Zv=0\\
		\Leftrightarrow Zv=-Z(Yu)
		\Leftrightarrow Z^+Zv=-Z^+ZYu\\
		\Leftrightarrow v=v_0+v_1,\;v_1\in\ker Z, v_0=-Z^+ZYu.
	\end{gather*}
	The proof of the theorem now follows from this and the fact that $W\in\mathcal{L}\left(H_{0}\right)  $ is
self-adjoint so is uniquely determined by its quadratic form\cite{80JW}.
\end{proof}

The preceding results will be used in conjunction with the following to derive the solutions to the $Z$-problem and its effective operator under weaker hypotheses.

\begin{definition}[Generalized Schur complement]
    Let $H=H_0\ho H_1$ be a Hilbert space with $\dim H<\infty$ and $X=[X_{ij}]_{i,j=0,1}\in\mathcal{L}(H)$. The generalized Schur complement (gsc) of $X$ with respect to $X_{11}$ is
	\begin{equation}
		X/X_{11}=X_{00}-X_{01}X_{11}^{+}X_{10}.
	\end{equation}
\end{definition}

\begin{proposition}[Aitken block-diagonalization conditions]\label{PropgscFact}
	Let $H=H_0\ho H_1$ be a Hilbert space with $\dim H<\infty$. If $X=[X_{ij}]_{i,j=0,1}\in\mathcal{L}(H)$ and $X^*=X$ then
	\begin{gather}\label{gscfact1}
		X=\begin{bmatrix}
		I_{H_0} & (X_{11}^+X_{10})^*\\
		0 & I_{H_1}
		\end{bmatrix}\begin{bmatrix}
		X/X_{11} & 0\\
		0 & X_{11}
		\end{bmatrix}\begin{bmatrix}
		I_{H_0} & 0\\
		X_{11}^+X_{10} & I_{H_1}
		\end{bmatrix},
	\end{gather}
	if and only if \vspace{-0.5em}
	\begin{gather}\label{gscfact2}
		\ker X_{11}\subseteq \ker X_{01}\; (\text{i.e., } 
		\ran X_{10}\subseteq \ran X_{11}).
	\end{gather}
	In particular, if $X\in\mathcal{L}(H)$ and $X^*=X\geq 0$ then (\ref{gscfact1}) and (\ref{gscfact2}) are true.
\end{proposition}
\begin{proof}
	Assume the hypotheses. Then
	\begin{gather}
		\begin{bmatrix}
			I_{H_0} & (X_{11}^+X_{10})^* \\
			0                 & I_{H_1}  
		\end{bmatrix}
		\begin{bmatrix}
			X/X_{11} & 0      \\
			0        & X_{11} 
		\end{bmatrix}
		\begin{bmatrix}
			I_{H_0} & 0                 \\
			X_{11}^+X_{10}    & I_{H_1} 
		\end{bmatrix}\\=
		\begin{bmatrix}
			I_{H_0} & (X_{11}^+X_{10})^* \\
			0                 & I_{H_1}  
		\end{bmatrix}
		\begin{bmatrix}
			X/X_{11}             & 0      \\
			X_{11}X_{11}^+X_{10} & X_{11} 
		\end{bmatrix}\\=
		\begin{bmatrix}
			X/X_{11}+(X_{11}X_{10})^*X_{11}X_{11}^+X_{10} & (X_{11}^+X_{10})^*X_{11} \\
			X_{11}X_{11}^+X_{10}                          & X_{11}                   
		\end{bmatrix}\\=
		\begin{bmatrix}
			X_{00}-X_{01}X_{11}^+X_{10}+X_{01}X_{11}^+X_{11}X_{11}^+X_{10} & X_{01}X_{11}^+X_{11} \\
			X_{11}X_{11}^+X_{10}                                           & X_{11}               
		\end{bmatrix}\\=
		\begin{bmatrix}
			X_{00}               & X_{01}X_{11}^+X_{11} \\
			X_{11}X_{11}^+X_{10} & X_{11}               
		\end{bmatrix},
	\end{gather}
	which is equal to $X$ if and only if $X_{01}X_{11}^+X_{11}=X_{01}$ and $X_{11}X_{11}^+X_{10}=X_{10}$.
	Hence, by Lemma \ref{LemMPProp}, we have equality if and only if $\ker X_{01}\subseteq \ker X_{11}$ or, equivalently, $\ran X_{11}\subseteq \ran X_{01}$. In particular, if $X\in\mathcal{L}(H)$ and $X^*=X\geq 0$ then (\ref{gscfact2}) is true by Lemma \ref{LemFundPosSemiDefImpliesKerRanBlockOpInclusion} and hence, from our proof just given, it follows (\ref{gscfact1}) is also true. This completes the proof.
\end{proof}

\begin{theorem}[Gen.\ Schur complement: min.\ principle]\label{ThmgscMinPrinc}
	Let $H=H_0\ho H_1$ be a Hilbert space with $\dim(H)<\infty$. If $X=[X_{ij}]_{i,j=0,1}\in\mathcal{L}(H)$, $X^*=X$, $X_{11}\geq 0$, and $\ker X_{11}\subseteq \ker X_{01}$ then $X/X_{11}$ is the unique self-adjoint operator satisfying the minimization principle
	\begin{equation}
		\left( u,(X/X_{11})u\right)=\min_{v\in H_1}\left( u+v, X(u+v) \right), \text{ for each } u\in H_0. \label{ThmgscMinPrincPart1}
	\end{equation}
	Furthermore, for each $u\in H_0$, the set of minimizers is given by 
	\begin{eqnarray}
\{v\in H_1:v=-(X_{11})^+X_{10}x+v_1,\; v_1\in \operatorname{Ker}{X_{11}}\}.\label{ThmgscMinPrincPart2}
	\end{eqnarray}
	Moreover,
	\begin{align}
	    X/X_{11}\leq X_{00}.\label{ThmgscMinPrincPart3}
	\end{align}
	In particular, if $X\in\mathcal{L}(H)$ and $X^*=X\geq 0$ then the above statements are true and
	\begin{equation}
		0\leq X/X_{11}\leq X_{00}.\label{ThmgscMinPrincPart4}
	\end{equation}
\end{theorem}
\begin{proof}
	Assume $H=H_0\ho H_1$ is a Hilbert space with $\dim(H)<\infty$. Suppose $X=[X_{ij}]_{i,j=0,1}\in\mathcal{L}(H)$, $X^*=X$, $X_{11}\geq 0$, and $\ker X_{11}\subseteq \ker X_{01}$. Then, by Proposition \ref{PropgscFact}, $X$ has the block factorization (\ref{gscfact1}) and the proof of the minimization variational principle follows immediately by Theorem \ref{ThmGenMinVarPrinc}, property $(1)$ in Def.\ \ref{DefMP} of the Moore-Penrose pseudoinverse, together with the identifications
	\begin{gather}
		W=X/X_{11},\; Z=X_{11},\; Y=X_{11}^+X_{10}.
	\end{gather}
	The upper bound (\ref{ThmgscMinPrincPart3}) follows from considering the RHS of (\ref{ThmgscMinPrincPart1}) with $v=0\in H_1$. The remaining part of the theorem is proved by Proposition \ref{PropgscFact} since if $X\in \mathcal{L}(H)$, $X^*=X\geq 0$ then (\ref{gscfact1}) and (\ref{gscfact2}) are true and hence our statement is true from the proof just given, and the lower bound (\ref{ThmgscMinPrincPart4}) follows immediately from (\ref{ThmgscMinPrincPart1}).
\end{proof}

\section{\label{sec:MainResults}Main results}

\subsection{\label{sec:MainResults:subsec:DirectProb}Direct \textit{Z}-problem}
The following theorem answers the questions \hyperlink{(i)}{(i)} and \hyperlink{(ii)}{(ii)} from Section \ref{sec:intro:subsec:ZprobEffOp}, but now under weaker hypotheses for which $\sigma_{11}$ need not be invertible.
\begin{theorem}\label{ThmZProbSol}
	Suppose $(\mathcal{H},\mathcal{U},\mathcal{E},\mathcal{J}, \sigma)$ is a $Z$-problem (as defined in Def.\ \ref{DefZProbMain}) with $\dim \mathcal{H}<\infty$. If 
	\begin{gather}
		\begin{bmatrix}
		\sigma_{00} & \sigma_{01}\\
		\sigma_{10} & \sigma_{11}
		\end{bmatrix}=\begin{bmatrix}
		\sigma_{00} & \sigma_{01}\\
		\sigma_{10} & \sigma_{11}
		\end{bmatrix}^*,\label{ThmZProbSolSelfAdjHyp}\\
		\ker \sigma_{11}\subseteq \ker \sigma_{01},
	\end{gather}
	then, for each $E_0\in \mathcal{U}$, the solutions of the $Z$-problem at $E_0$ (as defined in Def.\ \ref{DefZProbMain}) are given by the following formulas:
	\begin{gather}
	   J_0=\sigma_*E_0,\label{ThmZProbSolSelfAdjHypPart1}\\
	   E=-\sigma_{11}^+\sigma_{10}E_0+K,\;K\in\ker \sigma_{11},\label{ThmZProbSolSelfAdjHypPart2}\\
	   J=\sigma_{20}E_0+\sigma_{21}E,\label{ThmZProbSolSelfAdjHypPart3}\\
		\sigma_*=\begin{bmatrix}
		\sigma_{00} & \sigma_{01}\\
		\sigma_{10} & \sigma_{11}
		\end{bmatrix}/\sigma_{11}=\sigma_{00}-\sigma_{01}\sigma_{11}^+\sigma_{10}.\label{GenEffOpSchurCompFormula}
	\end{gather}
    Moreover, the effective operator $\sigma_*\in\mathcal{L}(\mathcal{U})$ of the $Z$-problem $(\mathcal{H},\mathcal{U},\mathcal{E},\mathcal{J}, \sigma)$ exists, is unique, and is given by the generalized Schur complement formula (\ref{GenEffOpSchurCompFormula}).
\end{theorem}
\begin{proof}
	Assume the hypotheses. Let $E_0\in \mathcal{U}$. The proof follows immediately using the equivalent form of the $Z$-problem (\ref{ZProbEquivFormPart1})--(\ref{ZProbEquivFormPart3}) and then applying Lemma \ref{LemGenConstEq} and Proposition \ref{PropgscFact} to the constrained linear equations (\ref{ZProbEquivFormPart1}) and (\ref{ZProbEquivFormPart2}) by making the identifications:
	\begin{gather}
		H_0=\mathcal{U},\;H_1=\mathcal{E}, X=[X_{ij}]_{i,j=0,1}=[\sigma_{ij}]_{i,j=0,1},\\
		X(u+v)=w,\;u=E_0,\;v=E,\;w=J_0.
	\end{gather}
	This completes the proof.
\end{proof}

Equipped with this theorem, we can now fully answer the questions \hyperlink{(i)}{(i)} and \hyperlink{(ii)}{(ii)} from Section \ref{sec:intro:subsec:ZprobEffOp} under the self-adjoint hypothesis (\ref{ThmZProbSolSelfAdjHyp}) and finite-dimensional hypothesis \hyperlink{(H4)}{(H4)}, i.e., $\dim \mathcal{H}<\infty$.

\begin{theorem}\label{ThmFullAnswersToiAndiiForSelfAdjFiniteDimHyps}
Suppose $(\mathcal{H},\mathcal{U},\mathcal{E},\mathcal{J}, \sigma)$ is a $Z$-problem (as defined in Def.\ \ref{DefZProbMain}) with $\dim \mathcal{H}<\infty$ and 
	\begin{gather}
		\begin{bmatrix}
		\sigma_{00} & \sigma_{01}\\
		\sigma_{10} & \sigma_{11}
		\end{bmatrix}=\begin{bmatrix}
		\sigma_{00} & \sigma_{01}\\
		\sigma_{10} & \sigma_{11}
		\end{bmatrix}^*.
	\end{gather}
Then the following statements are true:
\begin{itemize}
    \item[(a)] A necessary and sufficient condition for the $Z$-problem to have a unique solution for some $E_0\in \mathcal{U}$ is
\begin{align}
    \ker \sigma_{11}=\{0\},\label{ThmFullAnswersToiAndiiForSelfAdjFiniteDimHyps2ndHyp}
\end{align}
in which case $\sigma_{11}$ is invertible and Theorem \ref{ThmMainClassicalZProbEffOp} is true.
    \item[(b)] A necessary and sufficient condition for an effective operator $\sigma_*$ of the $Z$-problem $(\mathcal{H},\mathcal{U},\mathcal{E},\mathcal{J}, \sigma)$ (as defined in Def.\ \ref{DefZProbMain}) to exists is
\begin{align}
    \ker \sigma_{11}\subseteq \ker \sigma_{01},\label{ThmFullAnswersToiAndiiForSelfAdjFiniteDimHypsMainHyp}
\end{align}
    in which case Theorem \ref{ThmZProbSol} is true.
\end{itemize}
\end{theorem}
\begin{proof}
     Assume the hypotheses. We begin by proving $(b)$. If (\ref{ThmFullAnswersToiAndiiForSelfAdjFiniteDimHypsMainHyp}) is true then Theorem \ref{ThmZProbSol} is true and, in particular, an effective operator $\sigma_*$ of the $Z$-problem $(\mathcal{H},\mathcal{U},\mathcal{E},\mathcal{J}, \sigma)$ (as defined in Def.\ \ref{DefZProbMain}) exists. Conversely, if an effective operator of the $Z$-problem $(\mathcal{H},\mathcal{U},\mathcal{E},\mathcal{J}, \sigma)$ (as defined in Def.\ \ref{DefZProbMain}) exists then it follows by Theorem \ref{thm:FundThmExistenceUniquenessSolvabilityZProbEffOpVecSp} that (\ref{ThmFullAnswersToiAndiiForSelfAdjFiniteDimHypsMainHyp}) is true and hence Theorem \ref{ThmZProbSol} is true. This proves statement $(b)$.
     
     We now prove statement $(a)$. First, suppose (\ref{ThmFullAnswersToiAndiiForSelfAdjFiniteDimHyps2ndHyp}) is true. Then, since $\dim\mathcal{H}<\infty$, it follows that $\sigma_{11}$ is invertible, in which case Theorem \ref{ThmZProbSol} is true so that, in particular, the $Z$-problem has a unique solution for every $E_0\in \mathcal{U}$ and hence at least one by taking $E_0=0$. Conversely, suppose that the $Z$-problem has a unique solution for some $E_0\in \mathcal{U}$. Then $\sigma(E_0+E)=J_0+J$ for some $(J_0,E,J)\in \mathcal{U}\times \mathcal{E}\times \mathcal{J}$. Let $(J_0',E',J')\in \mathcal{U}\times \mathcal{E}\times \mathcal{J}$ be a solution of the $Z$-problem at $0$ [we know that at least one exists, namely, $(0,0,0)$] so that $\sigma(E')=J_0'+J'$. Then by linearity of $\sigma$ it follows that $\sigma(E_0+E-E')=J_0-J_0'+J-J'$ and hence $(J_0-J_0',E-E',J-J')$ is a solution of the $Z$-problem at $E_0$ implying by hypothesis $(J_0-J_0',E-E',J-J')=(J_0,E,J),$ that is, $J_0'=0, E'=0, J'=0$. This proves that the $Z$-problem has a unique solution at $0$. It now follows from Theorem \ref{thm:FundThmExistenceUniquenessSolvabilityZProbEffOpVecSp} that (\ref{ThmFullAnswersToiAndiiForSelfAdjFiniteDimHypsMainHyp}) is true and hence Theorem \ref{ThmZProbSol} is true. Therefore, by this theorem and the hypothesis that the $Z$-problem has a unique solution for every $E_0\in \mathcal{U}$, it follows that $\ker\sigma_{11}=\{0\}$. This completes the proof.
\end{proof}

The generalized Schur complement formula (\ref{GenEffOpSchurCompFormula}) for the effective operator $\sigma_*$ is useful for deriving some of it's basic properties as we shall now see.
\begin{corollary}\label{CorGenClassicalBasicPropEffOp}
If $(\mathcal{H},\mathcal{U},\mathcal{E},\mathcal{J}, \sigma)$ is a $Z$-problem satisfying the hypotheses in Theorem \ref{ThmZProbSol} then 
\begin{align}
    (\sigma_*)^*=\sigma_*.
\end{align}
\end{corollary}
\begin{proof}
    Assume the hypotheses. Then the hypotheses are also true for the $Z$-problem $(\mathcal{H},\mathcal{U},\mathcal{E},\mathcal{J}, \sigma^*)$, in particular, $\dim \mathcal{H}<\infty$ and
	\begin{gather*}
		\begin{bmatrix}
		(\sigma^*)_{00} & (\sigma^*)_{01}\\
		(\sigma^*)_{10} & (\sigma^*)_{11}
		\end{bmatrix}=	\begin{bmatrix}
		(\sigma_{00})^* & (\sigma_{10})^*\\
		(\sigma_{01})^* & (\sigma_{11})^*
		\end{bmatrix}\\
		=\begin{bmatrix}
		\sigma_{00} & \sigma_{01}\\
		\sigma_{10} & \sigma_{11}
		\end{bmatrix}^*
		=\begin{bmatrix}
		\sigma_{00} & \sigma_{01}\\
		\sigma_{10} & \sigma_{11}
		\end{bmatrix},
	\end{gather*}
    which implies
    \begin{gather*}
        \begin{bmatrix}
		(\sigma^*)_{00} & (\sigma^*)_{01}\\
		(\sigma^*)_{10} & (\sigma^*)_{11}
		\end{bmatrix}=\begin{bmatrix}
		(\sigma^*)_{00} & (\sigma^*)_{01}\\
		(\sigma^*)_{10} & (\sigma^*)_{11}
		\end{bmatrix}^*=\begin{bmatrix}
		\sigma_{00} & \sigma_{01}\\
		\sigma_{10} & \sigma_{11}
		\end{bmatrix},\\
		\ker (\sigma^*)_{11}=\ker \sigma_{11}\subseteq \ker \sigma_{01}=\ker (\sigma^*)_{01}.
    \end{gather*}
    Thus, it follows from this and Theorem \ref{ThmZProbSol} that
    \begin{gather*}
        (\sigma^*)_*=(\sigma^*)_{00}-(\sigma^*)_{01}(\sigma^*)_{00}^{-1}(\sigma^*)_{10}\\
        =\sigma_{00}-\sigma_{01}\sigma_{00}^{-1}\sigma_{10}=\sigma_*.
    \end{gather*}
This proves the corollary.
\end{proof}

Now the next theorem and corollary answers questions \hyperlink{(iii)}{(iii)}.$(2)$ under weaker hypotheses then Theorem \ref{ThmClassicalDiriMinPrin}.
\begin{theorem}[Generalized Dirichlet minimization principle]\label{ThmGenClassicalDiriMinPrin}
If $(\mathcal{H},\mathcal{U},\mathcal{E},\mathcal{J}, \sigma)$ is a $Z$-problem satisfying the hypotheses in Theorem \ref{ThmZProbSol} and $\sigma_{11}\geq 0$ then the effective operator $\sigma_*$ is unique self-adjoint operator satisfying the minimization principle:
\begin{equation}
		( E_0,\sigma_*E_0 )=\min_{E\in\mathcal{E}}( E_0+E,\sigma(E_0+E) ),\;\forall E_0\in\mathcal{U}.\label{GenClassicalDirMinPrincEffOp}
	\end{equation}
	Furthermore, for each $E_0\in \mathcal{U}$, the set of minimizers is given by 
	\begin{eqnarray}
			\{E\in \mathcal{E}:E =-(\sigma_{11})^+\sigma_{10}E_0+K,\;K\in \operatorname{Ker}{\sigma_{11}}\}.\label{GenClassicalDirMinPrincEffOpPart2}
	\end{eqnarray}
	Moreover, we have the following upper bound on the effective operator:
	\begin{align}
	    \sigma_*\leq \sigma_{00}.\label{GenClassicalDirMinPrincEffOpPart3}
	\end{align}
\end{theorem}
\begin{proof}
Assume the hypotheses. Then the proof of the theorem follows immediately from Theorem \ref{ThmgscMinPrinc} by making the identifications:
	\begin{gather*}
		H_0=\mathcal{U},\;H_1=\mathcal{E},\; X=[X_{ij}]_{i,j=0,1}=[\sigma_{ij}]_{i,j=0,1},\\
	u=E_0,\;v=E.
	\end{gather*}
	This completes the proof.
\end{proof}

The following corollary provides simple sufficient conditions for Theorem \ref{ThmZProbSol}, Corollary \ref{CorGenClassicalBasicPropEffOp}, and Theorem \ref{ThmGenClassicalDiriMinPrin} to be true.
\begin{corollary}\label{CorSuffCondForTrueThmGenClassicalDiriMinPrin}
If $(\mathcal{H},\mathcal{U},\mathcal{E},\mathcal{J}, \sigma)$ is a $Z$-problem with $\dim \mathcal{H}<\infty$ and either
\begin{gather}
    \sigma^*=\sigma\geq 0\label{WeakSuffCondForTrueThmGenClassicalDiriMinPrin}
\end{gather}
or
\begin{gather}
		\begin{bmatrix}
		\sigma_{00} & \sigma_{01}\\
		\sigma_{10} & \sigma_{11}
		\end{bmatrix}^*=\begin{bmatrix}
		\sigma_{00} & \sigma_{01}\\
		\sigma_{10} & \sigma_{11}
		\end{bmatrix}\geq 0\label{WeakerSuffCondForTrueThmGenClassicalDiriMinPrin}
	\end{gather}
then Theorem \ref{ThmZProbSol}, Corollary \ref{CorGenClassicalBasicPropEffOp}, and Theorem \ref{ThmGenClassicalDiriMinPrin} are true. In particular, we have the following bounds on the effective operator:
\begin{align}
    0\leq \sigma_*\leq \sigma_{00}.\label{CorSuffCondForTrueThmGenClassicalDiriMinPrinBounds}
\end{align}
\end{corollary}
\begin{proof}
The proof follows immediately from Proposition \ref{PropgscFact} and Theorem \ref{ThmgscMinPrinc}.
\end{proof}

\subsection{\label{sec:MainResults:subsec:DualProb}Dual \textit{Z}-problem}

\begin{definition}[Generalized dual $Z$-problem]\label{DefGenDualZProb}Given a $Z$-problem $(\mathcal{H},\mathcal{U},\mathcal{E},\mathcal{J},\sigma)$ with $\dim \mathcal{H}<\infty$ (the direct $Z$-problem), the corresponding generalized dual $Z$-problem is the $Z$-problem $(\mathcal{H},\mathcal{U},\mathcal{J},\mathcal{E},\sigma^+)$
associated with the orthogonal decomposition
\begin{equation}
\mathcal{H=U}\overset{\bot}{\mathcal{\oplus}}\mathcal{J}\overset{\bot
}{\mathcal{\oplus}}\mathcal{E}
\label{DefGenDualZProbOrthogonalTripleDecomposition}%
\end{equation}
and operator $\sigma^+\in\mathcal{L}\left(  \mathcal{H}\right)$, i.e., the problem: given $J_{0}\in\mathcal{U}$, find
vectors $\left(  E_{0},J,E\right)  \in\mathcal{U}\times\mathcal{J}%
\times\mathcal{E}$ satisfying%
\begin{equation}
E_{0}+E=\sigma^+\left(  J_{0}+J\right)  . \label{DefGenDualZProbEq}%
\end{equation}
An effective operator for this $Z$-problem will be denoted by $\left(
\sigma^+\right)_{\ast^{\prime}}  $. In other words, if there exists an operator $\left(
\sigma^+\right)  _{\ast^{\prime}}\in\mathcal{L}\left(  \mathcal{U}\right)  $
such that
\begin{equation}
E_{0}=\left(  \sigma^+\right)  _{\ast^{\prime}}J_{0}
\label{DefGenDualZProbDefEffOp}%
\end{equation}
whenever $\left(  E_{0},J,E\right)  $ is a solution of the generalized dual $Z$-problem at
$J_{0}$ then $\left(  \sigma^+\right)  _{\ast^{\prime}}$ is an effective
operator for the generalized dual $Z$-problem.
\end{definition}

The following result will be needed which is the analogy of Lemma \ref{LemBabachiewiczInversionFormula}.
\begin{lemma}[Generalized Babachiewicz inverse formula]\label{LemGenBabaInvFormula}
	Let $H=H_0\ho H_1$ be a Hilbert space with $\dim H<\infty$. If $X=[X_{ij}]_{i,j=0,1}\in\mathcal{L}(H)$ and $X^*=X$ then
    \begin{align}
        \hspace{-3.5em}X^{+}=\begin{bmatrix}
    (X/X_{11})^+& -(X/X_{11})^+X_{01}(X_{11})^+\\
     -(X_{11})^+X_{10} (X/X_{11})^+ &  (X_{11})^++(X_{11})^+X_{10}(X/X_{11})^+X_{01}(X_{11})^+
\end{bmatrix}\label{LemGenBabaInvFormulaPart1}
    \end{align}
if and only if
\begin{align}
    \ker X_{11}\subseteq\ker X_{01},\;\;\ker(X/X_{11})\subseteq\ker X_{10},\label{LemGenBabaInvFormulaPart2}
\end{align}
in which case we have
\begin{align}
    (X^{+})_{00}=(X/X_{11})^+,\;\;X/X_{11}=[(X^{+})_{00}]^{+}.\label{LemGenBabaInvFormulaPart3}
\end{align}
\end{lemma}
\begin{proof}
     Although this lemma is known [see, for instance, Theorem 1 and Corollary 3 in Ref.\ \onlinecite{74BC} and eqs.\ (1.5), (1.6')-(1.9'), and (2.6) within], we will give a short and self-contained proof here. Assume the hypotheses. First, the equalities (\ref{LemGenBabaInvFormulaPart3}) will follow immediately from the block form (\ref{LemGenBabaInvFormulaPart3}) for $X^+$ and property $(6)$ in Lemma \ref{LemMPProp}. Hence, we need only prove that (\ref{LemGenBabaInvFormulaPart1}) and (\ref{LemGenBabaInvFormulaPart2}) are equivalent. To do that define $Y, Z\in \mathcal{L}(H_0\ho H_1)$ by the RHS of (\ref{LemGenBabaInvFormulaPart1}) and the RHS of (\ref{gscfact1}), respectively. Then
     \begin{gather}
         Y=\begin{bmatrix}
		(X/X_{11})^+ & -(X/X_{11})^+X_{01}X_{11}^+\\
		-X_{11}^+X_{10}(X/X_{11})^+ & (X_{11})^++X_{11}^+X_{10}(X/X_{11})^+X_{01}X_{11}^+
		\end{bmatrix}\\
		=\begin{bmatrix}
		I_{H_0} & 0\\
		-X_{11}^+X_{10} & I_{H_1}
		\end{bmatrix}\begin{bmatrix}
		(X/X_{11})^+ & 0\\
		0 & (X_{11})^+
		\end{bmatrix}\begin{bmatrix}
		I_{H_0} & (-X_{11}^+X_{10})^*\\
		0 & I_{H_1}
		\end{bmatrix},\\
		Z=\begin{bmatrix}
		I_{H_0} & (X_{11}^+X_{10})^*\\
		0 & I_{H_1}
		\end{bmatrix}\begin{bmatrix}
		X/X_{11} & 0\\
		0 & X_{11}
		\end{bmatrix}\begin{bmatrix}
		I_{H_0} & 0\\
		X_{11}^+X_{10} & I_{H_1}
		\end{bmatrix}.
     \end{gather}
     Then from those block factorizations it follows immediately by block multiplication and Lemma \ref{LemMPProp} that
     \begin{gather}
         YZY=Y,\; ZYZ=Z,\; ZY=\begin{bmatrix}
		\Gamma_{\ran X/X_{11}} & \Gamma_{\ker X/X_{11}}X_{01}X_{11}^+\\
		0 & \Gamma_{\ran X_{11}}\label{LemGenBabaInvFormulaIffCondYIsZPlus}
		\end{bmatrix}.
     \end{gather}
     Now by Definition \ref{DefMP} and the fact that $Y^*=Y, Z^*=Z$, it follows that a necessary and sufficient condition for $Z^+=Y$ is $(ZY)^*=ZY$. From the block form $ZY$ in (\ref{LemGenBabaInvFormulaIffCondYIsZPlus}), it follows that that a necessary and sufficient condition for $Z^+=Y$ is
    \begin{gather}
	    \Gamma_{\ker X/X_{11}}X_{01}X_{11}^+=0,
	\end{gather}
	which is equivalent to
	\begin{gather}
	    \ker X/X_{11} \subseteq \ker X_{11}^+X_{10}.\label{LemGenBabaInvFormulaIffAltCondYIsZPlus}
	\end{gather}
	We are now ready to prove that (\ref{LemGenBabaInvFormulaPart1}) and (\ref{LemGenBabaInvFormulaPart2}) are equivalent.  
	
	Suppose that (\ref{LemGenBabaInvFormulaPart2}) is true. Then $Z=X$ by Proposition \ref{PropgscFact}. Hence to prove $X^+=Y$ it suffices to prove (\ref{LemGenBabaInvFormulaIffAltCondYIsZPlus}) which, by the hypothesis $\ker X/X_{11} \subseteq X_{10}$, we do so by proving $\ker X_{11}^+X_{10}=\ker X_{10}$. To show this we first notice that $\ker X_{10}\subseteq\ker X_{11}^+X_{10}$. To prove the reverse inclusion, we use the hypotheses $X^*=X$ and $\ker X_{11}\subseteq\ker X_{01}$ which implies that $X_{11}^*=X_{11}$ and hence $\ran X_{10}\subseteq\ran X_{11}$. Hence, if $E_0\in \ker X_{11}^+X_{10}$ then $X_{11}^+X_{10}E_0=0$ implying $0=X_{11}X_{11}^+X_{10}E_0=\Gamma_{\ran X_{11}}X_{10}E_0=X_{10}E_0$ so that $E_0\in \ker X_{10}$. This proves $\ker X_{11}^+X_{10}\subseteq \ker X_{11}$ and therefore we have proven that $\ker X_{11}^+X_{10}=\ker X_{10}$ which proves that $X^+=Y$, i.e., proves (\ref{LemGenBabaInvFormulaPart1}).
	
	Conversely, suppose that (\ref{LemGenBabaInvFormulaPart1}) is true, i.e., $Y=X^+$. Let $E\in \ker X_{11}$. Then $X(E)=X_{01}(E)+X_{11}(E)=X_{01}(E)$ and $\ker X_{11}=\ker X_{11}^+$ (since $X^*=X$ implies $X_{11}^*=X_{11}$ implies by Lemma \ref{LemMPProp} that $\ker X_{11}=\ker X_{11}^+$). From these facts together with the block form (\ref{LemGenBabaInvFormulaPart1}), it follows that $(X^+X)(E)=0$ and hence $X_{01}(E)=X(E)=(XX^+X)(E)=X(XX^+(E))=0$. This proves that $E\in \ker X_{01}$. Thus, $\ker X_{11}\subseteq\ker X_{01}$ and hence $Z=X$ by Proposition \ref{PropgscFact}. Consequently, as $Z=X, Y=X^+=Z^+$, it follows from above that the inclusion (\ref{LemGenBabaInvFormulaIffAltCondYIsZPlus}) is true. Hence, to prove (\ref{LemGenBabaInvFormulaPart2}) it suffices to show that $\ker X_{11}^+X_{10}\subseteq \ker X_{11}$. The same argument above proves this though. This proves (\ref{LemGenBabaInvFormulaIffAltCondYIsZPlus}), which completes the proof of the lemma.
\end{proof}

The following examples show that hypothesis (\ref{LemGenBabaInvFormulaPart2}) in the previous theorem cannot be weakened for if it is false then we can have both (\ref{LemGenBabaInvFormulaPart1}) and (\ref{LemGenBabaInvFormulaPart3}) be false.
\begin{example}\label{Ex1NoWeakenHypGenBabachInvFormula}
    Consider the following invertible self-adjoint operator $X\in \mathcal{L}(H)$ with $H=\mathbb{K}^2$ defined by 
    \begin{gather*}
        H=H_0\ho H_1,\;H_0=\left\{\begin{bmatrix}
        x\\
        0
        \end{bmatrix}:x\in \mathbb{K}\right\},\;H_1=\left\{\begin{bmatrix}
        0\\
        y
        \end{bmatrix}:y\in \mathbb{K}\right\},\\
        X=[X_{ij}]_{i,j=0,1}=\begin{bmatrix}
        [1]&[1]\\
        [1]&[0]
        \end{bmatrix}.
    \end{gather*}
    Then
    \begin{gather*}
        \mathbb{K}^2=\ker \begin{bmatrix}
        0
        \end{bmatrix}=\ker X_{11} \not\subseteq \ker X_{01}=\ker \begin{bmatrix}
        1
        \end{bmatrix}=\{0\},\\
        \{0\}=\ker \begin{bmatrix}
        1
        \end{bmatrix}=\ker X/X_{11}\subseteq \ker X_{10}=\ker \begin{bmatrix}
        1
        \end{bmatrix}=\{0\},\\
        X^+=X^{-1}=\begin{bmatrix}
        [0]&[1]\\
        [1]&[-1]
        \end{bmatrix},\\
        \hspace{-1em}\begin{bmatrix}
            (X/X_{11})^+& -(X/X_{11})^+X_{01}(X_{11})^+\\
             -(X_{11})^+X_{10} (X/X_{11})^+ &  (X_{11})^++(X_{11})^+X_{10}(X/X_{11})^+X_{01}(X_{11})^+
        \end{bmatrix}\\
        =\begin{bmatrix}
            [1]& [0]\\
             [0] &  [0]
        \end{bmatrix}\neq X^+,\\ X/X_{11}=[1]\not=[0]=[(X^{+})_{00}]^{+}.
    \end{gather*}
\end{example}

\begin{example}\label{Ex2NoWeakenHypGenBabachInvFormula}
 Consider the Hilbert space $H=\mathbb{K}^2=H_0\ho H_1$ defined in Example \ref{Ex1NoWeakenHypGenBabachInvFormula} and the self-adjoint positive semidefinite operator $X\in \mathcal{L}(H)$ defined by
    \begin{gather*}
        X=[X_{ij}]_{i,j=0,1}=\begin{bmatrix}
        [1]&[1]\\
        [1]&[1]
        \end{bmatrix}.
    \end{gather*}
    Then
    \begin{gather*}
        \{0\}=\ker \begin{bmatrix}
        1
        \end{bmatrix}=\ker X_{11} \subseteq \ker X_{01}=\ker \begin{bmatrix}
        1
        \end{bmatrix}=\{0\},\\
        \mathbb{K}^2=\ker \begin{bmatrix}
        0
        \end{bmatrix}=\ker X/X_{11}\not\subseteq \ker X_{10}=\ker \begin{bmatrix}
        1
        \end{bmatrix}=\{0\},\\
        X=\begin{bmatrix}
        [1]&[1]\\
        [1]&[1]
        \end{bmatrix},\;
        X^+=\begin{bmatrix}
            [\frac{1}{4}]&[\frac{1}{4}]\\
            [\frac{1}{4}]&[\frac{1}{4}]
        \end{bmatrix},\\
        \begin{bmatrix}
            (X/X_{11})^+& -(X/X_{11})^+X_{01}(X_{11})^+\\
             -(X_{11})^+X_{10} (X/X_{11})^+ &  (X_{11})^++(X_{11})^+X_{10}(X/X_{11})^+X_{01}(X_{11})^+
        \end{bmatrix}\\
        =\begin{bmatrix}
            [0]& [0]\\
             [0] &  [1]
        \end{bmatrix}\neq X^+,\\
        X/X_{11}=[0]\not=[4]=[(X^{+})_{00}]^{+}.
    \end{gather*}
\end{example}

For the next proposition,
$\sigma/\sigma_{11}$ will denote the generalized Schur complement of $\sigma$ with respect to $\sigma_{11}$ and the orthogonal decomposition $\mathcal{H}=(\mathcal{U}\ho \mathcal{J})\ho \mathcal{E}$. More precisely, with respect to the decomposition $\mathcal{H}=(\mathcal{U}\ho \mathcal{J})\ho \mathcal{E}$, the following block formulas hold:
    \begin{gather}
        \sigma=\begin{bmatrix}%
\begin{bmatrix}
\sigma_{00} & \sigma_{02}\\
\sigma_{20} & \sigma_{22}%
\end{bmatrix}
&
\begin{bmatrix}
\sigma_{01}\\
\sigma_{21}%
\end{bmatrix}
\\%
\begin{bmatrix}
\sigma_{10} & \sigma_{12}%
\end{bmatrix}
& \sigma_{11}%
\end{bmatrix},\label{PropEqPseudEffOpBlockDecompSigma}\\
\sigma/\sigma_{11}=\begin{bmatrix}
\sigma_{00} & \sigma_{02}\\
\sigma_{20} & \sigma_{22}%
\end{bmatrix}-\begin{bmatrix}
\sigma_{01}\\
\sigma_{21}%
\end{bmatrix}\sigma_{11}^{+}\begin{bmatrix}
\sigma_{10} & \sigma_{12}%
\end{bmatrix}\label{PropEqPseudEffOpBlockDecompSigmaSchurComp1}\\
=\begin{bmatrix}
\sigma_{00}-\sigma_{01}\sigma_{11}^{+}\sigma_{10} & \sigma_{02}-\sigma_{01}\sigma_{11}^{+}\sigma_{12}\\
\sigma_{20}-\sigma_{21}\sigma_{11}^{+}\sigma_{10} & \sigma_{22}-\sigma_{21}\sigma_{11}^{+}\sigma_{12}%
\end{bmatrix},\label{PropEqPseudEffOpBlockDecompSigmaSchurComp2}
    \end{gather}
where the latter block decomposition of $\sigma/\sigma_{11}\in\mathcal{L}(\mathcal{U}\ho \mathcal{J})$ is with respect to the Hilbert space $\mathcal{U}\ho \mathcal{J}$.

\begin{proposition}\label{PropEqPseudEffOp}
    If $(\mathcal{H},\mathcal{U},\mathcal{E},\mathcal{J},\sigma)$ is a $Z$-problem with $\dim \mathcal{H}<\infty$ such that
    \begin{enumerate}[(a)]
        \item $\sigma^*=\sigma$,
        \item $\ker \sigma_{11}\subseteq \ker \begin{bmatrix}
                \sigma_{01}\\
                \sigma_{21}
        \end{bmatrix}$,
         \item $\ker \sigma/\sigma_{11}\subseteq \ker \begin{bmatrix}
                \sigma_{10} & \sigma_{12}
        \end{bmatrix}$ [see (\ref{PropEqPseudEffOpBlockDecompSigma}), (\ref{PropEqPseudEffOpBlockDecompSigmaSchurComp1}), \ref{PropEqPseudEffOpBlockDecompSigmaSchurComp2}],
         \item $\ker (\sigma^+)_{22}\subseteq \ker (\sigma^+)_{02}$,
        \item $\ker (\sigma^+)_{*^\prime}\subseteq \ker (\sigma^+)_{20}$,
    \end{enumerate}
    then Theorem \ref{ThmZProbSol} is true for both the direct $Z$-problem $(\mathcal{H},\mathcal{U},\mathcal{E},\mathcal{J},\sigma)$ and for the dual $Z$-problem $(\mathcal{H},\mathcal{U},\mathcal{J},\mathcal{E},\sigma^+)$. Furthermore, the effective operator $\sigma_*$ and the dual effective operator $(\sigma^+)_{*'}$ satisfies the identities
    \begin{gather}
        \sigma_*=\sigma_{00}-\sigma_{01}\sigma_{11}^+\sigma_{10},\\(\sigma^+)_{*'}=(\sigma^+)_{00}-(\sigma^+)_{02}(\sigma^+)_{22}^+(\sigma^+)_{20},\\
        (\sigma^+)_{*'}=(\sigma_*)^+.
    \end{gather}
\end{proposition}
\begin{proof}
     Suppose that $(\mathcal{H},\mathcal{U},\mathcal{E},\mathcal{J},\sigma)$ is a $Z$-problem with $\dim \mathcal{H}<\infty$. Assume $(a)$ and $(b)$ are true. Then
     \begin{align*}
         \ker \sigma_{11}\subseteq \ker \begin{bmatrix}
                \sigma_{01}\\
                \sigma_{21}
        \end{bmatrix}=\ker\sigma_{01}\cap\ker\sigma_{21}\subseteq\ker\sigma_{01}
     \end{align*}
     implying Theorem \ref{ThmZProbSol} is true and hence
     \begin{align*}
         \sigma_*=\begin{bmatrix}
		\sigma_{00} & \sigma_{01}\\
		\sigma_{10} & \sigma_{11}
		\end{bmatrix}/\sigma_{11}=\sigma_{00}-\sigma_{01}\sigma_{11}^+\sigma_{10}.
     \end{align*}
     Assume also that $(c)$ is true. Then by Lemma \ref{LemGenBabaInvFormula} with $H=H_0\ho H_1, H_0=\mathcal{U}\ho \mathcal{J}, H_1=\mathcal{E},$ \begin{align*}
         \sigma=X=[X_{ij}]_{i,j=0,1}=
\begin{bmatrix}%
\begin{bmatrix}
\sigma_{00} & \sigma_{02}\\
\sigma_{20} & \sigma_{22}%
\end{bmatrix}
&
\begin{bmatrix}
\sigma_{01}\\
\sigma_{21}%
\end{bmatrix}
\\%
\begin{bmatrix}
\sigma_{10} & \sigma_{12}%
\end{bmatrix}
& \sigma_{11}%
\end{bmatrix}
\end{align*}
we have 
\begin{align*}
    \begin{bmatrix}
(\sigma^+)_{00} & (\sigma^+)_{02}\\
(\sigma^+)_{20} & (\sigma^+)_{22}%
\end{bmatrix}&=(X^+)_{00}=(X/X_{11})^+,\\
X/X_{11}=\sigma/\sigma_{11}&=\begin{bmatrix}
\sigma_* & \sigma_{02}-\sigma_{01}\sigma_{11}^{+}\sigma_{12}\\
\sigma_{20}-\sigma_{21}\sigma_{11}^{+}\sigma_{10} & \sigma_{22}-\sigma_{21}\sigma_{11}^{+}\sigma_{12}
\end{bmatrix}.
\end{align*}
Now assume also that $(d)$ is true. Then Theorem \ref{ThmZProbSol} is true for the dual $Z$-problem $(\mathcal{H},\mathcal{U},\mathcal{J},\mathcal{E},\sigma^+)$ so that
\begin{align*}
         (\sigma^+)_{*^\prime}&=\begin{bmatrix}
		(\sigma^+)_{00} & (\sigma^+)_{02}\\
		(\sigma^+)_{20} & (\sigma^+)_{22}
		\end{bmatrix}/(\sigma^+)_{22}\\
		&=(\sigma^+)_{00}-(\sigma^+)_{02}(\sigma^+)_{22}^+(\sigma^+)_{20}.
\end{align*}
Finally, assume also that $(e)$ is true. Then by Lemma \ref{LemGenBabaInvFormula} with $H=H_0\ho H_1, H_0=\mathcal{U}, H_1=\mathcal{J}$,
\begin{align*}
 X=[X_{ij}]_{i,j=0,1}=\begin{bmatrix}
		(\sigma^+)_{00} & (\sigma^+)_{02}\\
		(\sigma^+)_{20} & (\sigma^+)_{22}
		\end{bmatrix}
\end{align*}
we have 
     \begin{gather*}
         [(\sigma^+)_{*^\prime}]^+=\left\{\begin{bmatrix}
		(\sigma^+)_{00} & (\sigma^+)_{02}\\
		(\sigma^+)_{20} & (\sigma^+)_{22}
		\end{bmatrix}/(\sigma^+)_{22}\right\}^+\\
		=(X/X_{11})^+=(X^+)_{00}=\left\{\begin{bmatrix}
		(\sigma^+)_{00} & (\sigma^+)_{02}\\
		(\sigma^+)_{20} & (\sigma^+)_{22}
		\end{bmatrix}^+\right\}_{00}\\
		=(\sigma/\sigma_{11})_{00}=\sigma_*.
     \end{gather*} 
Therefore,  
\begin{gather*}
        (\sigma^+)_{*'}=\{[(\sigma^+)_{*'}]^+\}^+=(\sigma_*)^+.
\end{gather*}\end{proof}

Using Examples \ref{Ex1NoWeakenHypGenBabachInvFormula} and \ref{Ex2NoWeakenHypGenBabachInvFormula} with $H=\mathbb{K}^2,\mathcal{U}=H_0, \mathcal{E}=H_1, \mathcal{J}=\{0\}, \sigma=X$ shows that hypotheses $(b)$ and $(c)$ in the previous proposition cannot be weakened for if one of them is false then the conclusion of the proposition is false. In particular, the latter example has $\sigma^*=\sigma\geq 0$ with hypotheses $(a)$, $(b)$, $(d)$, and $(e)$ all true so that Theorem \ref{ThmZProbSol} is true for both the direct $Z$-problem $(\mathcal{H},\mathcal{U},\mathcal{E},\mathcal{J},\sigma)$ and for the dual $Z$-problem $(\mathcal{H},\mathcal{U},\mathcal{J},\mathcal{E},\sigma^+)$, but $(\sigma^+)_{*'}\not=(\sigma_*)^+$.

In order to prove below our generalization of the Thomson variational principle (Theorem \ref{ThmClassicalThomMinPrin}) for the generalized dual problem (Definition \ref{DefGenDualZProb}), we will need the following lemma. 
\begin{lemma}\label{LemThomPrelim}
    Let $H$ be a Hilbert space with $\dim H<\infty$ and $X,Y\in\mathcal{L}(H)$ with $X^*=X, Y^*=Y$. Then the following statements are true:
    \begin{enumerate}[(a)]
        \item $\ker X=\ker X^+$.
        \item $X^+|_{\ran X}=X|_{\ran X}^{-1}$.
        \item If $0\leq X$ then $0\leq X^+$.
        \item If $0\leq X\leq Y$, then $\ker X\subseteq \ker Y$.
        \item If $0\leq X\leq Y$ and $X$ is invertible, then $Y$ is invertible and $Y^{-1}\leq X^{-1}$.
        \item If $0\leq X\leq Y$ then $Y^+\leq X^+$ if and only if $\ker Y=\ker X$.
    \end{enumerate}
\end{lemma}
\begin{proof}
     All the statements in this lemma are well-known except possibly statement $(f)$ so we prove it here. Assume the hypotheses. Suppose that $0\leq X\leq Y$. Then by part $(d)$ we have $\ker X\subseteq \ker Y$. If $Y^+\leq X^+$ then by $(c)$ we have $0\leq Y^+\leq X^+$ so that by $(a)$ and $(d)$ we have $\ker Y=\ker Y^+\subseteq \ker X^+=\ker X$ which proves that $\ker Y=\ker X$. Conversely, if $\ker Y=\ker X$ then taking the orthogonal complement we have $\ran Y=\ran X$ and hence $ 0\leq X|_{\ran X}\leq Y|_{\ran Y}$ which implies by $(e)$ that $0 \leq Y|_{\ran Y}^{-1} \leq X|_{\ran X}^{-1}$ and thus from this, statement $(b)$, and the hypothesis that $\ker Y=\ker X$ we conclude that $Y^+\leq X^+$. This proves $(f)$.
\end{proof}

\begin{theorem}[Generalized Thomson minimization principle]\label{ThmGenThomMinPrin}
Suppose $(\mathcal{H},\mathcal{U},\mathcal{E},\mathcal{J},\sigma)$ is a $Z$-problem with $\dim \mathcal{H}<\infty$. If
\begin{align}
    (\sigma^+)_{22}\geq 0\label{ThmGenThomMinPrinHyp1}
\end{align}
and the hypotheses $(a)$-$(e)$ in Proposition \ref{PropEqPseudEffOp} are satisfied, then Theorem \ref{ThmZProbSol} is true for both the direct and dual $Z$-problems $(\mathcal{H},\mathcal{U},\mathcal{E},\mathcal{J},\sigma)$ and $(\mathcal{H},\mathcal{U},\mathcal{J},\mathcal{E},\sigma^+)$, respectively, the effective operator $\sigma_*$ and dual effective operator $(\sigma^+)_{*^\prime}$ satisfy the identity
\begin{align}
    (\sigma^+)_{*^\prime}=(\sigma_*)^+,
\end{align}
and $(\sigma_*)^+$ is the unique self-adjoint operator satisfying the minimization principle:
	\begin{equation}
		\left(J_0,(\sigma_*)^+J_0\right)=\min_{J\in \mathcal{J}}\left( J_0+J, \sigma^+(J_0+J) \right),\;\forall J_0\in \mathcal{U}.
	\end{equation}
Furthermore, for each $J_0\in \mathcal{U}$, the set of minimizers is given by 
	\begin{eqnarray}
        \left\{J\in\mathcal{J}:J=-[(\sigma^+)_{22}]^+(\sigma^+)_{20}J_0+L,\;L\in \ker(\sigma^+)_{22}\right\}.\nonumber\\
	\end{eqnarray}
Moreover, we have the following upper bound on the dual effective operator:
	\begin{gather}
	    (\sigma_*)^+\leq (\sigma^+)_{00}.\label{ThmGenThomMinPrinUpperBoundsDualEffOp}
	\end{gather}
If, in addition, $\sigma_*\geq 0$ then we have the following upper and lower bound on the effective operator:
    \begin{gather}
        0\leq [\Gamma_{\ran \sigma_*}(\sigma^+)_{00}\Gamma_{\ran \sigma_*}]^+\leq \sigma_*\leq \sigma_{00},\label{ThmGenThomMinPrinUpperAndLowerBounds1}
    \end{gather}
    where $\Gamma_{\ran \sigma_*}$ is the orthogonal projection of $\;\mathcal{U}$ onto $\ran \sigma_*$, i.e., 
    \begin{align}
        \Gamma_{\ran \sigma_*}=\sigma_*(\sigma_*)^+=(\sigma_*)^+\sigma_*.\label{ThmGenThomMinPrinUpperAndLowerBounds2}
    \end{align}
\end{theorem}
\begin{proof}
    Suppose $(\mathcal{H},\mathcal{U},\mathcal{E},\mathcal{J},\sigma)$ is a $Z$-problem with $\dim \mathcal{H}<\infty$, and the hypotheses $(a)$-$(e)$ in Proposition \ref{PropEqPseudEffOp} are satisfied. Then, in particular, Proposition \ref{PropEqPseudEffOp} is true so that Theorem \ref{ThmZProbSol} is true for both the direct and dual $Z$-problems $(\mathcal{H},\mathcal{U},\mathcal{E},\mathcal{J},\sigma)$ and $(\mathcal{H},\mathcal{U},\mathcal{J},\mathcal{E},\sigma^+)$, respectively, the effective operator $\sigma_*$ and dual effective operator $(\sigma^+)_{*^\prime}$ satisfy the identity $ (\sigma_*)^+\leq (\sigma^+)_{00}.$ Now assume also that $(\sigma^+)_{22}\geq 0$. Then the generalized Dirichlet minimization principle (i.e., Theorem \ref{ThmGenClassicalDiriMinPrin}) applies to the dual $Z$-problem $(\mathcal{H},\mathcal{U},\mathcal{J},\mathcal{E},\sigma^+)$ which proves the theorem [with the exception of (\ref{ThmGenThomMinPrinUpperAndLowerBounds1})] by this duality. Suppose, in addition, that $\sigma_*\geq 0$. First, by hypothesis we have $\sigma^*=\sigma$ which implies that $(\sigma_*)^*=\sigma_*$ and hence by Lemma \ref{LemMPProp} [namely, statements $(2)$ and $(3)$] we have that $\Gamma_{\ran \sigma_*}$, the orthogonal projection of $\;\mathcal{U}$ onto $\ran \sigma_*$, satisfies the identities (\ref{ThmGenThomMinPrinUpperAndLowerBounds2}). Next, by Lemma \ref{LemThomPrelim} and inequality (\ref{ThmGenThomMinPrinUpperBoundsDualEffOp}) it follows that
    \begin{align}
        0\leq (\sigma_*)^+=\Gamma_{\ran \sigma_*}(\sigma_*)^+\Gamma_{\ran \sigma_*}\leq \Gamma_{\ran \sigma_*}(\sigma^+)_{00}\Gamma_{\ran \sigma_*}.
    \end{align}
    Finally, from this and Lemma \ref{LemThomPrelim}.$(f)$ along with inequality (\ref{ThmGenThomMinPrinUpperBoundsDualEffOp}) and Lemma \ref{LemMPProp}.$(6)$, the inequality (\ref{ThmGenThomMinPrinUpperAndLowerBounds1}) follows immediately. This completes the proof of the theorem.
\end{proof}

\begin{corollary}
If $(\mathcal{H},\mathcal{U},\mathcal{E},\mathcal{J},\sigma)$ is a $Z$-problem with $\dim \mathcal{H}<\infty$ such that
\begin{align}
    \sigma^*=\sigma\geq 0,
\end{align}
then the inequality (\ref{ThmGenThomMinPrinHyp1}) holds,  hypotheses $(a)$, $(b)$, and $(d)$ in Proposition \ref{PropEqPseudEffOp} are satisfied, and $(\sigma_*)^*=\sigma_*\geq 0.$
If, in addition, hypotheses $(c)$ and $(e)$ in Proposition \ref{PropEqPseudEffOp} are satisfied then Theorem \ref{ThmGenThomMinPrin} is true and, in particular, the inequalities (\ref{ThmGenThomMinPrinUpperAndLowerBounds1}) are true, i.e.,
    \begin{gather}
        0\leq [\Gamma_{\ran \sigma_*}(\sigma^+)_{00}\Gamma_{\ran \sigma_*}]^+\leq \sigma_*\leq \sigma_{00}.
    \end{gather}
\end{corollary}
\begin{proof}
    Suppose $(\mathcal{H},\mathcal{U},\mathcal{E},\mathcal{J},\sigma)$ is a $Z$-problem with $\dim \mathcal{H}<\infty$ such that $\sigma^*=\sigma\geq 0$ [hence hypotheses $(a)$ in Proposition \ref{PropEqPseudEffOp} is satisfied]. Then by Lemma \ref{LemMPProp}.$(1)$ and Lemma \ref{LemThomPrelim}.$(c)$ we have $(\sigma^+)^*=\sigma^+\geq 0$. From the block decomposition (\ref{PropEqPseudEffOpBlockDecompSigma}) for $\sigma$ and the corresponding one for $\sigma^+$, it immediately follow from Proposition \ref{PropgscFact} that hypotheses $(b)$ and $(d)$ in Proposition \ref{PropEqPseudEffOp} is satisfied. The rest of the proof of this corollary is now obvious.
\end{proof}

\section{\label{sec:DiscreteNetworkExamples}Examples from Discrete electrical network problems}

The following function space and its basic properties as a vector space are standard in linear algebra (see, for example, Ref.\ \onlinecite{19FIS}).
\begin{definition}\label{DefGenFunSpace}
	Let $\mathcal{T}$ be a nonempty set. The set of all functions from $\mathcal{T}$ to $\mathbb{K}$ (where $\mathbb{K}=\mathbb{R}$ or $\mathbb{K}=\mathbb{C}$) is denoted by
	\begin{align}
		\mathcal{F}(\mathcal{T},\mathbb{K})=\{f\;|\;f:\mathcal{T}\to\mathbb{K}\}. 
	\end{align}
	The functions $0_{\mathcal{T}}$ and $1_{\mathcal{T}}$ are given by
	\begin{gather}
		0_{\mathcal{T}}(t)=0,\;\forall t\in \mathcal{T},\\
		1_{\mathcal{T}}(t)=1,\;\forall t\in \mathcal{T},
	\end{gather}
	and the identity operator on $\mathcal{F}(\mathcal{T},\mathbb{K})$ is denoted by $I_{\mathcal{F}(\mathcal{T},\mathbb{K})}$.
\end{definition}

\begin{lemma}\label{PropHilbertSpace}
	If $\mathcal{T}$ is a non-empty set, then $\mathcal{F}(\mathcal{T},\mathbb{K})$ is a vector space over the field $\mathbb{K}$. Furthermore, if $\mathcal{T}=\{t_1,\ldots,t_{|\mathcal{T}|}\}$ is finite, then $\mathcal{F}(\mathcal{T},\mathbb{K})$ is a Hilbert space under the inner product
	\begin{equation}\label{PropInnerProd}
		(f,g)_\mathcal{T}=\sum_{t\in\mathcal{T}}\overline{f(t)}g(t),\; \forall f,g\in \mathcal{F}(\mathcal{T},\mathbb{K}).
	\end{equation}
	Moreover, the set
	\begin{equation}\label{PropOrthBas}
		\alpha_{\mathcal{T}}=\{\delta_{t_i}: i=1,\ldots, |\mathcal{T}|\},\quad \delta_{t_i}(t_j)=\delta_{ij},\;\;1\leq i,j\leq |\mathcal{T}|,
	\end{equation}
	is an orthonormal basis for $\mathcal{F}(\mathcal{T},\mathbb{K})$, where $\delta_{ij}$ is the Kronecker delta, i.e.,
	\begin{equation}
	    \delta_{ij}=
	    \begin{cases}
	        1,\quad i=j,\\
	        0,\quad i\neq j.
	    \end{cases}\label{DefKroneckerDelta}
	\end{equation}
\end{lemma}

\begin{notation}
     The matrix representation for any linear operator $L\in\mathcal{L}(\mathcal{F}(\mathcal{T}_1,\mathbb{K}),\mathcal{F}(\mathcal{T}_2,\mathbb{K}))$ with respect to the ordered bases $\alpha_{\mathcal{T}_1},\; \alpha_{\mathcal{T}_2}$ will be denoted by 
     $\begin{bmatrix}
        L
    \end{bmatrix}_{\alpha_{\mathcal{T}_1}}^{\alpha_{\mathcal{T}_2}}$.
\end{notation}

\subsection{Operator Framework for Electrical Conductivity in Finite Networks}\label{SubSectOpView}

Let $G=(P_G, E_G)$ denote a finite linear digraph with finite nonempty node and directed edge sets $P_G=\{p_j:j=1, \ldots , |P_G|\}$ and $E_G=\{e_j:j=1,,\ldots,|E_G|\}$, respectively. For each edge $e\in E_G$, the outgoing and incoming incident nodes are denoted by $e_-$ and $e_+$, respectively [more precisely, these are functions $(\cdot)_{\pm}:E_G\to P_G$] as shown in Figure \ref{fig:eplusminus}. We will assume that the graph $G$ has no (self) loops, i.e., $(e)_-\not=(e)_+$ for all $e\in E_G$.
\begin{figure}[ht!]
	\begin{center} 
		\begin{circuitikz}[scale=1.5, /tikz/circuitikz/bipoles/length=0.75cm]
			\draw(0,0) node[above] {$e_-$} to[thick, short, i=\phantom{},l=$e$] (2,0) node[circ, color=black] {};
			\draw  (0,0) node[circ, color=black] {};
			\draw (2,0) node[above] {$e_+$};
			\draw  (1,0) node[circ, color=black] {}; 
		\end{circuitikz}
	\end{center}
	\caption{Example of a directed edge.}\label{fig:eplusminus}
\end{figure}
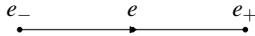

The following definitions of $D$ and $D^{\bullet}$ are the discrete analogs of the gradient $\nabla$ and divergence $\nabla\cdot$ in the continuum.

\begin{definition}\label{DefAnalGradDiv}
	We define the function $D:\mathcal{F}(P_G,\mathbb{K})\to\mathcal{F}(E_G,\mathbb{K})$ by
	\begin{equation}\label{DefDGrad}
		(Df)(e)=f(e_+)-f(e_-),\;\forall (f,e)\in \mathcal{F}(P_G,\mathbb{K})\times E_G,
	\end{equation}
	and the function $D^\bullet:\mathcal{F}(E_{G},\mathbb{K})\to\mathcal{F}(P_G,\mathbb{K})$ by
	\begin{equation}\label{DefDDiv}
		(D^\bullet f)(p)=\sum_{\begin{subarray}{c}e\in E_G,\\e_-=p\end{subarray}}f(e)-\sum_{\begin{subarray}{c}e\in E_G,\\e_+=p\end{subarray}}f(e),\;\forall(f,p)\in \mathcal{F}(E_G,\mathbb{K})\times P_G.
	\end{equation}
\end{definition}

It follows that the matrix representation $\begin{bmatrix}
\cdot
\end{bmatrix}^{\alpha_{E_G}}_{\alpha_{P_G}}$ of $D$ is the edge-node incidence matrix $A_a^T$, i.e.,
\begin{align}\label{RefDMat}
	\begin{bmatrix}                                                   
	D                                                                 
	\end{bmatrix}^{\alpha_{E_G}}_{\alpha_{P_G}}=A_a^T=\begin{bmatrix} 
	a_{ji}                                                            
	\end{bmatrix},                                                     
\end{align}
where $A_a=[a_{ij}]$ is the well known node-edge incidence matrix\cite{69BB, 21AS} defined by
\begin{align} 
	\label{Defnode-edgeincidencematrix}
	a_{ij} & =                                                                                           
	\begin{cases}
	1      & \text{if }(e_j)_-=p_i, \\
	-1     & \text{if }(e_j)_+=p_i,\\
	0      & \text{if }(e_j)_-\not=p_i\not=(e_j)_+ ,                
	\end{cases} 
	\\&=\begin{cases} 
	1      & \begin{circuitikz}[scale=1, /tikz/circuitikz/bipoles/length=0.75cm]                         
	\draw(0,0) node[above] {$p_i$} to[short, i=\phantom{},l=$e_j$] (2,0); 
	\draw  (0,0) node[circ, color=black] {};
	\draw (2,0) node[circ, color=black] {};
	\draw  (1,0) node[circ, color=black] {}; 
	\end{circuitikz} \;,
	\\
	-1     & \begin{circuitikz}[scale=1, /tikz/circuitikz/bipoles/length=0.75cm]                         
	\draw(0,0) node[above] {$p_i$}to[short, i<=\phantom{},l=$e_j$, color=black] (2,0); 
	\draw  (2,0) node[circ, color=black] {}; 
	\draw  (0,0) node[circ, color=black] {}; 
	\end{circuitikz}  \;,
	\\
	0      & \begin{circuitikz}[scale=1, /tikz/circuitikz/bipoles/length=0.75cm]                         
	\draw (1,0) node[circ, color=black] {} -- (3,0) node[circ, color=black] {};
	\draw (2,-0.05) node[above] {$e_j$};
	\draw (0,0) node[above] {$p_i$};
	\draw (0,0) node[circ, color=black] {};
	\end{circuitikz} \;.
	\end{cases}
\end{align} 
It is an easy to verify that the matrix representation $\begin{bmatrix}
\cdot
\end{bmatrix}^{\alpha_{P_G}}_{\alpha_{E_G}}$ of $D^{\bullet}$ is the matrix $-A_a$, i.e., 
\begin{align}\label{RefDBulMat}
	\begin{bmatrix}                                    
	D^{\bullet}                                        
	\end{bmatrix}^{\alpha_{P_G}}_{\alpha_{E_G}} =-A_a. 
\end{align}
In this context, the proof of the following lemma is almost immediate.
\begin{lemma}\label{LemAdjFin}
	The functions $D$ and $D^{\bullet}$ [defined by (\ref{DefDGrad}) and (\ref{DefDDiv})] are well defined linear operators and satisfy the Hilbert space adjoint relations 
	\begin{align}
		D^*=- D^{\bullet},\;(D^{\bullet})^*=-D.
	\end{align} 
\end{lemma}
\begin{proof}
	First, it is obvious the functions $D$ and $D^\bullet$ are well defined and linear. Second, using the identities (\ref{RefDMat}) and (\ref{RefDBulMat}) and the fact that $\alpha_{P_G}$ and $\alpha_{E_G}$ are orthonormal bases, it follows from standard results from linear algebra\cite{19FIS} that
	\begin{gather*}
		\begin{bmatrix}
			D^* 
		\end{bmatrix}^{\alpha_{P_G}}_{\alpha_{E_G}}=
		(\begin{bmatrix}
		D
		\end{bmatrix}^{\alpha_{E_G}}_{\alpha_{P_G}})^*=(\begin{bmatrix}
		D
		\end{bmatrix}^{\alpha_{E_G}}_{\alpha_{P_G}})^T\\
		=(A_a^T)^T=A_a=-\begin{bmatrix}
		D^{\bullet}
		\end{bmatrix}^{\alpha_{P_G}}_{\alpha_{E_G}}=\begin{bmatrix}
		-D^{\bullet}
		\end{bmatrix}^{\alpha_{P_G}}_{\alpha_{E_G}},
	\end{gather*}
	which implies $D^*=- D^{\bullet}$.
\end{proof}

\begin{remark}[Kirchhoff's laws]
    In this network setting, the sets of edge currents and edge voltages are $\ker D^\bullet$ and $\ran D$, respectively. In particular, Kirchhoff's current law (KCL) is equivalent to: $I$ is an edge current if and only if $D^\bullet I=0$. Similarly, Kirchhoff's voltage law (KVL) is equivalent to: $V$ is an edge voltage if and only $V=Du$, for some potential $u$.
\end{remark}

The following result is well known\cite{69BB,97NB, 21AS} (for a proof see Ref.\ \onlinecite{22KB}).

\begin{lemma}\label{LemConKer}
     Let $G_1=(P_{G_1},E_{G_1}),\ldots,G_k=(P_{G_k},E_{G_k})$ be the connected components of the finite linear digraph $G=(P_G,E_G)$. Denote the characteristic function on the node set $P_{G_i}$ of the graph $G_i$ by $\chi_{P_{G_i}}$, i.e.,
     \begin{gather}
        \chi_{P_{G_i}}(x)=
        \begin{cases}
            1,\;\text{if}\; x\in P_{G_i},\\
            0,\;\text{if}\; x\in P_G\setminus P_{G_i},
        \end{cases}
     \end{gather}
     for $i=1,\ldots,k$. Then $\{\chi_{P_{G_1}},\ldots,\chi_{P_{G_k}}\}$ is an orthonormal basis for $\ker D$,
     \begin{gather}
         \ker D=\operatorname{span}\{\chi_{P_{G_1}},\ldots,\chi_{P_{G_k}}\},\; \dim\ker D=k, \;1_{P_G}\in \ker D.
    \end{gather}
     In particular, $G$ is connected (i.e., $k=1$) if and only if $\ker D=\operatorname{span}\{1_{P_G}\}$.
\end{lemma}

\begin{definition}[Conductivity on electrical networks]\label{DefCondElecNet}
    An electrical network is a triple $(P_G,E_G,\sigma)$, with graph $G=(P_G,E_G)$ and operator $\sigma\in\mathcal{L}(\mathcal{F}(E_G,\mathbb{K}))$ satisfying $\sigma^*=\sigma\geq 0$. We call $\sigma$ the network conductivity and we define the Kirchhoff operator $K_\sigma\in\mathcal{L}(\mathcal{F}(P_G,\mathbb{K}))$ by
	\begin{gather}
		K_\sigma = -D^\bullet \sigma D.\label{DefKirchhoffOperator}
	\end{gather}
	In particular, for $\sigma=I$ [the identity operator on $\mathcal{F}(E_G,\mathbb{K})$], the ``normalized" Kirchhoff operator is 
	\begin{align}\label{KirchhoffGammaI}
        K_I = -D^{\bullet} D.
    \end{align}
\end{definition}

\begin{remark}
The normalized Kirchhoff operator $K_I$ is the graph-theoretic analog of $-\Delta=-\nabla\cdot\nabla$, i.e., the negative of the Laplacian in the continuum. Similarly, the Kirchhoff operator $K_\sigma$ (sometimes referred to in the literature as the graph Laplacian) is graph-theoretic analog of $-\nabla\cdot\sigma\nabla$ in the continuum. In the matrix setting, the well known\cite{00CM, 21AS} Kirchhoff matrix is the matrix representation of the Kirchhoff operator $K_{\sigma}$ with respect to the orthonormal bases $\alpha_{P_G}, \alpha_{E_G}$, i.e.,
\begin{align}
	\begin{bmatrix}
    K_{\sigma}
    \end{bmatrix}^{\alpha_{E_G}}_{\alpha_{E_G}} = 
    -\begin{bmatrix}                          
	    D^{\bullet}                                                   
	\end{bmatrix}_{\alpha_{E_G}}^{\alpha_{P_G}} 
	\begin{bmatrix}
        \sigma
    \end{bmatrix}^{\alpha_{E_G}}_{\alpha_{E_G}}
    \begin{bmatrix}                                   
	    D                                                              
	\end{bmatrix}^{\alpha_{E_G}}_{\alpha_{P_G}} = A_a 
	\begin{bmatrix}
        \sigma
    \end{bmatrix}^{\alpha_{E_G}}_{\alpha_{E_G}} A_a^T,
\end{align}
where $\begin{bmatrix}
        \sigma
       \end{bmatrix}^{\alpha_{E_G}}_{\alpha_{E_G}}$
     is the matrix representation of the conductivity $\sigma$, i.e., the conductivity matrix. Typically, in a resistor-only network, the conductivity matrix $\begin{bmatrix}
    \sigma
    \end{bmatrix}^{\alpha_{E_G}}_{\alpha_{E_G}}$ is a diagonal matrix. This corresponds to the conductivity $\sigma$ being given by
    \begin{align}
        \sigma=\sum_{i=1}^{|E_G|}\sigma_{ii}\delta_{e_i},
    \end{align}
    where $\sigma_{ii}$ is the conductivity (i.e., the inverse of resistance) on the $i$th edge of the graph $G=(P_G,E_G)$. 
\end{remark}

Let us conclude this section now with the following physical interpretation for an electrical network $(P_G,E_G,\sigma)$ following Ref.\ \onlinecite{69BB}. By KVL, the set of edge voltages of the network (containing no voltage sources) is $\ran D$. Now, any edge voltage $V\in \ran D$ in the network can be considered as being generated by a current source $I_g\in \mathcal{F}(E_G,\mathbb{K})$ that is balanced by a passive current $I\in \mathcal{F}(E_G,\mathbb{K})$ in the network, i.e., they must satisfy $I-I_g\in \ker D^{\bullet}$ by KCL, and
$$I=\sigma V$$ by Ohm's law. As $V=Du$, for some potential $u\in \mathcal{F}(P_G,\mathbb{K}),$ then it follows that $u$ satisfies the equations
$$K_{\sigma}u=\Phi,$$ where $\Phi=-D^\bullet I_g\in \ran D^\bullet$ is considered as a node current-source since, for each $p\in P_G,$ $$\Phi(p)=-(D^\bullet I_g)(p)=-\sum_{\begin{subarray}{c}e\in E_G,\\e_-=p\end{subarray}}I_g(e)+\sum_{\begin{subarray}{c}e\in E_G,\\e_+=p\end{subarray}}I_g(e)$$
is the algebraic sum of current sources incident at the node $p$ (with references chosen so that they enter the node). 

Conversely, suppose $u\in \mathcal{F}(P_G,\mathbb{K}).$ Then $K_\sigma u\in \ran D^\bullet$ and so $\exists\; I_g\in \mathcal{F}(E_G,\mathbb{K})$ such that $K_{\sigma}u=-D^\bullet I_g.$ Define $V=Du$ and $I=\sigma V$. Then $V\in \ran D, I\in \mathcal{F}(E_G,\mathbb{K}),$ and $-D^\bullet I=-D^\bullet \sigma V=K_{\sigma}u=-D^\bullet I_g$ implying $I-I_g\in \ker D^{\bullet}$. Thus, according to the interpretation above, the edge voltage $V$ in the network can be considered as being generated by the source current $I_g$ that is balanced by the passive current $I$ in the network. 

\subsection{Dirichlet-to-Neumann Map}\label{SubSectDtNZ}
In this subsection, we will use our operator framework (from Sec.\ \ref{SubSectOpView}) to reformulate an example from Ref.\ \onlinecite{16GM} (see Ch. 2, Sec. 2.13) on finite linear digraphs and the characterization of their discrete Dirichlet-to-Neumann (DtN) maps in terms of effective operators and a Hodge decomposition associated with the discrete Dirichlet problem. For further background and motivation for the DtN map in this setting, see Refs.\ \onlinecite{00CM, 21AS}, for instance.

Consider an electrical network $(P_G,E_G,\sigma)$ (see Def.\ \ref{DefCondElecNet}) with the node set $P_G$ partitioned into two nonempty subsets $P_{\partial G}$ and $P_{G^\circ}$, i.e.,
\begin{align}
    P_G=P_{\partial G}\cup P_{G^\circ},\;\; P_{\partial G}\cap P_{G^\circ}=\emptyset,\;\; P_{\partial G}\neq\emptyset,\;\; P_{G^\circ}\neq\emptyset.\label{AssDtNMapGraphBddIntNodes}
\end{align}
The boundary nodes and interior nodes for the graph $G=(P_G,E_G)$ are defined as $P_{\partial G}$ and $P_{G^\circ}$, respectively. 

We will now define the discrete DtN map in terms of the discrete analog of the Dirichlet problem for the conductivity equation in the continuum. To do this, we will need the restriction maps $|_{P_{G^\circ}}:\mathcal{F}(P_G,\mathbb{K})\rightarrow\mathcal{F}(P_{G^\circ},\mathbb{K})$ and $|_{P_{\partial G}}:\mathcal{F}(P_G,\mathbb{K})\rightarrow\mathcal{F}(P_{\partial G},\mathbb{K})$ which are the bounded linear operators defined by
\begin{gather}
    u|_{P_{G^\circ}}(p)=u(p),\;\forall p\in P_{G^\circ},\label{DefRestrictionMapInteriorNodes}\\
    u|_{P_{\partial G}}(p)=u(p),\;\forall p\in P_{\partial G},\label{DefRestrictionMapBoundaryNodes}
\end{gather}
for each $u\in \mathcal{F}(P_G,\mathbb{K})$.
\begin{definition}[Discrete Dirichlet problem]\label{DefDiscreteDirchProb}
    The discrete Dirichlet problem is defined as follows: given a function $f\in\mathcal{F}(P_{\partial G},\mathbb{K})$, find a function $u\in\mathcal{F}(P_G,\mathbb{K})$ satisfying
    \begin{gather}
        (-D^\bullet\sigma Du)|_{P_{G^\circ}}=0,\label{DefDisDirProb}\\
        u|_{P_{\partial G}}=f\label{DefDisBoundCond}.
    \end{gather}
    For each such solution $u\in\mathcal{F}(P_G,\mathbb{K})$, the boundary (source) current $\phi$ is defined by
    \begin{equation}\label{DefBoundCurPhi}
        \phi=(-D^\bullet\sigma Du)|_{P_{\partial G}}.
    \end{equation}
\end{definition}

\begin{definition}[Dirichlet-to-Neumann map]\label{DefDtNMap}
    The Dirichlet-to-Neumann (DtN) map is the operator $\Lambda_\sigma\in\mathcal{L}(\mathcal{F}(P_{\partial G},\mathbb{K}))$ that maps the potential on the boundary $u|_{P_{\partial G}}$ for each solution $u\in\mathcal{F}(P_G,\mathbb{K})$ of the discrete Dirichlet problem (\ref{DefDisDirProb}) and (\ref{DefDisBoundCond}) to the boundary current $\phi$ (\ref{DefBoundCurPhi}), i.e.,
    \begin{gather}
        \Lambda_\sigma u|_{P_{\partial G}}=\phi.
    \end{gather}
\end{definition}

\subsubsection{Generalized Schur Complement Representation}

In this section we will derive a generalized Schur complement formula for the DtN map $\Lambda_\sigma$. To do so, we will need the following. The Hilbert space $H=\mathcal{F}(P_G,\mathbb{K})$ decomposes into the orthogonal direct sum 
\begin{gather}
    \mathcal{F}(P_G,\mathbb{K})=H_0\ho H_1,\;H_0=\ker (|_{P_{G^\circ}}),\;H_1=\ker (|_{P_{\partial G}}).\label{DefDtNMapSchurComplOrthogDecomposition}
\end{gather}
With respect to this decomposition, the Kirchhoff operator [see Def.\ \ref{DefCondElecNet}] can be written in $2\times 2$ block operator form as
    \begin{gather}
        K_{\sigma}=\begin{bmatrix}
            (K_{\sigma})_{00} & (K_{\sigma})_{01}\\
            (K_{\sigma})_{10} & (K_{\sigma})_{11}
        \end{bmatrix},\\
        \mathcal{L}(H_j,H_i)\ni(K_{\sigma})_{ij}=\Gamma_{H_i}K_{\sigma}\Gamma_{H_j}:H_j\rightarrow H_i,\;i,j=0,1,\label{DefKirchhoffOpBlocks}
    \end{gather}
where $\Gamma_{H_0}$ and $\Gamma_{H_1}$ are the orthogonal projections of $\mathcal{F}(P_G,\mathbb{K})$ onto $H_0$ and $H_1$, respectively. 

Next, we define the zero extension maps $\iota_0:\mathcal{F}(P_{\partial G},\mathbb{K})\rightarrow \mathcal{F}(P_{G},\mathbb{K})$ and $\iota_1:\mathcal{F}(P_{G^\circ},\mathbb{K})\rightarrow \mathcal{F}(P_{G},\mathbb{K})$ by
\begin{gather}
    [\iota_0(f)](p)=\left\{\begin{matrix}
        f(p) & \text{if }p\in P_{\partial G},\\
        0 & \text{if }p\in P_{G^\circ},
    \end{matrix}\right.\label{DefZeroExtIota0Map}\\
    [\iota_1(g)](p)=\left\{\begin{matrix}
        0 & \text{if }p\in P_{\partial G},\\
        g(p) & \text{if }p\in P_{G^\circ},
    \end{matrix}\right.\label{DefZeroExtIota1Map}
\end{gather}
for each $f\in\mathcal{F}(P_{\partial G},\mathbb{K}),$ $g\in \mathcal{F}(P_{G^\circ},\mathbb{K}),$ and every $p\in P_G$. 

The fundamental relationship between the restriction and zero extension maps is described in the next lemma.
\begin{lemma}\label{LemFundRelRestrZeroExtMaps}
    The zero extension maps $\iota_0:\mathcal{F}(P_{\partial G},\mathbb{K})\rightarrow \mathcal{F}(P_{G},\mathbb{K})$ and $\iota_1:\mathcal{F}(P_{G^\circ},\mathbb{K})\rightarrow \mathcal{F}(P_{G},\mathbb{K})$ [as defined in (\ref{DefZeroExtIota0Map}) and (\ref{DefZeroExtIota1Map})] are bounded linear operators and isometries whose adjoints are the restriction maps $|_{P_{G^\circ}}:\mathcal{F}(P_G,\mathbb{K})\rightarrow\mathcal{F}(P_{G^\circ},\mathbb{K})$ and $|_{P_{\partial G}}:\mathcal{F}(P_G,\mathbb{K})\rightarrow\mathcal{F}(P_{\partial G},\mathbb{K})$ [as defined in (\ref{DefRestrictionMapInteriorNodes}) and (\ref{DefRestrictionMapBoundaryNodes})]. In particular,
    \begin{gather}
        \iota_0^*=|_{P_{\partial G}},\;\iota_1^*=|_{P_{G^\circ}}.\label{LemFundRelRestrZeroExtMapsAdjointRel}
    \end{gather}
    Moreover,
    \begin{align}
         \ran \iota_0&=H_0,\;\ran \iota_1=H_1,\\
         |_{P_{\partial G}}\circ \iota_0&=I_{\mathcal{F}(P_{\partial G},\mathbb{K})},\;\iota_0\circ |_{P_{\partial G}}=\Gamma_{H_0},\label{LemFundRelRestrZeroExtMapsLeftInvH0}\\
        |_{P_{G^\circ}}\circ \iota_1&=I_{\mathcal{F}(P_{G^\circ},\mathbb{K})},\;\iota_1\circ |_{P_{G^\circ}}=\Gamma_{H_1},\label{LemFundRelRestrZeroExtMapsLeftInvH1}\\
        \iota_0^+&=|_{P_{\partial G}},\;\iota_1^+=|_{P_{G^\circ}},
    \end{align}
    where $(\cdot)^+$ is the Moore-Penrose pseudoinverse (see Def.\ \ref{DefMP}).
\end{lemma}
\begin{proof}
    First, it is obvious that the zero extension maps (\ref{DefZeroExtIota0Map}) and (\ref{DefZeroExtIota1Map}) are well-defined linear operators which implies they are bounded since the Hilbert space $\mathcal{F}(P_G,\mathbb{K})$ is finite-dimensional. It also follows immediately from their definitions that $\ran \iota_0=H_0,\;\ran \iota_1=H_1$. Second, for any $f_1,f_2,f\in \mathcal{F}(P_{\partial G},\mathbb{K})$ and $u\in \mathcal{F}(P_G,\mathbb{K})$ we have
    \begin{gather*}
        (\iota_0(f_1),\iota_0(f_2))_{P_{G}}=\sum_{p\in P_{G}}\overline{\iota_0(f_1)(p)}\iota_0(f_2)(p)\\
        =\sum_{p\in P_{\partial G}}\overline{f_1(p)}f_2(p)=(f_1,f_2)_{P_{\partial G}},\\
        (u|_{P_{\partial G}},f)_{P_{\partial G}}=\sum_{p\in P_{\partial G}}\overline{u|_{P_{\partial G}}(p)}f(p)=\sum_{p\in P_{\partial G}}\overline{u(p)}f(p)\\
        =\sum_{p\in P_{G}}\overline{u(p)}\iota_0(f)(p)=(u,\iota_0(f))_{P_{G}},\\
        [|_{P_{\partial G}}\circ \iota_0(f)](p)=[\iota_0(f)|_{P_{\partial G}}](p)\\
        =[\iota_0(f)](p)=f(p),\;\forall p\in P_{\partial G},\\
        [\iota_0\circ |_{P_{\partial G}}(u)](p)=[\iota_0(u|_{P_{\partial G}})](p)=\left\{\begin{matrix}
        u(p) & \text{if }p\in P_{\partial G},\\
        0 & \text{if }p\in P_{G^\circ},
    \end{matrix}\right.\\
        =[\Gamma_{H_0}(u)](p),\;\forall p\in P_{G},
    \end{gather*}
    which proves that $\iota_0$ is an isometry, the adjoint relations $\iota_0^*=|_{P_{\partial G}}$ hold, $|_{P_{\partial G}}\circ \iota_0=I_{\mathcal{F}(P_{\partial G},\mathbb{K})},$ and $\iota_0\circ |_{P_{\partial G}}=\Gamma_{H_0}$. It now follows from this with $X=\iota_0, Y=|_{P_{\partial G}}$ that the four Penrose equations $YXY=Y,XYX=X,(YX)^*=YX,(XY)^*=XY$ are satisfied so that $X^+=Y$ (by Def.\ \ref{DefMP}). The rest of the proof follows similarly.
\end{proof}

\begin{theorem}\label{ThmDtNMapSchurComp}
    Let $(P_G,E_G,\sigma)$ be an electrical network (see Def. \ref{DefCondElecNet}) with boundary and interior nodes, $P_{\partial G}$ and $P_{G^\circ}$, respectively [i.e., satisfying (\ref{AssDtNMapGraphBddIntNodes})]. Then 
    \begin{align}
        K_{\sigma}^*=K_{\sigma}\geq 0,
    \end{align}
    the discrete Dirichlet problem (see Def.\ \ref{DefDiscreteDirchProb}) has a solution for each $f\in\mathcal{F}(P_{\partial G},\mathbb{K})$, and the DtN map $\Lambda_\sigma$ (see Def.\ \ref{DefDtNMap}) is given by the formula
    \begin{gather}
        \Lambda_\sigma=\iota_0^*K_{\sigma}/(K_{\sigma})_{11}\iota_0,\label{ThmDtNMapSchurCompFormula}\\
        K_{\sigma}/(K_{\sigma})_{11}=(K_{\sigma})_{00}-(K_{\sigma})_{01}(K_{\sigma})_{11}^{+}(K_{\sigma})_{10},\label{ThmDtNMapSchurCompFormulaPart2}
    \end{gather}
    where $\iota_0:\mathcal{F}(P_{\partial G},\mathbb{K})\rightarrow H_0$ is the zero extension map [as defined in (\ref{DefZeroExtIota0Map})] and $K_{\sigma}/(K_{\sigma})_{11}$ is the generalized Schur complement of the Kirchhoff operator $K_{\sigma}=[(K_{\sigma})_{ij}]_{i,j=0,1}$ [as defined in (\ref{DefKirchhoffOperator}), (\ref{DefDtNMapSchurComplOrthogDecomposition}), and (\ref{DefKirchhoffOpBlocks})] with respect to $(K_{\sigma})_{11}$.
\end{theorem}

\begin{proof}
    Assume the hypotheses. First, by hypothesis $\sigma^*=\sigma\geq 0$ and $D^*=-D^{\bullet}$ by Lemma \ref{LemAdjFin}, hence
    \begin{align}
        K_\sigma=-D^{\bullet}\sigma D=D^*\sigma D
    \end{align} 
    from which it follows that $K_\sigma^*=K_\sigma\geq 0$. Next, by Lemma \ref{LemFundRelRestrZeroExtMaps} it follows that the discrete Dirichlet problem (Def.\ \ref{DefDiscreteDirchProb}) is equivalent to the constrained linear system of equations
    \begin{gather}
        K_{\sigma}(\iota_0(f)+\iota_1(g))=\iota_0(\phi),\\
        u|_{P_{\partial G}}=f,\; u|_{P_{G^\circ}}=g,\;(K_\sigma u)|_{P_{\partial G}}=\phi,
    \end{gather}
    or, equivalently, in block form
    \begin{gather}\label{RefBlockDisDir}
        \begin{bmatrix}
            (K_{\sigma})_{00} & (K_{\sigma})_{01}\\
            (K_{\sigma})_{10} & (K_{\sigma})_{11}
        \end{bmatrix}
        \begin{bmatrix}
           \iota_0(f)\\
           \iota_1(g)
        \end{bmatrix}
        =
        \begin{bmatrix}
            \iota_0(\phi)\\
            0
        \end{bmatrix}.
    \end{gather}
    Thus, by Proposition \ref{PropgscFact}, Lemma \ref{LemGenConstEq}, and Lemma \ref{LemFundRelRestrZeroExtMaps}, the above system is solvable for any $f\in\mathcal{F}(P_{\partial G},\mathbb{K})$ and $\iota_0(\phi)$ is given in terms of the generalized Schur complement (\ref{ThmDtNMapSchurCompFormula}) by
    \begin{gather}
        \iota_0(\phi)=K_{\sigma}/(K_{\sigma})_{11}\iota_0(f)
    \end{gather}
    from which it follows by (\ref{LemFundRelRestrZeroExtMapsAdjointRel}) and (\ref{LemFundRelRestrZeroExtMapsLeftInvH0}) in Lemma \ref{LemFundRelRestrZeroExtMaps} that
        \begin{gather}
        \phi=(\iota_0^*\circ\iota_0)(\phi)=[\iota_0^*K_{\sigma}/(K_{\sigma})_{11}\iota_0](f).
    \end{gather}
    As $f$ was an arbitrary element of $\mathcal{F}(P_{\partial G},\mathbb{K})$, this proves the identity (\ref{ThmDtNMapSchurCompFormula}), which completes the proof of the theorem. 
\end{proof}

\begin{remark}
    In the matrix setting, the well known (see, Refs.\ \onlinecite{00CM, 21AS}) response matrix of the electric network $(P_G,E_G,\sigma)$ would be $\begin{bmatrix}
        \Lambda_\sigma
    \end{bmatrix}_{\alpha_{P_{\partial G}}}^{\alpha_{P_{\partial G}}}$, the matrix representation of the DtN map $\Lambda_\sigma$ with respect to the orthonormal basis $\alpha_{P_{\partial G}}$ for $\mathcal{F}(P_{\partial G},\mathbb{K})$.
\end{remark}

\begin{corollary}\label{CorConditionsForUniqueSolvDiricProb}
 Let $(P_G,E_G,\sigma)$ be an electrical network (see Def. \ref{DefCondElecNet}) with boundary and interior nodes, $P_{\partial G}$ and $P_{G^\circ}$, respectively [i.e., satisfying (\ref{AssDtNMapGraphBddIntNodes})]. If $G$ is a connected graph and $\sigma$ is invertible then $(K_{\sigma})_{11}$ is invertible and the discrete Dirichlet problem (see Def.\ \ref{DefDiscreteDirchProb}) has a unique solution for each $f\in\mathcal{F}(P_{\partial G},\mathbb{K})$.
\end{corollary}
\begin{proof}
    Assume the hypotheses. Then $\ker D=\operatorname{span}\{1_{P_G}\}$ by Lemma \ref{LemConKer}. As $(K_{\sigma})_{11}\in \mathcal{L}(H_1)$ and $\dim H_1<\infty$, it suffices to prove that $\ker (K_{\sigma})_{11}=\{0\}$ in order to prove $(K_{\sigma})_{11}$ is invertible. Let $u\in \ker (K_{\sigma})_{11}$. As $K_{\sigma}^*=K_{\sigma}\geq 0$ and $\mathcal{F}(P_G,\mathbb{K})=H_0\ho H_1$, then it follows from Proposition \ref{PropgscFact} that $u\in \ker (K_{\sigma})_{01}$ and hence $K_{\sigma}u=0$. As $K_{\sigma}= (\sigma^{1/2}D)^*(\sigma^{1/2}D)$, this implies $u\in \ker D$ so that $u=c1_{P_G}$ for some scalar $c$. As $u\in H_1=\ker(|_{P_{\partial G}})$, we have $0=u|_{P_{\partial G}}=c1_{P_{\partial G}}$ and since $P_{\partial G}\not=\emptyset$ then $c=0$. This proves $\ker (K_{\sigma})_{11}=\{0\}$ and hence $(K_{\sigma})_{11}$ is invertible. The rest of the proof follows immediately from this and the proof of Theorem \ref{ThmDtNMapSchurComp}.
\end{proof}

\subsubsection{Effective Operator Representation}\label{sec:EffOpReprOfDtNMap}

In this section we will derive effective operator representation for the DtN map in terms of a $Z$-problems. To do so, we need will need the following discrete analog of the Hodge decomposition for the Dirichlet problem in the continuum\cite{40WH} in which the gradient $\nabla$, divergence $\nabla\cdot$, and Laplacian $\Delta$ are replaced by their discrete analogs $D, D^{\bullet},$ and  $D^{\bullet}D$, respectively.

\begin{theorem}\label{ThmKHodge}
    The sets $\mathcal{U},\;\mathcal{E},\;\mathcal{J}$ defined by 
	\begin{align}
		\mathcal{U} & = \{D u : u \in \mathcal{F}(P_G, \mathbb{K}), (D^{\bullet}D u)|_{P_{G^{\circ}}}=0\} \label{Uspace}, \\
		\mathcal{E} & = \{ D u : u \in \mathcal{F}(P_G, \mathbb{K}) \text{ and } u|_{P_{\partial G}} = 0\}, \label{Espace}                           \\
		\mathcal{J} & = \operatorname{ker}(D^{\bullet}), \label{Jspace}                                  
	\end{align}
	are mutually orthogonal subspaces in the Hilbert space $ \mathcal{F}(E_G, \mathbb{K})$. Furthermore, 
	\begin{gather} \mathcal{F}(E_G, \mathbb{K})=\mathcal{U} \overset{\perp}{\oplus} \mathcal{E} \ho \mathcal{J}, \label{orthotripleDtNZproblem}\\
	\operatorname{ran}(D)= \mathcal{U} \ho \mathcal{E}.
	\end{gather}
\end{theorem}
\begin{proof}
The proof follows immediately from the adjoint relation $D^*=-D^{\bullet}$ (by Lemma \ref{LemAdjFin}) together with the abstract Hodge decomposition theorem (i.e., Theorem \ref{ThmHodgeDecomp}) in which
	\begin{gather}
		\mathcal{B}=\mathcal{C}=\FSpace{E_G}{\mathbb{K}},\; \mathcal{A}=\FSpace{P_G}{\mathbb{K}},\\ T^*=T=\Gamma_{\ker D^\bullet},\;
		U=D\Gamma_{H_1},\; U^*=-\Gamma_{H_1}D^\bullet,\\ 
		\ran T^*=\ker D^\bullet=\mathcal{J},\; \ran U=\ran (D\Gamma_{H_1})=\mathcal{E},\\
		\ker (T^*T+UU^*)=\ran D\cap \ker (\Gamma_{H_1}D^\bullet)=\mathcal{U}.
	\end{gather}
This completes the proof.
\end{proof}

We now give the definition of the $Z$-problem and effective operator associated with this Hodge decomposition, the discrete Dirichlet problem, and the DtN map.

\begin{definition}[Dirichlet $Z$-problem and effective operator]\label{DefKZProb}
	The (Dirichlet) $Z$-problem $(\mathcal{H},\mathcal{U},\mathcal{E},\mathcal{J},\sigma)$ associated with the Hilbert space $\mathcal{H}=\mathcal{F}(E_G, \mathbb{K})$ and the orthogonal triple decomposition of $\mathcal{H}$ in \eqref{orthotripleDtNZproblem} is defined as follows: given $V_0\in\mathcal{U}$, find triples $(I_0, V, I)\in \mathcal{U} \times \mathcal{E} \times \mathcal{J}$ satisfying 
	\begin{align}
	    I_0 + I = \sigma (V_0 + V), \label{ZproblemDtN}
	\end{align} 
	such a triple $(I_0, V, I)$ is called a solution to the Z-problem at $V_0$. If there exists an operator $\sigma_* \in \mathcal{L}(\mathcal{U})$ such that
	\begin{align}
		I_0={\sigma}_*V_0,
	\end{align}
	whenever $V_0\in\mathcal{U}$ and $(I_0,V,I)$ is a solution of the $Z$-problem at $V_0$, then $\sigma_*$ is called an effective operator of this $Z$-problem.
\end{definition}

\begin{theorem}\label{ThmKEffOp}
    Let $(P_G,E_G,\sigma)$ be an electrical network (see Def. \ref{DefCondElecNet}) with boundary and interior nodes, $P_{\partial G}$ and $P_{G^\circ}$, respectively [i.e., satisfying (\ref{AssDtNMapGraphBddIntNodes})]. Then the Dirichlet $Z$-problem (\ref{DefZProbEq}) has a solution for each $V_0\in \mathcal{U}$ and it is given by the formulas  
    \begin{gather}
	   I_0=\sigma_*V_0,\label{ThmDirichZProbSolpPart1}\\
	   V=-\sigma_{11}^+\sigma_{10}V_0+K,\;K\in\ker \sigma_{11},\label{ThmDirichZProbSolpPart2}\\
	   I=\sigma_{20}V_0+\sigma_{21}V,\label{ThmDirichZProbSolpPart3}\\
        \sigma_*=\begin{bmatrix}
		\sigma_{00} & \sigma_{01}\\
		\sigma_{10} & \sigma_{11}
		\end{bmatrix}/\sigma_{11}=\sigma_{00}-\sigma_{01}\sigma_{11}^+\sigma_{10},\label{ThmDirichZProbSolpPart4}
\end{gather}
where the $3\times 3$ block operator matrix $\sigma=[\sigma_{ij}]_{i,j=0,1,2}$ is with respect to the orthogonal triple decomposition of $\mathcal{F}(E_G, \mathbb{K})$ in \eqref{orthotripleDtNZproblem} [see Sec.\ \ref{sec:ClassicalResults}, Eq.(\ref{DefOfSigmaSubblocks}) for $\sigma_{ij}$ formulas]. Moreover, the effective operator of the $Z$-problem exists, is unique, and is given by the generalized Schur complement formula (\ref{ThmDirichZProbSolpPart4}).
\end{theorem}
\begin{proof}
    Assume the hypotheses. Then $\sigma^*=\sigma\geq 0$ and hence the result follows immediately from Corollary \ref{CorSuffCondForTrueThmGenClassicalDiriMinPrin}.
\end{proof}

In order to state the main result of this section, that relates the DtN map $\Lambda_{\sigma}$ to the effective operator $\sigma_*$ of the Dirichlet $Z$-problem, we require the following definition and lemma.

\begin{definition}[The lift operator]\label{DefLiftOp}
    For a connected graph $G$, the lift operator
    \begin{align}
        \Pi\in\mathcal{L}(\mathcal{F}(P_{\partial G},\mathbb{K}),\mathcal{U}),
    \end{align} 
    where $\mathcal{U}$ is defined by (\ref{Uspace}),
    is the map defined by
    \begin{gather}\label{DefLiftOpEq}
        \Pi(f)=Du,\;\;\text{for every } f\in\mathcal{F}(P_{\partial G},\mathbb{K}),
    \end{gather}
    where $u$ is the unique solution to the discrete Dirichlet problem (\ref{DefDisDirProb}), (\ref{DefDisBoundCond}) with conductivity $\sigma=I_{\mathcal{F}(E_G, \mathbb{K})}$ [i.e., the identity operator on $\mathcal{F}(E_G, \mathbb{K})$] and $u|_{P_{\partial G}}=f$.
\end{definition}

\begin{lemma}
    If $G$ is a connected graph, then the lift operator is well defined.
\end{lemma}
\begin{proof}
    The result follows immediately from Corollary \ref{CorConditionsForUniqueSolvDiricProb} applied to the electrical network $(P_G,E_G,I_{\mathcal{F}(E_G, \mathbb{K})})$ under the assumption that $G$ is a connected graph.
\end{proof}

We now prove the main result of this section.
\begin{theorem}\label{DtNZprob}  Let $(P_G,E_G,\sigma)$ be an electrical network (see Def. \ref{DefCondElecNet}) with boundary and interior nodes, $P_{\partial G}$ and $P_{G^\circ}$, respectively [i.e., satisfying (\ref{AssDtNMapGraphBddIntNodes})]. If $G=(P_G,E_G)$ is a connected graph then the DtN map $\Lambda_{\sigma}$ (in Def.\ \ref{DefDtNMap}), the effective operator $\sigma_*$ of the Dirichlet $Z$-problem (in Def.\ \ref{DefKZProb}), and the lift operator $\Pi$ (in Def.\ \ref{DefLiftOp}) satisfy the identity
    \begin{align}
	    \Lambda_{\sigma} = \Pi^*\sigma_*\Pi.
	\end{align}
\end{theorem}

\begin{proof}
    Assume the hypotheses. Let $f\in\mathcal{F}(P_{\partial G},\mathbb{K})$. Then $\Pi (f)=Du_f\in\mathcal{U}$, where $u_f$ is the unique solution to the discrete Dirichlet problem (i.e., the problem in Definition \ref{DefKZProb}) with conductivity $\sigma=I_{\mathcal{F}(E_G, \mathbb{K})}$ satisfying $u_f|_{P_{\partial G}}=f$. By Theorem \ref{ThmKEffOp} there is a solution to the $Z$-problem $(I_0,V,I)\in\mathcal{U}\times\mathcal{E}\times\mathcal{J}$ at $V_0=Du_f$, i.e.,
    \begin{gather}
        I_0+I=\sigma(V_0+V).
    \end{gather}
    In particular, since $V\in\mathcal{E}$ there exists a $v\in\mathcal{F}(P_G,\mathbb{K})$ such that $v|_{P_{\partial G}}=0$ and $V=Dv$.
    We claim that $u=u_f+v$ is a solution to the discrete Dirichlet problem for conductivity $\sigma$ satisfying $u|_{P_{\partial G}}=f$. To prove this, notice that
    \begin{gather}
        u|_{P_{\partial G}}=(u_f+v)|_{P_{\partial G}}=f+0=f
    \end{gather}
    and since $I\in\mathcal{J}=\ker D^\bullet$ and $I_0\in\mathcal{U}$ then
    \begin{gather}
        0=(D^\bullet I_0)|_{P_{G^\circ}}=(D^\bullet (I_0+I))|_{P_{G^\circ}}=(D^\bullet\sigma (V_0+V))|_{P_{G^\circ}}\\=(D^\bullet\sigma D(u_f+v))|_{P_{G^\circ}}=(D^\bullet\sigma D(u_f+v))|_{P_{G^\circ}}.
    \end{gather}
    This proves the claim. It follows from this, the definition of the DtN map $\Lambda_{\sigma}$ (in Def.\ \ref{DefDtNMap}), and Theorem \ref{ThmDtNMapSchurComp} that  $\Lambda_\sigma f=\Lambda_{\sigma}(u|_{P_{\partial G}})=(-D^\bullet\sigma Du)|_{P_{\partial G}}$. Thus,
    \begin{gather*}
        (f,\Pi^*\sigma_*\Pi f)_{\mathcal{F}(P_{\partial G},\mathbb{K})}=(\Pi f, \sigma_*\Pi f)_{\mathcal{F}(E_G,\mathbb{K})}\\=(Du_f, \sigma_*Du_f)_{\mathcal{F}(E_{G},\mathbb{K})}=(V_0+V, \sigma(V_0+V))_{\mathcal{F}(E_{G},\mathbb{K})}\\
       =(D(u_f+v), \sigma D(u_f+v))_{\mathcal{F}(E_{G},\mathbb{K})}\\=(u_f+v, -D^\bullet\sigma D(u_f+v))_{\mathcal{F}(E_{G},\mathbb{K})}\\
       =(u_f+v, (-D^\bullet\sigma D)(u_f+v))_{\mathcal{F}(P_{G},\mathbb{K})}\\=((u_f+v)|_{P_{\partial G}},(-D^\bullet\sigma Du)|_{P_{\partial G}})_{\mathcal{F}(P_{\partial G},\mathbb{K})}\\
       =(f,\Lambda_\sigma f)_{\mathcal{F}(P_{\partial G},\mathbb{K})}.
    \end{gather*}
    As this equality holds for all $f\in\mathcal{F}(P_{\partial G},\mathbb{K})$ and both $\Pi^*\sigma_*\Pi$ and $\Lambda_\sigma$ are self-adjoint (by Theorems \ref{ThmDtNMapSchurComp} and \ref{ThmKEffOp}), it follows that $\Lambda_\sigma=\Pi \sigma_*\Pi$ (by uniqueness of quadratic forms for self-adjoint operators\cite{80JW}). This completes the proof.
\end{proof}
	
	\subsection{Effective Conductivity and Effective Resistance}\label{SubSectDtNeffcond}
	
	In this section, $(P_G,E_G,\sigma)$ will denote an electrical network (see Def.\ \ref{DefCondElecNet}) with node set $P_G$ containing least two nodes. Let $p,q$ denote any two nodes in the graph $G=(P_G,E_G)$, i.e.,
	\begin{align}
	    p,q\in P_G,\;p\not=q.
	\end{align}
	Define their delta functions by
	\begin{gather}
	\delta_p, \delta_q\in \mathcal{F}(P_G,\mathbb{K}),\\
	    \forall x\in P_G,\;\;\delta_p(x)=
        \begin{cases}
            1,\;\text{if}\; x=p,\\
            0,\;\text{if}\; x\not=p,
        \end{cases}\;\;
        \delta_q(x)=
        \begin{cases}
            1,\;\text{if}\; x=q,\\
            0,\;\text{if}\; x\not=q.
        \end{cases}
	\end{gather}
	The following is a well-known definition from electrical network theory\cite{49FM, 61FM, 97NB, 09JP}.
	
	\begin{definition}[Effective conductivity/resistance]\label{DefEffCondEffRes}
	    A scalar $\sigma_{eff}=\sigma_{eff}(p,q)\in \mathbb{K}$ is called an effective conductivity [of the electrical network $(P_G,E_G,\sigma)$ with respect to $p,q$]  if
	    \begin{gather}
	        j=\sigma_{eff}[u(p)-u(q)],
	    \end{gather}
	    whenever $j\in\mathbb{K}$ and $u\in \mathcal{F}(P_G,\mathbb{K})$ is a solution of the equation
	    \begin{align}
	        (-D^\bullet \sigma D)u=j(\delta_p-\delta_q)\label{DefEffCondProb}.
	    \end{align}
    It's inverse $r_{eff}=\frac{1}{\sigma_{eff}}$ is called the effective resistance (with the convention $r_{eff}=\infty$, if $\sigma_{eff}=0$).
	\end{definition}

\subsubsection{Effective operator representation of the effective conductivity}
    The goal of this section is to derive an effective operator representation for the effective conductivity $\sigma_{eff}$ in terms of an associated $Z$-problem. To do this, we need the next lemma and definition.
	\begin{lemma}\label{LemHodEffCond}
	
		The sets $\;\mathcal{U},\;\mathcal{E},$ $\mathcal{J}$ defined by
		\begin{align}
			\mathcal{U} & =\{u\in \mathcal{F}(P_G,\mathbb{K}):u(x)=0,\forall\; x\in P_G\setminus\{p\}\},\\
			\mathcal{E} & = \{u\in \mathcal{F}(P_G,\mathbb{K}):u(p)=u(q)=0\}, \\
			\mathcal{J} & =\{u\in \mathcal{F}(P_G,\mathbb{K}):u(x)=0,\forall\; x\in P_G\setminus\{q\}\},                 
		\end{align}
		are mutually orthogonal subspaces in the Hilbert space $\mathcal{F}(P_G,\mathbb{K})$.
		Furthermore,
		\begin{gather}
			\mathcal{F}(P_G,\mathbb{K})=\mathcal{U}\ho \mathcal{E}\ho \mathcal{J},\label{LemHodEffCond1}\\
			\mathcal{U}=\operatorname{span}\{\delta_p\},\;
			\mathcal{J}=\operatorname{span}\{\delta_q\},\;
			\mathcal{E}=\operatorname{span}\{\delta_p\}^{\perp} \cap \operatorname{span}\{\delta_q\}^{\perp}.
		\end{gather}
		Moreover, the orthogonal projections of $\mathcal{F}(P_G,\mathbb{K})$ onto $\mathcal{U},\;\mathcal{J},\;\mathcal{E}$ are the following left multiplication operators, respectively:
		\begin{gather}\label{EffCondResistanceGammaOpsExplicit}
			\Gamma_0=\delta_p I_{\mathcal{F}(P_G,\mathbb{K})},\;\;
		    \Gamma_2=\delta_q I_{\mathcal{F}(P_G,\mathbb{K})},\;\;
			\Gamma_1=(1-\delta_p-\delta_q)I_{\mathcal{F}(P_G,\mathbb{K})}.
		\end{gather}
	\end{lemma}
	
	\begin{proof}
		The proof follows immediately by the abstract Hodge decomposition (i.e., Theorem \ref{ThmHodgeDecomp}) in which
		\begin{gather}
			\mathcal{A}=\mathcal{B}=\mathcal{C}=\FSpace{P_G}{\mathbb{C}},\; T^*=T=\Gamma_0,\; U^*=U=\Gamma_2,\\
			\ran T^*=\operatorname{span}\{\delta_p\}=\mathcal{U},\; \ran U=\operatorname{span}\{\delta_q\}=\mathcal{J},\;\\
			\ker (T^*T+UU^*)=\operatorname{span}\{\delta_p\}^{\perp}\cap\operatorname{span}\{\delta_q\}^{\perp}=\mathcal{E}.
		\end{gather}
	\end{proof}

	\begin{definition}\label{DefDipProb}
	The $Z$-problem $(\mathcal{H},\mathcal{U},\mathcal{E},\mathcal{J},K_\sigma)$ associated with the Hilbert space $\mathcal{H}=\mathcal{F}(P_G, \mathbb{K})$ and the orthogonal triple decomposition of $\mathcal{H}$ in \eqref{LemHodEffCond1} and with the Kirchhoff operator $K_\sigma=-D^\bullet \sigma D\in\mathcal{L}(\mathcal{F}(P_G,\mathbb{K}))$ (in  Def. \ref{DefCondElecNet}) is defined as follows: given $v_0\in\mathcal{U}$, find triples $(\rho_0, v, \rho)\in \mathcal{U} \times \mathcal{E} \times \mathcal{J}$ satisfying 
	\begin{align}
	    \rho_0 + \rho = K_\sigma (v_0 + v), \label{ZproblemEffConductivityResistance}
	\end{align} 
	such a triple $(\rho_0, v, \rho)$ is called a solution to the $Z$-problem at $v_0$. If there exists an operator $(K_\sigma)_* \in \mathcal{L}(\mathcal{U})$ such that
	\begin{align}
		\rho_0=(K_\sigma)_*v_0,
	\end{align}
	whenever $v_0\in\mathcal{U}$ and $(\rho_0,v,\rho)$ is a solution of the $Z$-problem at $v_0$, then $(K_\sigma)_*$ is called an effective operator of this $Z$-problem.
\end{definition}

\begin{theorem}\label{ThmKEffOpZProb}
The solutions of the $Z$-problem (in Def.\ \ref{DefDipProb}) at each $v_0\in \mathcal{U}$ are given by the formulas  
    \begin{gather}
	   \rho_0=(K_\sigma)_*v_0,\label{ThmKEffOpZProbPart1}\\
	   v=-(K_\sigma)_{11}^+(K_\sigma)_{10}v_0+\kappa,\;\kappa\in\ker (K_\sigma)_{11},\label{ThmKEffOpZProbPart2}\\
	   \rho=(K_\sigma)_{20}v_0+(K_\sigma)_{21}v,\label{ThmKEffOpZProbPart3}\\
        \hspace{-2.25em}(K_\sigma)_*=\begin{bmatrix}
		(K_\sigma)_{00} & (K_\sigma)_{01}\\
		(K_\sigma)_{10} & (K_\sigma)_{11}
		\end{bmatrix}/(K_\sigma)_{11}
		=(K_\sigma)_{00}-(K_\sigma)_{01}(K_\sigma)_{11}^+(K_\sigma)_{10},\label{ThmKEffOpZProbPart4}
\end{gather}
where the $3\times 3$ block operator matrix $K_\sigma=[(K_\sigma)_{ij}]_{i,j=0,1,2}$ is with respect to the orthogonal triple decomposition of $\mathcal{F}(P_G, \mathbb{K})$ in \eqref{LemHodEffCond1}, i.e.,
\begin{gather}
        K_{\sigma}=[(K_\sigma)_{ij}]_{i,j=0,1,2},\\
        \mathcal{L}(H_j,H_i)\ni(K_{\sigma})_{ij}=\Gamma_iK_{\sigma}\Gamma_j:H_j\rightarrow H_i,\;i,j=0,1,2,\label{ThmKEffOpZProbPart5}\\
        H_0=\mathcal{U},  H_1=\mathcal{E},  H_2=\mathcal{J},
    \end{gather}
where $\Gamma_i$, $i=0,1,2$ are given explicitly by formulas (\ref{EffCondResistanceGammaOpsExplicit}). Moreover, the effective operator of the $Z$-problem (in Def.\ \ref{DefDipProb}) exists, is unique, and is given by the generalized Schur complement formula (\ref{ThmKEffOpZProbPart4}).
\end{theorem}
\begin{proof}
    Assume the hypotheses. Then $\sigma^*=\sigma\geq 0, K_{\sigma}=D^*\sigma D$ which implies $(K_{\sigma})^*=K_{\sigma}\geq 0$ and hence the result follows immediately from Corollary \ref{CorSuffCondForTrueThmGenClassicalDiriMinPrin}.
\end{proof}
	
We now prove the main result of this section.	
\begin{theorem}\label{ThmEffCondIsEffOpFiniteGraph}
The effective conductivity $\sigma_{eff}$ [of the electrical network $(P_G,E_G,\sigma)$ with respect to $p,q$] (see Def.\ \ref{DefEffCondEffRes}) exists, is unique, and is given by the formula
    \begin{align}
	    \sigma_{eff}I_{\mathcal{U}}=(K_\sigma)_*,
	\end{align}
where $I_{\mathcal{U}}$ is the identity operator on $\mathcal{U}$ and $(K_\sigma)_*$ is the effective operator of the $Z$-problem in Def.\ \ref{DefDipProb}.
\end{theorem}
\begin{proof}
    Assume the hypotheses. First, by Theorem \ref{ThmKEffOpZProb}, the effective operator $(K_\sigma)_*$ of the $Z$-problem (in Def. \ref{DefDipProb}) exists and is unique.
    Let $c\in\mathbb{K}\setminus\{0\}$. Then $v_0=c\delta_p\in\mathcal{U}$ and so by Theorem \ref{ThmKEffOpZProb} there exists a solution $(\rho_0,v,\rho)\in\mathcal{U}\times\mathcal{E}\times\mathcal{J}$ to the $Z$-problem (in Def. \ref{DefDipProb}) satisfying
	\begin{gather*}
	    c_1\delta_p-c_2\delta_q=\rho_0+\rho=K_\sigma(c\delta_p+v),
	\end{gather*}
	where $\rho_0=c_1\delta_p$ and $\rho=-c_2\delta_q$ for some $c_1,c_2\in\mathbb{K}$
	and hence 
	\begin{gather*}
	    (K_\sigma)_*(c\delta_p)=c_1\delta_p.
	\end{gather*}
	or, equivalently,
	\begin{gather}
	    (K_{\sigma})_*=\frac{c_1}{c}I_{\mathcal{U}}\label{PrfCRatio}.
	\end{gather}
    It also follows,
    \begin{gather*}
        c_1\delta_p-c_2\delta_q\in\ran K_\sigma=(\ker K_\sigma)^\perp,
    \end{gather*}
    and, in particular, since $1_{P_G}\in\ker K_\sigma$ (by Lemma \ref{LemConKer}) then
    \begin{gather*}
        0=(1_{P_G}, c_1\delta_p-c_2\delta_q)_{P_G}=c_1( 1_{P_G}, \delta_p)_{P_G}-c_2(1_{P_G}, \delta_q)_{P_G}\\
        =c_1-c_2.
    \end{gather*}
Hence,
    \begin{gather*}
        c_1(\delta_p-\delta_q)=K_\sigma(c\delta_p+v),
    \end{gather*}
    and $c\delta_p+v$ is a solution of linear eq.\ \eqref{DefEffCondProb} with the scalar $c_1$. Thus, either $c_1=0$, in which case $u=0$ is a solution of linear eq.\ \eqref{DefEffCondProb} with the scalar $j=0$ or $c_1\not =0$, in which case $\frac{jc}{c_1}\delta_p+\frac{j}{c_1}v$ is a solution of linear eq.\ \eqref{DefEffCondProb} for any scalar $j\in \mathbb{K}$.

    Next, suppose $u\in \mathcal{F}(P_G,\mathbb{K})$ is a solution to eq.\ \eqref{DefEffCondProb} with scalar $j$. We claim that
    \begin{gather*}
        \frac{c_1}{c}[u(p)-u(q)]=j.
    \end{gather*}
    First, we have
	\begin{align*}
		K_\sigma(v_0'+v')=K_\sigma(u-u(q)1_{P_G})=K_\sigma u = j(\delta _p - \delta_q)=\rho_0'+\rho',
	\end{align*}
	where
	\begin{align*}
		  & v_0'=[u(p)-u(q)]\delta_p\in\mathcal{U},\;\rho_0'=j\delta _p\in \mathcal{U},  \\
		  & v'=u-u(q)1_{P_G}-[u(p)-u(q)]\delta_p\in\mathcal{E},\;\;\rho'=-j\delta_q\in\mathcal{J}. 
	\end{align*}
	In particular, $(\rho_0',v',\rho')\in \mathcal{U}\times \mathcal{E}\times \mathcal{J}$ is a solution to the $Z$-problem $(\mathcal{F}(P_G,\mathbb{C}), \mathcal{U},\mathcal{E},\mathcal{J},K_\sigma)$ (in Def. \ref{DefDipProb}) at $v_0'$. It follows this, Theorem \ref{ThmKEffOpZProb}, and \eqref{PrfCRatio} that
	\begin{align*}
	   \frac{c_1}{c}[u(p)-u(q)]\delta_p=\frac{c_1}{c}v_0'=(K_\sigma)_*v_0'=\rho_0'=j\delta_p,
	\end{align*}
	which proves the claim. This also proves that $\sigma_{eff}$ exists, is unique, and
	\begin{gather*}
	    \sigma_{eff}=\frac{c_1}{c}.
	\end{gather*}
	Therefore, it follows from this and \eqref{PrfCRatio} that $\sigma_{eff}I_{\mathcal{U}}=(K_\sigma)_*$. This completes the proof.
\end{proof}

The next result provides insight into the role that connectedness plays in the notion of effective conductivity.
\begin{corollary}\label{cor:EffConductivityAndConnectedness}
    Let $\sigma_{eff}=\sigma_{eff}(p,q)$ be the effective conductivity of an electrical network $(P_G,E_G,\sigma)$ with respect to $p,q$ (see Def.\ \ref{DefEffCondEffRes}). Then the following statements are true:
    \begin{itemize}
        \item[(a)] If $p$ and $q$ belong to distinct connected components of the graph $G=(P_G,E_G)$ then $\sigma_{eff}=0$.
        \item[(b)] $\sigma_{eff}=0$ if and only if $\exists\; V\in \operatorname{ran} D\cap \ker \sigma$ such that $V=Du$ for some $u\in \mathcal{F}(P_G,\mathbb{K})$ with $u(p)\not=u(q)$.
        \item[(c)] If $\sigma$ is invertible, then $\sigma_{eff}\not=0$ if and only if $p$ and $q$ belong to the same connected component of the graph $G=(P_G,E_G)$.
    \end{itemize}
\end{corollary}
\begin{proof}
Assume the hypotheses. First, Theorem \ref{ThmEffCondIsEffOpFiniteGraph} gives us existence and uniqueness of the effective conductivity $\sigma_{eff}$. 

Let us prove (a). Assume that $p$ and $q$ belong to distinct connected components of the graph $G=(P_G,E_G)$. Then there exists a connected component $G_i=(P_{G_i},E_{G_i})$ of $G$ such that $p\in P_{G_i}$ and $q\not\in P_{G_i}$. Third, by Lemma \ref{LemConKer} we know that $\chi_{P_{G_i}}\in \ker D$, where $\chi_{P_{G_i}}$ is the characteristic function on the node set $P_{G_i}$. From this it follows that $u=\chi_{P_{G_i}}\in \mathcal{F}(P_G,\mathbb{K})$ with $u(p)=1,u(q)=0$ and $K_{\sigma}u=0=0(\delta_p-\delta_q)$ so that we must have (with $j=0$) that $0=\sigma_{eff}[u(p)-u(q)]=\sigma_{eff}$. This proves (a).

Next, we will prove (b). Suppose there exists a $V\in \operatorname{ran} D\cap \ker \sigma$ such that $V=Du$ for some $u\in \mathcal{F}(P_G,\mathbb{K})$ with $u(p)\not=u(q)$. Then $K_{\sigma}u=0=0(\delta_p-\delta_q)$ implying by Theorem \ref{ThmEffCondIsEffOpFiniteGraph} that $0=\sigma_{eff}[u(p)-u(q)]$ and hence as $u(p)-u(q)\not=0$ then $\sigma_{eff}=0$. Conversely, suppose that $\sigma_{eff}=0$. Then it follows from Theorem \ref{ThmKEffOpZProb} with $v_0=\delta_p\in \mathcal{U}$ that there exists a $(\rho_0,v,\rho)\in \mathcal{U}\times \mathcal{E}\times \mathcal{J}$ such that $\rho_0+\rho =K_{\sigma}(v_0+v)$ and so by Theorem \ref{ThmEffCondIsEffOpFiniteGraph} we have $\rho_0=\sigma_{eff}v_0=0$. As $1_{P_G}\in \ker D$, hence $1_{P_G}\in \ker K_{\sigma}=(\operatorname{ran}K_{\sigma})^{\perp}$ and $\rho\in \mathcal{J}=\operatorname{span}\{\delta_q\}$ it follows that $c\delta_q=\rho\in \operatorname{ran}K_{\sigma}$ for some $c\in \mathbb{K}$ so that $c=(1_{P_G},\rho)_{\mathcal{F}(P_G,\mathbb{K})}=0$. From this we conclude that $u=v_0+v$ satisfies $0=K_{\sigma}u$ implying $0=(u,K_{\sigma}u)_{\mathcal{F}(P_G,\mathbb{K})}=(Du,\sigma Du)_{\mathcal{F}(E_G,\mathbb{K})}$ so that since $\sigma^*=\sigma\geq 0$ then it follows that $\sigma Du=0$, that is, $Du\in\operatorname{ran}D\cap \ker \sigma$. Finally, since $v\in\mathcal{E}$ then $v(p)=v(q)=0$ and hence $u(p)=v_0(p)+v(p)=1\not=0=v_0(q)+v(q)=u(q)$. This proves the converse. Therefore, (b) is true.

Finally, we prove (c). Assume $\sigma$ is invertible. Then by (a) we know that if $p$ and $q$ belong to distinct connected components of the graph $G=(P_G,E_G)$ then $\sigma_{eff}=0$. Conversely, suppose $\sigma_{eff}=0$. Then $\operatorname{ran} D\cap \ker \sigma=\{0\}$ and so by (b) there exists $u\in\mathcal{F}(P_G,\mathbb{K})$ such that $Du=0$ and $u(p)\not=u(q)$. In particular, $u\in\ker D$ so by Lemma \ref{LemConKer} it follows that $u=\sum_{i=l}^kc_l\chi_{P_{G_l}}$ for some $c_1,\ldots, c_k\in \mathbb{K}$, where $G_1=(P_{G_1},E_{G_1}),\ldots,G_k=(P_{G_k},E_{G_k})$ are the connected components of the graph $G=(P_G,E_G)$ and $\chi_{P_{G_l}}$ is the characteristic function on the node set $P_{G_l}$ of the graph $G_l$ for each $l=1,\ldots, k$. Hence, there exists an $i,j\in \{1,\ldots,k\}$ such that $p\in P_{G_i}, q\in P_{G_j}$ implying $c_i=u(p)\not=u(q)=c_j$ so that $i\not= j$. This proves that $p$ and $q$ belong to distinct connected components of the graph $G=(P_G,E_G)$. Therefore, (c) is true.
\end{proof}

\begin{example}
 The following is a counterexample to the converse of Corollary \ref{cor:EffConductivityAndConnectedness}.(a). Let $G=(P_G,E_G)$ with $P_G=\{p_1,p_2,p_3\}, E_{G}=\{e_1,e_2\}$, $|P_G|=3, |E_G|=2$, $(e_1)_-=p_1, (e_1)_+=p_2, (e_2)_-=p_2, (e_2)_+=p_2$. Then $G$ is a connected graph and so trivially $p=p_1,q=p_2$ belong to the same connected component of the graph $G$ (i.e., $k=1, G_1=G$ in Lemma \ref{LemConKer}). Consider the electrical network $(P_G,E_G,\sigma)$, where $\sigma = \chi_{e_2} I_{\mathcal{F}(E_G,\mathbb{K})}\in \mathcal{L}(\mathcal{F}(E_G,\mathbb{K}))$ is left-multiplication by the characteristic function $\chi_{e_2}$ of the edge $e_2$, i.e., for each $e\in E_G$, $\chi_{e_2}(e)=1$ if $e=e_2$ and $\chi_{e_2}(e)=0$ if $e=e_1$. Then it follows that $\sigma^*=\sigma\geq 0$ and $\sigma D\delta_p=0$, since $(\sigma D\delta_p)(e)=\chi_{e_2}(e)(D\delta_p)(e)=\chi_{e_2}(e)[\delta_p(e_+)-\delta_p(e_-)]=0$ for each $e\in E_G$. This proves that $D\delta_p\in   \operatorname{ran} D\cap \ker \sigma$ with $\delta_p\in \mathcal{F}(P_G,\mathbb{K})$ and $\delta_p(p)=1\not=0=\delta_p(q)$ so that, by Corollary \ref{cor:EffConductivityAndConnectedness}.(b), we have $\sigma_{eff}(p,q)=\sigma_{eff}=0$.
\end{example}

\subsubsection{Relationship to the DtN map}

The goal of this section is to derive a representation of the effective conductivity $\sigma_{eff}$ of the electrical network $(P_G,E_G,\sigma)$ with respect to distinct nodes $p,q\in P_G$ in terms of the DtN map $\Lambda_{\sigma}$ as introduced in Sec.\ \ref{SubSectDtNZ} when the boundary nodes $P_{\partial G}$ and interior nodes $P_{G^{\circ}}$ are defined to be
\begin{align}
    P_{\partial G}=\{p,q\},\;P_{G^{\circ}}=P_G\setminus P_{\partial G},\label{def:BddIntNodesForDtNMapRelToEffCond}
\end{align}
and assuming that $\{p,q\}\subsetneq P_G$. 

This goal is achieved with the following theorem.
\begin{theorem}
    Let $(P_G,E_G,\sigma)$ be an electrical network. If $p,q\in P_G$ with $p\not=q$ and $\{p,q\}\subsetneq P_G$ then the effective conductivity $\sigma_{eff}=\sigma_{eff}(p,q)$ (see Def.\ \ref{DefEffCondEffRes}) and the DtN map $\Lambda_{\sigma}$ [with boundary and interior nodes defined by (\ref{def:BddIntNodesForDtNMapRelToEffCond})] (see Def.\ \ref{DefDtNMap}) are related by the formula
    \begin{align}
        \Lambda_{\sigma}(\cdot) = \sigma_{eff} ((\delta_p-\delta_q)|_{P_{\partial G}},\cdot)_{\mathcal{F}(P_{\partial G},\mathbb{K})} (\delta_p-\delta_q)|_{P_{\partial G}}.
    \end{align}
\end{theorem}
\begin{proof}
    Assume the hypotheses. First, Theorems \ref{ThmEffCondIsEffOpFiniteGraph} and \ref{ThmDtNMapSchurComp} give us existence and uniqueness of both the effective conductivity $\sigma_{eff}$ and the DtN map $\Lambda_{\sigma}$ for this electrical network $(P_G,E_G,\sigma)$ with boundary and interior nodes (\ref{def:BddIntNodesForDtNMapRelToEffCond}). Next, let $f\in \mathcal{F}(P_{\partial G},\mathbb{K})$. Then by Theorem \ref{ThmDtNMapSchurComp} there exists a solution $u\in \mathcal{F}(P_G,\mathbb{K})$ to the discrete Dirichlet problem (see Def.\ \ref{DefDiscreteDirchProb}):
    \begin{align*}
        (K_{\sigma}u)|_{P_{G^\circ}}=0,\;u|_{P_{\partial G}}=f,
    \end{align*}
    and
    \begin{align*}
        \Lambda_{\sigma}(f)=(K_{\sigma}u)|_{P_{\partial G}}.
    \end{align*}
    Now it follows from this and the two facts
    \begin{gather*}
        \mathcal{F}(P_{\partial G},\mathbb{K})=\operatorname{span}\{\delta_p|_{P_{\partial G}},\delta_q|_{P_{\partial G}}\},\\
        \;1_{P_G}|_{P_{\partial G}}\in\ker \Lambda_{\sigma}=(\operatorname{ran}\Lambda_{\sigma})^{\perp},
    \end{gather*}
    that we have
    \begin{align*}
         \Lambda_{\sigma}(f)=j\delta_p|_{P_{\partial G}}-j\delta_q|_{P_{\partial G}}=j(\delta_p-\delta_q)|_{P_{\partial G}},
    \end{align*}
    for some scalar $j\in\mathbb{K}$, which implies that
    \begin{align*}
        K_{\sigma}u=j(\delta_p-\delta_q)|_{P_{\partial G}}.
    \end{align*}
    Hence, by this and Theorem \ref{ThmEffCondIsEffOpFiniteGraph}, we have
    \begin{align*}
        j=\sigma_{eff}[u(p)-u(q)].
    \end{align*}
    Thus, we have
    \begin{gather*}
        \Lambda_{\sigma}(f)=j(\delta_p-\delta_q)|_{P_{\partial G}}
        =\sigma_{eff}[u(p)-u(q)](\delta_p-\delta_q)|_{P_{\partial G}}\\
        = \sigma_{eff} ((\delta_p-\delta_q)|_{P_{\partial G}},f)_{\mathcal{F}(P_{\partial G},\mathbb{K})} (\delta_p-\delta_q)|_{P_{\partial G}}.
    \end{gather*}
    Therefore, as this is true for all $f\in \mathcal{F}(P_{\partial G},\mathbb{K})$, this proves the theorem.
\end{proof}

\subsection{Conductivity for Periodic Linear Graphs}\label{SectCondInf}
	
	In this section, we study the discrete analog of the continuum electrical conductivity problem and effective conductivity $\sigma_*$ (cf.\, Example \ref{ExContinuumPeriodicCondZProb}) for periodic linear graphs. Although we are still in a discrete setting, we must treat new complications that arise when one moves from finite linear graphs to infinite linear graphs. 
	
	In order to gain a clear insight into this transition, we only consider Cartesian graphs. This is also motivated by Refs.\ \onlinecite{15MG, 90BG}, which study the effective conductivity in lattice networks, but use analytic and probabilistic methods from Ref.\ \onlinecite{83GP} as opposed to the Hilbert space approach that we are advocating for in our paper. Nevertheless, we believe that our methods are complementary to theirs and as such may find applications in the theory of composites for lattice networks.
    
    Now there are some similarities and some changes that occur when transitioning from Example \ref{ExContinuumPeriodicCondZProb} and Section \ref{SubSectOpView}. First, we continue to use the function spaces $\mathcal{F}(\mathcal{T},\mathbb{K})$, but now they are infinite-dimensional vector spaces without an inner product. Also, the gradient $D$ and divergence $D^\bullet$ are still well defined linear operators, but we lose their adjoint relationship. Instead, we replace these with their periodic counterparts $\mathcal{F}_\#(\mathcal{T},\mathbb{K})$, $D_\#$, and $(D^\bullet)_\#$. Now, the function space $\mathcal{F}_\#(\mathcal{T},\mathbb{K})$ is a finite-dimensional Hilbert space with inner product $(\cdot,\cdot)_{\mathcal{F}_\#(\mathcal{T},\mathbb{K})}$ and we have the adjoint relationship $D_\#^*=-(D^\bullet)_\#$. From this we are able to get a Hodge decomposition and an associated $Z$-problem in this setting, in analogy with the continuum problem under these periodic spaces and operators. In addition, we obtain the analogous relation that the effective conductivity is the effective operator $\sigma_*$ under certain additional hypotheses, but without which there are some interesting nuances that we discuss.
	
	\subsubsection{Operator Framework for Electrical Conductivity in Lattice Networks}
	Fix a $d$-tuple of positive integers
	\begin{align}
	    \tau=(\tau_1,\ldots,\tau_d)\in\mathbb{N}^d\label{ref:latticevecperiodicitytau}
	\end{align}
	and consider the lattice 
	\begin{align}\label{ref:latticevec}
		L=\left\{\sum_{i=1}^dn_ia_i:n_1,\ldots,n_d\in\mathbb{Z}\right\}\subseteq \mathbb{Z}^d,
	\end{align}
	spanned by the $d$ primitive lattice vectors
	\begin{align}
		\{a_i:a_i=\tau_i\mathbf{e}_i,\;i=1,\ldots,d\}, 
	\end{align}
	where $\mathbf{e}_i$ denotes the $i$th standard basis vector, i.e., the $d$-tuple in $\mathbb{Z}^d$ whose $i$th coordinate is $1$ and all others are $0$. Consider the linear directed graph $G=(\mathbb{Z}^d,E)$ with node set $\mathbb{Z}^d$ and directed edge set
	\begin{align}\label{ref:edgeset}
		E=\{\overrightarrow{v\left(v+\mathbf{e}_i\right)}:v\in \mathbb{Z}^d, i=1,\ldots,d\}, 
	\end{align}
	that is, each edge $e=\overrightarrow{v\left(v+\mathbf{e}_i\right)}$ is a directed line segment with initial node $e_-=v$ and terminal node $e_+=v+\mathbf{e}_i$. In particular, if $e_1,e_2\in E$ then $e_1=e_2$ if and only if $(e_1)_\pm=(e_2)_\pm$. 
	
	Notice that, for each edge $e\in E$ and lattice vector $R\in L$, there is a unique edge $e+R\in E$ defined by translation (see Figure \ref{FigEdgTrans}) as
	\begin{align}\label{ref:edgeadd}
		e+R=\overrightarrow{(e_-+R)(e_++R)}\in E. 
	\end{align}
	In particular, $G$ becomes a periodic digraph under the lattice $L$.
	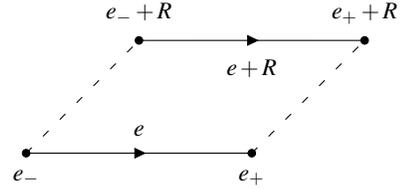
\begin{figure}[ht!]
		\centering
		\begin{circuitikz}[scale=0.75, /tikz/circuitikz/bipoles/length=1.25cm]
			\draw[loosely dashed]
			(0,0) -- (2,2)
			(4,0) -- (6,2);
			\draw
			(0,0) node[label={below:$e_-$}] {}
			to[short, i=\phantom{},bipole nodes={circ}{circ},l=$e$]
			(4,0) node[label={below:$e_+$}] {}
			                    
			(2,2) node[label={above:$e_-+R$}] {}
			to[short, i=\phantom{},bipole nodes={circ}{circ},l_=$e+R$]
			(6,2) node[label={above:$e_++R$}] {}
			;
		\end{circuitikz}
		\caption{For a given edge $e\in E$ in the periodic digraph $G$ and a lattice vector $R\in L$, the translation $e+R$ is another edge in $E$.}
		\label{FigEdgTrans}
	\end{figure}
	\begin{definition}[Primitive unit cell]\label{DefPrimCell}
		A primitive unit cell (fundamental domain) with respect to the lattice $L$ is a subset $\mathcal{S}\subseteq \mathcal{T}\in\{\mathbb{Z}^d, E\}$ satisfying the properties
			\begin{enumerate}[(a)]
				\item $\mathcal{S}\neq \emptyset$,
				\item $\mathcal{T}=\mathcal{S}+L$,
				\item $(s,R)\in \mathcal{S}\times(L\setminus\{0\})\Rightarrow s+R\not\in \mathcal{S}$.
			\end{enumerate}
	\end{definition}
	In particular, a primitive unit cell $P$ for the nodes $\mathbb{Z}^d$ is given by
	\begin{align}
		P=\mathbb{Z}^d\cap\prod_{k=1}^d[0,\tau_k)=\{p_1,\ldots,p_{|P|}\}\label{def:primitive_unit_cell_P}. 
	\end{align}
	Similarly, a primitive unit cell $E_P^-$ for the edges $E$ is given by
	\begin{align}
		E_P^-=\{e\in E: e_-\in P\}=\{e_1,\ldots,e_{|E_P^-|}\}.\label{def:primitive_unit_cell_E} 
	\end{align}
	An illustrated example of this construction can be found in Figure \ref{FigLatEx}.
	\begin{figure}[ht!]
		\centering
		\begin{subfigure}[t]{0.5\textwidth}
			\centering
			\begin{circuitikz}[scale=0.5, /tikz/circuitikz/bipoles/length=1.5cm]
				\draw [black,thick] (-0.5,0)--(8.5,0);
				\draw [blue, very thick] (2,0)--(4,0);
				\draw [black, dashed] (2,0)--(1,-2)
				(4,0)--(5,-2)
				;
				\draw [blue]
				(1,-2) node[] {}
				to[short, i=\phantom{},bipole nodes={ocirc, fill=red}{ocirc, fill=red}]
				(3,-2) node[] {}
				to[short, i=\phantom{},bipole nodes={}{ocirc}]
				(5,-2) node[] {}
				;
			\end{circuitikz}
			\subcaption{$d=1$, $\tau=2$.}
		\end{subfigure}%
		\vspace{\floatsep}
		\begin{subfigure}[t]{0.5\textwidth}
			\centering
					\begin{circuitikz}[scale=0.5, /tikz/circuitikz/bipoles/length=1.5cm]
				\draw [step=1.0, black, thick] (0.5,0.5) grid (5.5,5.5);
				\draw [blue, very thick] (2,2)--(2,3)
				(3,2)--(2,2) (2,3)--(3,3) (3,3)--(3,2);
				\draw [blue, very thick] (2,3)--(2,4)
				 (3,3)--(3,4) (3,3)--(4,3) (3,2)--(4,2);
				\draw [black, dashed] (2,4)--(6,4.5);
				\draw [black, dashed] (2,2)--(6,1.5);
				\draw [draw=blue]
				(6,1.5) node[] {}
				to[short, i=\phantom{},bipole nodes={ocirc, fill=red}{ocirc, fill=red}]
				(7.5,1.5) node[] {}
				to[short, i=\phantom{},bipole nodes={ocirc, fill=red}{ocirc, fill=red}]
				(9,1.5) node[ocirc] {}
				                    
				(6,3) node[] {}
				to[short, i=\phantom{},bipole nodes={ocirc, fill=red}{ocirc, fill=red}]
				(7.5,3) node[] {}
				to[short, i=\phantom{},bipole nodes={ocirc, fill=red}{ocirc, fill=red}]
				(9,3) node[ocirc] {}
				                    
				(6,1.5)
				to[short, i=\phantom{},bipole nodes={ocirc, fill=red}{ocirc, fill=red}]
				(6,3)
				to[short, i=\phantom{},bipole nodes={ocirc, fill=red}{ocirc, fill=red}]
				(6,4.5) node[ocirc] {}
				                    
				(7.5,1.5)
				to[short, i=\phantom{},bipole nodes={ocirc, fill=red}{ocirc, fill=red}]
				(7.5,3)
				to[short, i=\phantom{},bipole nodes={ocirc, fill=red}{ocirc, fill=red}]
				(7.5,4.5) node[ocirc] {};
			\end{circuitikz}
			\subcaption{$d=2$, $\tau=(2,2)$.}
		\end{subfigure}
		\caption{Two illustrations showing the periodic digraph $G$ [in dimensions (a) $d=1$ and (b) $d=2$ and periods (a) $\tau=\tau_1=2$ and (b) $\tau=(\tau_1,\tau_2)=(2,2)$] is generated by translation of a unit cell with node unit cell $P$ (\textcolor{red}{red}) and edge unit cell $E_P^-$ (\textcolor{blue}{blue}).}\label{FigLatEx}
	\end{figure}
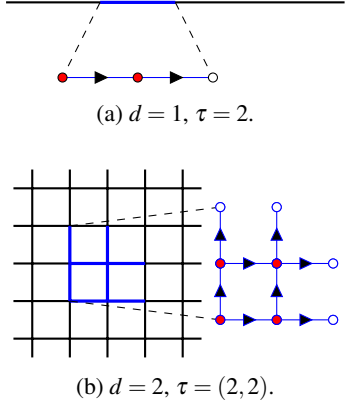

	Observe, the cardinalities of both of these primitive unit cells are nonzero and finite, i.e., 
	\begin{align}\label{ref:finitecard}
		0<|P|<\infty \text{ and } 0<|E_P^-|<\infty, 
	\end{align}
	and the in-degree $\deg_{in}(\cdot)$ and out-degree $\deg_{out}(\cdot)$ of each node in the graph is finite and equal, i.e.,
	\begin{align}\label{ref:inout}
		0<\deg_{out}(p)=\sum_{\begin{subarray}{c}e\in E, \\e_-=p\end{subarray}}1=\sum_{\begin{subarray}{c}e\in E,\\e_+=p\end{subarray}}1=\deg_{in}(p)<\infty,
	\end{align}
	for each $p\in \mathbb{Z}^d$.
	\begin{definition}\label{def:LPerSpace}
		Let $\mathcal{T}\in\{\mathbb{Z}^d,E\}$. We define the subspace of $\mathcal{F}(\mathcal{T},\mathbb{K})$, containing all periodic functions with respect to a lattice $L$, by
		\begin{equation}\label{DefLPerSpace}
			\hspace{-1em}\mathcal{F}_{\#}(\mathcal{T},\mathbb{K})=\left\{f\in \mathcal{F}(\mathcal{T},\mathbb{K}):f(t+R)=f(t),\;\forall(t,R)\in \mathcal{T}\times L\right\}.
		\end{equation}
	\end{definition}
	
	A key difference between the space $\mathcal{F}(\mathcal{T},\mathbb{K})$ and $\mathcal{F}_\#(\mathcal{T},\mathbb{K})$ is that the former is finite-dimensional only if $\mathcal{T}$ is finite-dimensional, whereas the latter is always finite-dimensional, a consequence of (\ref{ref:finitecard}) and the next lemma.
	\begin{lemma}\label{LemExtRes}
		Let $(\mathcal{S},\mathcal{T})\in\{(P,\mathbb{Z}^d),(E_P^-,E)\}$ with lattice $L$.
		For each $f\in \mathcal{F}(\mathcal{S},\mathbb{K})$, define its periodic extension $f^{\#}\in \mathcal{F}_\#(\mathcal{T},\mathbb{K})$ by
		\begin{align}
			f^{\#}(s+R)=f(s),\;\;\forall (s,R)\in \mathcal{S}\times L. 
		\end{align}
		Then the periodic extension map, i.e., the function
		\begin{align}
			(\cdot)^{\#}:\mathcal{F}(\mathcal{S},\mathbb{K})\to \mathcal{F}_\#(\mathcal{T},\mathbb{K}),
		\end{align}
		is well defined and, moreover, it is an invertible linear transformation with inverse the restriction map
		\begin{align}
			(\cdot)|_{\mathcal{S}}:\mathcal{F}_\#(\mathcal{T},\mathbb{K})\to \mathcal{F}(\mathcal{S},\mathbb{K}), 
		\end{align}
		defined, for each $f\in \mathcal{F}_\#(\mathcal{T},\mathbb{K})$, by
		\begin{align}
			f|_{\mathcal{S}}(s)=f(s), \forall s\in \mathcal{S}. 
		\end{align}
		In particular, the following vector spaces are isomorphic:
		\begin{align}
			\mathcal{F}(\mathcal{S},\mathbb{K})\cong \mathcal{F}_\#(\mathcal{T},\mathbb{K})\cong \mathbb{K}^{|\mathcal{S}|}. 
		\end{align}
	\end{lemma}
	
	\begin{proof}
		We begin by proving that the periodic extension map is well defined. By Definition \ref{DefPrimCell}.$(b)$ we have $\mathcal{T}=\mathcal{S}+L$, that is for each $t\in \mathcal{T}$ there exist $(s,R)\in \mathcal{S}\times L$ such that $t=s+R$. Second, by Definition \ref{DefPrimCell}.$(c)$ if $s\in \mathcal{S}$ and $0\neq R\in L$ then $s+L\not\in \mathcal{S}$. Now suppose, $t,t'\in \mathcal{T}$ then $t=s+R$ for some $(s,R)\in \mathcal{S}\times L$ and $t'=s'+R'$ for some $(s',R')\in \mathcal{S}\times L$. Hence, if $t=t'$ then $s+R=s'+R'$ implying $s=s'+(R'-R)$ which by Definition \ref{DefPrimCell}.$(c)$ further implies $R'-R=0$ so that $s=s'$. Together with $\mathcal{T}=\mathcal{S}+L$, this proves the well-definedness of the map $(\cdot)^\#:\mathcal{F}(\mathcal{S},\mathbb{K})\to\mathcal{F}(\mathcal{T},\mathbb{K})$. 
		    
		Next we will prove that if $f\in \mathcal{F}(\mathcal{S},\mathbb{K})$ then $f^\#\in\mathcal{F}_\#(\mathcal{T},\mathbb{K})$. Let $f\in\mathcal{F}(\mathcal{S},\mathbb{K})$ and $(t,R)\in \mathcal{T}\times L$, then there is an $(s,R')\in \mathcal{S}\times L$ such that $t=s+R'$ and so $f^\#(t+R)=f^\#(s+R'+R)=f^\#(s)=f^\#(s+R')=f^\#(t)$ which proves $f^\#\in\mathcal{F}_\#(\mathcal{T},\mathbb{K})$. This proves the well-definedness of the map $(\cdot)^\#:\mathcal{F}(\mathcal{S},\mathbb{K})\to\mathcal{F}_\#(\mathcal{T},\mathbb{K})$.
		    
		We will now prove that the map $(\cdot)^\#:\mathcal{F}(\mathcal{S},\mathbb{K})\to\mathcal{F}_\#(\mathcal{T},\mathbb{K})$ is linear. First, both $\mathcal{F}(\mathcal{S},\mathbb{K})$ and $\mathcal{F}_\#(\mathcal{T},\mathbb{K})$ are vectors spaces over the field $\mathbb{K}$ with addition and scalar multiplication for functions as usual. Next, let $f,g\in\mathcal{F}(\mathcal{S},\mathbb{K})$ and $(c,s,R)\in \mathbb{K}\times \mathcal{S}\times L$. Then
		\begin{gather*}
			(cf+g)^\#(s+R)=(cf+g)(s)\\
			=cf(s)+g(s)
			=cf^\#(s+R)+g^\#(s+R)\\=(cf^\#+g^\#)(s+R).
		\end{gather*}
		This proves by Definition \ref{DefPrimCell}.$(b)$ that $(cf+g)^\#=cf^\#+g^\#$. Hence, the periodic extension map is linear.
		    
		Next, we will prove that the restriction map is well defined. Let $f\in\mathcal{F}_\#(\mathcal{S},\mathbb{K})$ and $s\in \mathcal{S}$. Then, as $f:\mathcal{T}\to\mathbb{K}$ and $\mathcal{S}\subseteq \mathcal{T}$, it follows $g(s)\in \mathbb{K}$ implying $f|_{\mathcal{S}}:\mathcal{S}\to\mathbb{K}$ is a well-defined function in $\mathcal{F}(\mathcal{S},\mathbb{K})$. Suppose $f=g$ for some $g\in\mathcal{F}(\mathcal{S},\mathbb{K})$. Then $f(s)=g(s)$ $\forall s\in \mathcal{T}$ and, in particular, $f(s)=g(s)$ $\forall s\in \mathcal{S}$ which implies $f|_{\mathcal{S}}=g|_{\mathcal{S}}$. This proves $(\cdot)|_{\mathcal{S}}:\mathcal{F}_\#(\mathcal{T},\mathbb{K})\to\mathcal{F}(\mathcal{S},\mathbb{K})$ is well defined.
		    
		Lastly, we prove that the periodic extension map and the restriction map are inverses of each other. Let $f\in\mathcal{F}_\#(\mathcal{T},\mathbb{K})$. Then for each $t\in \mathcal{T}$ there exists $(s,R)\in \mathcal{S}\times L$ with $t=s+R$ and so
		\begin{gather*}
			\{[(\cdot)^\#\circ (\cdot)|_{\mathcal{S}}](f)\}(t)=(f|_{\mathcal{S}})^\#(t)=(f|_{\mathcal{S}})^\#(s+R)\\
			=f|_{\mathcal{S}}(s)=f(s)=f(s+R)=f(t).
		\end{gather*}
		This implies that $[(\cdot)^\#\circ (\cdot)|_{\mathcal{S}}](f)=f$. As $f\in\mathcal{F}_\#(\mathcal{T},\mathbb{K})$ was arbitrary it follows that $(\cdot)^\#\circ (\cdot)|_{\mathcal{S}}=I_{\mathcal{F}_\#(\mathcal{T},\mathbb{K})}$. Now let $f\in\mathcal{F}(\mathcal{S},\mathbb{K})$ and $s\in \mathcal{S}$. Then
		\begin{gather*}
			\{[(\cdot)|_{\mathcal{S}}\circ (\cdot)^\#](f)\}(s)=f^\#|_{\mathcal{S}}(s)=f^\#(s)=f^\#(s+0)=f(s),
		\end{gather*}
		which implies $[(\cdot)|_{\mathcal{S}}\circ (\cdot)^\#](f)=f$. Since $f\in\mathcal{F}(\mathcal{S},\mathbb{K})$ was arbitrary it follows that $(\cdot)|_{\mathcal{S}}\circ (\cdot)^\#=I_{\mathcal{F}(\mathcal{S},\mathbb{K})}$. Therefore, we have proven that the periodic extension map and restriction map are inverses of each other. The remainder of the proof now follows immediately. 
	\end{proof}
	
	It follows by Proposition \ref{PropHilbertSpace} and Lemma \ref{LemExtRes} that, for  $(\mathcal{S},\mathcal{T})\in\{(P,\mathbb{Z}^d),(E_P^-,E)\}$, a basis for $\mathcal{F}(\mathcal{S},\mathbb{K})$ is given by $\alpha_{\mathcal{S}}$ defined by (\ref{PropOrthBas}). Consequently, we can apply $(\cdot)^\#$ to the elements of $\alpha_{\mathcal{S}}$ to obtain a basis $\alpha_{\mathcal{S}}^\#$ for $\mathcal{F}_\#(\mathcal{T},\mathbb{K})$, i.e.,
	\begin{gather}\label{ref:perbasis}
		\alpha_{\mathcal{S}}^\#=\left\{\delta_{s_i}^\#: \delta_{s_i}\in \alpha_{\mathcal{S}}\right\},\quad
		\delta_{s_i}^\#(s_j+R)=\delta_{ij},
	\end{gather}
	for all $R\in L$ and $i,j=1,\ldots,|\mathcal{S}|$ [where $\delta_{ij}$ is the Kronecker delta, see (\ref{DefKroneckerDelta})]. 
	The following lemma is a corollary of Proposition \ref{PropHilbertSpace} and Lemma \ref{LemExtRes}.
	\begin{lemma}\label{lem:findimh}
		Let  $(\mathcal{S},\mathcal{T})\in\{(P,\mathbb{Z}^d),(E_P^-,E)\}$ with lattice $L$. Then the vector space $\mathcal{F}_\#(\mathcal{T},\mathbb{K})$ [defined in (\ref{DefLPerSpace})] is finite-dimensional and
		\begin{align}
			0<\dim \mathcal{F}_\#(\mathcal{T},\mathbb{K})=|\mathcal{S}|<\infty. 
		\end{align}
		Furthermore, every $f\in \mathcal{F}_\#(\mathcal{T},\mathbb{K})$ is determined by its evaluation on $\mathcal{S}$, that is,
		\begin{align}
			f=\sum_{i=1}^{|\mathcal{S}|}f_i\delta_{s_i}^\#, 
		\end{align}
		where
		\begin{align}
			f_i=f(s_i),\quad i=1,\ldots,|\mathcal{S}|. 
		\end{align}
		Moreover,
		$\mathcal{F}_\#(\mathcal{T},\mathbb{K})$ is a finite-dimensional Hilbert space with inner product
		\begin{align}\label{Ref:iprod}
			(f,g)_{\mathcal{F}_\#(\mathcal{T},\mathbb{K})}=\sum_{i=1}^{|\mathcal{S}|}\overline{f_i}g_i,\;\;\forall f,g\in\mathcal{F}_\#(\mathcal{T},\mathbb{K}), 
		\end{align}
		and orthornormal basis $\alpha_{\mathcal{S}}^\#$ [defined in (\ref{ref:perbasis})].
	\end{lemma}
	
	\begin{proof}
		Assume the hypotheses. By Lemma \ref{LemExtRes}, the periodic extension map is an isomorphism from $\mathcal{F}(\mathcal{S},\mathbb{K})$ to $\mathcal{F}_\#(\mathcal{T},\mathbb{K})$ and since $\dim \mathcal{F}(\mathcal{S},\mathbb{K})=|\mathcal{S}|$ this implies $\dim \mathcal{F}_\#(\mathcal{T},\mathbb{K})=|\mathcal{S}|$. Moreover, the periodic extension map takes the basis of $\mathcal{F}(\mathcal{S},\mathbb{K})$ to a basis for $\mathcal{F}_\#(\mathcal{T},\mathbb{K})$. In particular, $\alpha_\mathcal{S}\mapsto\alpha_{\mathcal{S}}^\#$ under the periodic extension map. The rest of the corollary now follows immediately from this.
	\end{proof}
	
	Similarly to Definition \ref{DefAnalGradDiv}, the following definitions of $D$ and $D^{\bullet}$ are the discrete analogs of the gradient $\nabla$ and divergence $\nabla\cdot$ in the continuum. After this, we define their restrictions to periodic functions.

\begin{definition}\label{DefAnalGradDivLattice}
	We define the function $D:\mathcal{F}(\mathbb{Z}^d,\mathbb{K})\to\mathcal{F}(E,\mathbb{K})$ by
	\begin{equation}\label{DefDGradLattice}
		(Df)(e)=f(e_+)-f(e_-),\;\forall (f,e)\in \mathcal{F}(\mathbb{Z}^d,\mathbb{K})\times E,
	\end{equation}
	and the function $D^\bullet:\mathcal{F}(E,\mathbb{K})\to\mathcal{F}(\mathbb{Z}^d,\mathbb{K})$ by
	\begin{equation}\label{DefDDivLattice}
		(D^\bullet f)(p)=\sum_{\begin{subarray}{c}e\in E,\\e_-=p\end{subarray}}f(e)-\sum_{\begin{subarray}{c}e\in E,\\e_+=p\end{subarray}}f(e),\;\forall(f,p)\in \mathcal{F}(E,\mathbb{K})\times \mathbb{Z}^d.
	\end{equation}
\end{definition}

	\begin{definition}
		We define the function $D_\#:\Hzero\to\Hone$ by
		\begin{align}
			(D_\#f)(e)=(Df)(e) \text{ for all } (f,e)\in \Hzero\times E, \label{DefDGradLatticePeriodic}
		\end{align}
		and the function $(D^\bullet)_\#:\Hone\to\Hzero$ by
		\begin{align}
			[(D^\bullet)_\#g](p)=(D^\bullet g)(p),\;\text{for all}\;(g,p)\in \Hone\times \mathbb{Z}^d. \label{DefDDivLatticePeriodic}
		\end{align}
	\end{definition}
	    
	The following lemma shows that $D,\;D_\#,\; D^\bullet$ and $(D^\bullet)_\#$ are still well defined in this setting.
	
	\begin{lemma}\label{LemDBulletWellDefLinearOp}
		The functions $D$, $D_\#$, $D^\bullet$ and $(D^\bullet)_\#$ are well defined linear operators.
	\end{lemma}
	\begin{proof}
		Clearly, $\Hzero$ and $\Hone$ are subspaces of the vector spaces $\Fzero$ and $\Fone$, respectively. It is also clear that $D$ is a well defined linear function and, since the in-degree and out-degree of every node $p\in \mathbb{Z}^d$ is finite [see (\ref{ref:inout})], $D^\bullet$ is also a well defined linear function. Thus, to prove the lemma we need only show that $\Hzero$ and $\Hone$ are invariant subspaces of $D$ and $D^\bullet$, respectively. 
		For any $f\in\Hzero$ and $R\in L$, we have
		\begin{align*}
			  (Df)(e+R)&=f((e+R)_+)-f((e+R)_-)\\
			  &=f(e_++R)-f(e_-+R) \\
			  &=f(e_+)-f(e_-)=(Df)(e),\;\forall e\in E.
		\end{align*}
		Next, let $g\in \mathcal{F}_{\#}(E,\mathbb{C})$. Then for every $p\in \mathbb{Z}^d$ and every $R\in L$, we have $-R\in L$ hence $g(e-R)=g(e)$ for every $e\in E$ and 
		\begin{gather*}
			\{e\in E:e_{\pm}=p+R\}=\{e-R\in E:(e-R)_{\pm}=p\}\\
			=\{e'\in E:(e')_{\pm}=p\},
		\end{gather*}
		which implies that 
		\begin{gather*}
			(D^{\bullet}g)(p+R)=\sum_{\begin{subarray}{c}e\in E,\\e_-=p+R\end{subarray}}g(e)-\sum_{\begin{subarray}{c}e\in E,\\e_+=p+R\end{subarray}}g(e) \\=\sum_{\begin{subarray}{c}e-R\in E,\\(e-R)_-=p\end{subarray}}g(e-R)-\sum_{\begin{subarray}{c}e-R\in E,\\(e-R)_+=p\end{subarray}}g(e-R)
			\\ =\sum_{\begin{subarray}{c}e'\in E,\\(e')_-=p\end{subarray}}g(e')-\sum_{\begin{subarray}{c}e'\in E,\\(e')_+=p\end{subarray}}g(e')
			=(D^{\bullet}g)(p).
		\end{gather*}
		This completes the proof.
	\end{proof}
	
	The following theorem is the analogy of Lemma \ref{LemAdjFin}, but the proof is more technical.
	
	\begin{theorem}\label{ThmLatAdj}
		The functions $D_\#:\Hzero\to\Hone$ and $-(D^{\bullet})_{\#}:\Hone\to\Hzero$ [defined by (\ref{DefDGradLatticePeriodic}) and (\ref{DefDDivLatticePeriodic})] are bounded linear operators and satisfy the Hilbert space adjoint relations
		\begin{equation}
			(D_{\#})^*=-(D^{\bullet})_{\#},\;[(D^{\bullet})_{\#}]^*=-D_{\#}.\label{RefAdjPerOps} 
		\end{equation}
	\end{theorem}
	\begin{proof}
		First, by Lemma \ref{LemDBulletWellDefLinearOp}, both $D_{\#}$ and $D^{\bullet}$ are linear operators. Second, by Lemma \ref{lem:findimh}, $\alpha^\#_P$ and $\alpha^\#_{E_P^-}$ are orthonormal bases of the finite-dimensional Hilbert spaces $\Hzero$ and $\Hone$, respectively.  Hence, for each $i,j$ it follows that
		\begin{gather*}
			( D_\#\delta_{p_j}^\#,\delta_{e_i}^\#)_{\Hone}
			=\sum_{k=1}^{|E_P^-|}\overline{(D_\#\delta_{p_j}^\#)_k}(\delta_{e_i}^\#)_k
			=(D_\#\delta_{p_j}^\#)_i\\
			=-\delta_{p_j}^\#[(e_i)_-]+\delta_{p_j}^\#[(e_i)_+]
			=-\delta_{p_j}[(e_i)_-]+\delta_{p_j}^\#[(e_i)_+]\\
			=-\begin{cases}1,\quad (e_i)_-= p_j\\0,\quad (e_i)_-\neq p_j\end{cases}+\begin{cases}1,\quad (e_i)_+\in p_j+L\\0,\quad (e_i)_+\not\in p_j+L\end{cases} \\=-\sum_{\begin{subarray}{c}e\in E,\\e_-=p_j\end{subarray}}\begin{cases}1,\quad e\in e_i+L\\0,\quad e\not\in e_i+L \end{cases}+\sum_{\begin{subarray}{c}e\in E,\\e_+=p_j\end{subarray}}\begin{cases}1,\quad e\in e_i+L\\0,\quad e\not\in e_i+L \end{cases}
			\\=-\sum_{\begin{subarray}{c}e\in E,\\e_-=p_j\end{subarray}}\delta_{e_i}^\#(e)+\sum_{\begin{subarray}{c}e\in E,\\e_+=p_j\end{subarray}}\delta_{e_i}^\#(e)=[-(D^\bullet)_\#\delta_{e_i}^\#]_j\\=\sum_{k=0}^{|P|-1}\overline{(\delta_{p_j}^\#)_k}[-(D^\bullet)_\#\delta_{e_i}^\#]_k=( \delta_{p_j}^\#,-(D^\bullet)_\#\delta_{e_i}^\# )_{\Hzero},
		\end{gather*}
		which implies
		\begin{gather*}
			(\delta_{p_j}^\#,-(D^\bullet)_\#\delta_{e_i}^\#)_{\Hzero}=( D_\#\delta_{p_j}^\#,\delta_{e_i}^\#)_{\Hone}.
		\end{gather*}
		This proves the theorem.
	\end{proof}
	
	\subsubsection{Lattice \textit{Z}-problem and Effective Operator}
	
		\begin{definition}[Edge cell average]\label{DefAvgEdgeUnitCell}
		 The (edge cell) average of a periodic function $f\in\mathcal{F}_\#(E,\mathbb{K})$, denoted by $\langle f \rangle_{\mathcal{F}_\#(E,\mathbb{K})}$, is the constant function
		\begin{align}\label{AvgFuncEdgeUnitCell}
			\langle f \rangle_{\mathcal{F}_\#(E,\mathbb{K})}=\left(\frac{1}{|E_P^-|}\sum_{e\in E_P^-}f(e)\right) 1_{E}\in \mathcal{F}_\#(E,\mathbb{K}).
		\end{align}
		The (edge cell) average operator $\Gamma_0\in\mathcal{L}(\mathcal{F}_\#(E,\mathbb{K}))$ is defined by
		\begin{align}\label{def:edgavgop}
			\Gamma_0(f)=\langle f \rangle_{\mathcal{F}_\#(E,\mathbb{K})},\;\forall f\in\mathcal{F}_\#(E,\mathbb{K}). 
		\end{align}
	\end{definition}
	
	\begin{theorem}\label{ThmLatDecomp}
		The sets $\mathcal{U}$, $\mathcal{E}$, $\mathcal{J}$ defined by 
		\begin{align}
			\mathcal{U}&=\{U\in\Hone:U= \langle U \rangle_{\Hone}\},\label{ThmLatDecompU}\\
			\mathcal{E}&=\{D_\#u: u\in\Hzero,\langle D_\#u\rangle_{\Hone}=0\},\label{ThmLatDecompE}\\
			\mathcal{J}&=\{J\in\Hone: (D^\bullet)_\# J=0, \langle J \rangle_{\Hone}=0\}\label{ThmLatDecompJ}
		\end{align}
		are mutually orthogonal subspaces in the Hilbert space $\Hone$. Furthermore,
		\begin{gather}
			\Hone= \mathcal{U}\ho \mathcal{E}\ho \mathcal{J},\label{LatticeHodgeDecomp}\\
			\mathcal{U}\ho\mathcal{J}=\ker (D^\bullet)_\#,\;\mathcal{U}=\ran\Gamma_0,\;\mathcal{E}=\ran D_\#.
		\end{gather}
	\end{theorem}
	\begin{proof}
		First, its obvious from the definition (Def.\ \ref{DefAvgEdgeUnitCell}) that the average operator $\Gamma_0$ is an orthogonal projection of $\Hone$ onto $\mathcal{U}$ so that, in particular, $\Gamma_0^*=\Gamma_0$ and $\ran (\Gamma_0)=\mathcal{U}$ is the set of all constant functions in $\Hone$. Second, it follows from (\ref{ref:inout}) that all these constant functions are in $\ker (D^\bullet)_\#,$ or, equivalently, $(D^\bullet)_\#\Gamma_0=0$. Hence, it follows from this and Theorem \ref{ThmLatAdj} that $0=-((D^\bullet)_\#\Gamma_0)^*=-\Gamma_0^*((D^\bullet)_\#)^*=\Gamma_0D_\#$ so that $\operatorname{ran}D_\#=\mathcal{E}$. Finally, the proof of (\ref{LatticeHodgeDecomp}) follows immediately from this and Theorem \ref{ThmLatAdj} together with the abstract Hodge decomposition (i.e., Theorem \ref{ThmHodgeDecomp}) in which
		\begin{gather*}
			\mathcal{A}=\mathcal{B}=\Hone,\; \mathcal{C}=\Hzero,\; T=-(D^\bullet)_\#,\; T^*=D_\#,\;\\ U=U^*=\Gamma_0,\;
			\ran U=\ran \Gamma_0=\mathcal{U},\;\ran (T^*)=\operatorname{ran}D_\#=\mathcal{E},\\
		\ker(T^*T+UU^*)=\ker T\cap\ker(U^*)=\ker(D^\bullet)_\#\cap \ker \Gamma_0=\mathcal{J}.
	\end{gather*}
	In particular, it follows now from (\ref{LatticeHodgeDecomp}) and Theorem \ref{ThmLatAdj} that $\mathcal{U}\ho\mathcal{J}=\mathcal{E}^{\perp}=(\operatorname{ran}D_\#)^{\perp}=\ker(D_\#)^*=\ker (D^\bullet)_\#$. This completes the proof.
	\end{proof}
	
	\begin{definition}[Lattice Z-problem and effective operator]
		\label{def:latzprob}The (lattice) $Z$-problem $(\mathcal{H},\mathcal{U},\mathcal{E},\mathcal{J},\sigma)$ is the following problem associated with the Hilbert space $\mathcal{H}=\Hone$, the orthogonal triple decomposition of $\mathcal{H}$ in (\ref{LatticeHodgeDecomp}), and a linear operator $\sigma\in \mathcal{L}(\mathcal{H})$:
		given $V_{0}\in\mathcal{U}$, find triples $\left(  I_{0},V,I\right)
		\in\mathcal{U}\times\mathcal{E}\times\mathcal{J}$ satisfying%
		\begin{equation}
			I_{0}+I=\sigma\left(V_{0}+V\right),\label{def:latzprobeq}
		\end{equation}
		such triple $\left(  I_{0},V,I\right)$ is called a solution of the
		$Z$-problem at $V_{0}$. If there exists an operator $\sigma_{\ast}\in \mathcal{L}(\mathcal{U})$ such that
		\begin{equation}
			I_{0}=\sigma_{\ast}V_{0},\label{def:latzprobeffop}
		\end{equation}
		whenever $V_{0}\in\mathcal{U}$ and $\left(  I_{0},V,I\right)$ is a solution of the $Z$-problem at
		$V_{0}$, then $\sigma_{\ast}$ is called an effective operator for the $Z$-problem.
	\end{definition}
	
		The following theorem is the main result of this section.
	\begin{theorem}\label{ThmEffCondEffOp}
    Let $(\mathbb{Z}^d,E,\sigma)$ be a lattice electrical network with conductivity $\sigma\in \mathcal{L}(\Hone)$ satisfying $\sigma^*=\sigma\geq 0$ and consider the lattice $Z$-problem $(\mathcal{H},\mathcal{U},\mathcal{E},\mathcal{J},\sigma)$ (in Def.\ \ref{def:latzprob}) with Hilbert space $\mathcal{H}=\Hone$ and orthogonal triple decomposition of $\mathcal{H}$ in (\ref{LatticeHodgeDecomp}). Then the solutions of the lattice $Z$-problem (\ref{def:latzprobeq}) at $V_0\in \mathcal{U}$ are given by the formulas 
    \begin{gather}
	   I_0=\sigma_*V_0,\label{ThmLatticeZProbSolpPart1}\\
	   V=-\sigma_{11}^+\sigma_{10}V_0+K,\;K\in\ker \sigma_{11},\label{ThmLatticeZProbSolpPart2}\\
	   I=\sigma_{20}V_0+\sigma_{21}V,\label{ThmLatticeZProbSolpPart3},\\
        \sigma_*=\begin{bmatrix}
		\sigma_{00} & \sigma_{01}\\
		\sigma_{10} & \sigma_{11}
		\end{bmatrix}/\sigma_{11}=\sigma_{00}-\sigma_{01}\sigma_{11}^+\sigma_{10},\label{ThmLatticeZProbSolpPart4}
\end{gather}
where the $3\times 3$ block operator matrix $\sigma=[\sigma_{ij}]_{i,j=0,1,2}$ is with respect to the orthogonal triple decomposition of $\mathcal{F}(E, \mathbb{K})$ in (\ref{LatticeHodgeDecomp}) [see Sec.\ \ref{sec:ClassicalResults}, Eq.(\ref{DefOfSigmaSubblocks}) for $\sigma_{ij}$ formulas]. Moreover, the effective operator of the $Z$-problem exists, is unique, and is given by the generalized Schur complement formula (\ref{ThmLatticeZProbSolpPart4}).
\end{theorem}
\begin{proof}
    Assume the hypotheses. Then $\sigma^*=\sigma\geq 0$ and hence the result follows immediately from Corollary \ref{CorSuffCondForTrueThmGenClassicalDiriMinPrin}.
\end{proof}
	
	\subsubsection{On Effective Conductivity as an Effective Operator}\label{SecEffCond}
	
	Here we define a discrete analog to Ohm's law in the continuum for the periodic electrical conductivity problem and the effective conductivity (cf.\ Example \ref{ExContinuumPeriodicCondZProb} and Remark \ref{rem:DefPreciseSpacesContinuum}).
	
	\begin{definition}[Effective conductivity]\label{def:perohmlaw}
		The periodic Ohm's law for a lattice electric network $(\mathbb{Z}^d,E,\sigma)$, where $\sigma\in\mathcal{L}(\Hone)$ (i.e., the conductivity), is the following relation between voltage $V$ and current $I$:
		\begin{gather}
			I=\sigma V,\label{def:perohmlawpart1}\\
			I\in\ker D^\bullet \cap \Hone, V\in\operatorname{ran}D\cap \Hone.\label{def:perohmlawpart2}
		\end{gather}
		If a function $\sigma_{eff}:\mathcal{U}\rightarrow \mathcal{U}$ exists such that
		\begin{align}
			\langle I\rangle_{\Hone} = \sigma_{\text{eff}}(\langle V\rangle_{\Hone}), 
		\end{align}
		whenever $I$ and $V$ satisfy (\ref{def:perohmlawpart1}) and (\ref{def:perohmlawpart2}), then $\sigma_{eff}$ is called an effective conductivity of the lattice electric network.
	\end{definition}
	
	The natural questions to ask are: $(1)$ Does an effective conductivity exist? $(2)$ If so, how is it related to the lattice $Z$-problem and the effective operator? $(3)$ If an effective operator of the lattice $Z$-problem exists, does an effective conductivity also exist? In this section, we answer all these questions: $(1)$ Not always (see Example \ref{ExEffCondLatticeZProbIsContrived}). Necessary and sufficient conditions for the existence of an effective conductivity are given (see Theorem \ref{thm:NecSuffCondEffCondExistLattice} which is essentially a corollary of Theorem \ref{thm:FundThmExistenceUniquenessSolvabilityZProbEffOpVecSp} in Appendix \ref{SectAbsTheoryCompositesVecSpFramework}). $(2)$ If it does exist and $\sigma^*=\sigma$ then it equals the effective operator (see Theorem \ref{ThmMainResultEffCondEqsEffOpLatticeZProb} and Example \ref{ExZProbNoOhmContinued}). $(3)$ Not necessarily (see Example \ref{ExEffCondLatticeZProbIsContrived}). 
	
    To elaborate on all this, we prove the next two key lemmas.
	\begin{lemma}\label{prop:PeriodicVoltCurrentSpBreakdown}
		Let $\mathcal{U},\mathcal{E},\mathcal{J}$ denote the subspaces in (\ref{ThmLatDecompU}), (\ref{ThmLatDecompE}), and (\ref{ThmLatDecompJ}). Then
		\begin{gather}
			\ker D^\bullet \cap \Hone=\ker (D^\bullet)_\#=\mathcal{U}\ho\mathcal{J}\label{ref:JUeqaul},\\
			\ran D\cap \Hone\supseteq\mathcal{U}\ho\mathcal{E}\label{ref:ranDsubsetUplusE}.
		\end{gather}
	\end{lemma}
	\begin{proof}
		First, $D^\bullet$ and $(D^\bullet)_\#$ are equal in value on $\Hone$ which implies $\ker D^\bullet\cap\Hone=\ker (D^\bullet)_\#$ and hence (\ref{ref:JUeqaul}) then follows by Theorem $\ref{ThmLatDecomp}$. Second, since $D$ and $D_\#$ are equal in value on $\Hzero$ it follows immediately by the definition of $\mathcal{E}$ that $\mathcal{E}\subseteq \ran D\cap\Hzero$. As the $\ran D \cap \Hone$ is a subspace of $\Hone$ and $\operatorname{span}\{1_E\}=\mathcal{U}$, it suffices to show that $1_E\in\ran D\cap\Hone$. Define the function $u_{\text{pos}}:\Fzero\to\Fzero$ by
		\begin{align}
			u_{\text{pos}}[(x_1,\ldots,x_d)]=\sum_{j=1}^dx_j,\text{ for all } (x_1,\ldots, x_d)\in\mathbb{Z}^d. 
		\end{align}
		Then for each $e\in E$ there exists a standard basis vector $\mathbf{e}_i$ such that $e_+=e_-+\textbf{e}_i$ so that
		\begin{gather*}
			(Du)(e)=u_{\text{pos}}(e_+)-u_{\text{pos}}(e_-)=\sum_{j=1}^d(e_+)_j-(e_-)_j\\
			=\sum_{j=1}^d(e_-+\mathbf{e}_i)_j-(e_-)_j
			=\sum_{j=1}^d(e_-+\mathbf{e}_i-e_-)_j\\=\sum_{j=1}^d (\mathbf{e}_i)_j=(\mathbf{e}_i)_i=1=1_{E}(e).
		\end{gather*}
		Therefore, $Du=1_E\in \ran D\cap \Hone$. This proves (\ref{ref:ranDsubsetUplusE}), which completes the proof.
	\end{proof}
	
	\begin{lemma}\label{lem:PeriodicVoltCurrentSpBreakdown2}
	    If $(I_0,V,I)\in\mathcal{U}\times\mathcal{E}\times\mathcal{J}$ is a solution of the lattice $Z$-problem (in Def.\ \ref{def:latzprob}) at $V_0\in\mathcal{U}$ then
	    \begin{gather}
		    I_0+I=\sigma(V_0+V),\\
		    V_0+V\in\mathcal{U}\ho\mathcal{E}\subseteq \operatorname{ran}D\cap \Hone,\\
		    I_0+I\in\mathcal{U}\ho\mathcal{J}=\ker D^\bullet \cap \Hone
		\end{gather}
		and, in particular, $V_0+V, I_0+I$ are a pair of voltage and current, respectively, satisfying the periodic Ohm's law (in Def.\ \ref{def:perohmlaw}).
	\end{lemma}
	\begin{proof}
	    The proof of this lemma follows immediately from Def.\ \ref{def:latzprob} and Def.\ \ref{def:perohmlaw} by Lemma \ref{prop:PeriodicVoltCurrentSpBreakdown}.
	\end{proof}
	
	We are now ready to prove the main result of this section.
	\begin{theorem}\label{ThmMainResultEffCondEqsEffOpLatticeZProb}
	    If an effective conductivity $\sigma_{eff}$ of the lattice electrical network $(\mathbb{Z}^d,E,\sigma)$ (in Def.\ \ref{def:perohmlaw}) exists and $\sigma^*=\sigma$ then an effective operator of lattice $Z$-problem $(\mathcal{H},\mathcal{U},\mathcal{E},\mathcal{J},\sigma)$ (in Def.\ \ref{def:latzprob}) exists, is unique, and equals $\sigma_{eff}$, i.e.,
	    \begin{align}
	        \sigma_{eff}=\sigma_*.
	    \end{align}
	\end{theorem}
	\begin{proof}
	Assume the hypotheses. Suppose  $(I_0,V,I)\in\mathcal{U}\times\mathcal{E}\times\mathcal{J}$ is a solution of the lattice $Z$-problem at $V_0\in\mathcal{U}$. Then by Lemma \ref{lem:PeriodicVoltCurrentSpBreakdown2} we know that $V_0+V, I_0+I$ are a pair of voltage and current, respectively, satisfying the periodic Ohm's law and so by hypothesis we must have $\langle I_0+I\rangle_{\Hone} = \sigma_{eff}(\langle V_0+V\rangle_{\Hone})$, that is,
		\begin{align*}
		    I_0=\langle I_0+I\rangle_{\Hone} = \sigma_{eff}(\langle V_0+V\rangle_{\Hone})=\sigma_{eff}(V_0).
		\end{align*}
		The result now follows immediately from this by Theorem \ref{thm:FundThmExistenceUniquenessSolvabilityZProbEffOpVecSp} and Corollary \ref{cor:KeyResultAppendixComposites}.
	\end{proof}
	
	The next theorem gives us necessary and sufficient conditions for the effective conductivity to exist.
	
	\begin{theorem}\label{thm:NecSuffCondEffCondExistLattice}
	    An effective conductivity $\sigma_{eff}$ of the lattice electrical network $(\mathbb{Z}^d,E,\sigma)$ (in Def.\ \ref{def:perohmlaw}) exists if and only if
	    \begin{align}
	    \ker(\Gamma_{\mathcal{E}}\sigma \Gamma_{\mathcal{V}})\subseteq \ker(\Gamma_{\mathcal{U}}\sigma \Gamma_{\mathcal{V}}),    
	    \end{align}
	    where $\Gamma_{\mathcal{U}}$ and $\Gamma_{\mathcal{E}}$ are the orthogonal projections of $\Hone$ onto $\mathcal{U}$ and $\mathcal{E}$, respectively, and $\Gamma_{\mathcal{V}}$ is the orthogonal projection of $\Hone$ onto 
	    \begin{align}
	        \mathcal{V}=[\operatorname{ran}D\cap \Hone]\overset{\perp}{\ominus} \mathcal{U},\label{def:VInthm:NecSuffCondEffCondExistLattice}
	    \end{align}
	    i.e., $\mathcal{V}$ is the orthogonal complement of $\mathcal{U}$ in $\operatorname{ran}D\cap \Hone$.
	\end{theorem}
	\begin{proof}
	    First, both $\mathcal{U}$ and $\operatorname{ran}D\cap \Hone$ are subspaces of the finite-dimensional Hilbert space $\Hone$ and so by Lemma \ref{prop:PeriodicVoltCurrentSpBreakdown}, $\mathcal{U}$ is a subspace of $\operatorname{ran}D\cap \Hone$ and hence it follows that
	    \begin{align}
	        \Hone = \mathcal{U}\ho \mathcal{V}\ho \mathcal{W},
	    \end{align}
	    where
	    \begin{align}
	        \mathcal{W}=\Hone\overset{\perp}{\ominus}[\operatorname{ran}D\cap \Hone]
	    \end{align}
	    is the orthogonal complement of $\operatorname{ran}D\cap \Hone$ in $\Hone$ and $\mathcal{V}$ is defined by (\ref{def:VInthm:NecSuffCondEffCondExistLattice}). The proof now follows immediately from Theorem \ref{thm:FundThmExistenceUniquenessSolvabilityZProbEffOpVecSp} in Appendix \ref{SectAbsTheoryCompositesVecSpFramework} by setting
	    \begin{gather*}
	        V=W=\Hone,V_0=W_0=\mathcal{U},\\
	        V_1=\mathcal{V}, V_2=\mathcal{W},W_1=\mathcal{E},W_2=\mathcal{J}.
	    \end{gather*}
	\end{proof}
	
	In the rest of this section, we give several examples that demonstrate the complexity associated with the Def.\ \ref{def:perohmlaw}. Our first example shows that, in general, (\ref{ref:ranDsubsetUplusE}) is not an equality.
		\begin{example}\label{ExZProbNoOhm} Here we give an example of a lattice electric network $(\mathbb{Z}^d,E,\sigma)$ with $\sigma\in\mathcal{L}(\Hone)$ and $\sigma^*=\sigma\geq 0$, and a pair $V,I$ satisfying the periodic Ohm's law (\ref{def:perohmlawpart1}) and (\ref{def:perohmlawpart2}), but $V\not\in \mathcal{U}\ho\mathcal{E}$ from which it will follow immediately by Lemma \ref{prop:PeriodicVoltCurrentSpBreakdown} that
		\begin{align*}
		    \mathcal{U}\ho\mathcal{E}\subsetneq\ran D\cap \Hone.
		\end{align*}
		
		Consider the periodic digraph $G=(\mathbb{Z}^d,E)$ with lattice $L$ in (\ref{ref:latticevecperiodicitytau})-(\ref{ref:edgeset}) with $d=2$ and $\tau=(\tau_1,\tau_2)=(2,2)$ as shown in Fig.\ \ref{FigLatEx}.$(b)$. Let $\sigma=I_{\Hone}$ be the identity operator on $\Hone$. Define the functions $u\in \mathcal{F}(\mathbb{Z}^2,\mathbb{K})$ and $V\in \ran D$ by $u(i,j)=i-j$ for each $(i,j)\in \mathbb{Z}^2$ and $V=Du$. Then $V\in\ran D\cap\Hone$ since, for each $e\in E$, we have
		\begin{align*}
			V(e)=\begin{cases}1, \text{ if }e||\mathbf{e}_1,                                                                     \\
			-1, \text{ if } e||\mathbf{e}_2,\end{cases}
		\end{align*}
		which implies $V(e+R)=V(e) \text{ for all } (e,R)\in E\times L$. Let us define $I$ by $I=\sigma V$, i.e., $I=V$. We will now show that the pair $V,I$ satisfies the periodic Ohm's law (\ref{def:perohmlawpart1}) and (\ref{def:perohmlawpart2}), but $V\not\in \mathcal{U}\ho\mathcal{E}$. To do this, it suffices to show $I\in\mathcal{J}$, i.e., $\langle I \rangle_{\Hone}=0$ and $D^\bullet I=0$. For this lattice, we have $|E_P^-|=8$, $|\{e\in E_P^-:e||\mathbf{e}_1\}|=4$, and $|\{e\in E_P^-:e||\mathbf{e}_2\}|=4$ so that
		\begin{gather*}
			\langle I\rangle_{\Hone}=\frac{1}{8}\left(\sum_{e\in E_P^-:e||\mathbf{e}_1}I(e)+\sum_{e\in E_P^-:e||\mathbf{e}_2}I(e)\right)=0.
		\end{gather*}
		Also, $D^\bullet I=0$ since, for any $p=(i,j)\in\mathbb{Z}^2$, we have
		\begin{gather*}
			(-D^\bullet I)(p)=(-D^\bullet Du)(p)=-\sum_{p'\sim p}u(p)-u(p')\\
			=u(i,j)-u(i-1,j)+u(i,j)-u(i,j-1)\\+u(i+1,j)-u(i,j)+u(i,j+1)-u(i,j)=0,
		\end{gather*}
		where $p'\sim p$ means $p$ is connected to $p'\in \mathbb{Z}^2$ by an edge in $E$. This proves $I\in \mathcal{J}$, hence $V=I\not\in \mathcal{U}\ho\mathcal{E}$ and $\mathcal{U}\ho\mathcal{E}\subsetneq\ran D\cap \Hone.$
	\end{example}

The next two examples show that even if $\mathcal{U}\ho\mathcal{E}\subsetneq\ran D\cap \Hone$, it is possible that $\sigma_{eff}=\sigma_*$ in one example, but its also possible, in another example, that $\sigma_{eff}$ does not exist yet $\sigma_*$ does. 
\begin{example}\label{ExZProbNoOhmContinued}
 		 Consider the lattice electric network $(\mathbb{Z}^d,E,\sigma)$ in Example \ref{ExZProbNoOhm}. In this example, the effective conductivity $\sigma_{eff}$ and the effective operator $\sigma_*$ exist with $\sigma_{eff}=\sigma_*$. Let us prove this. As $\sigma=I_{\Hone}$ so that $\sigma^*=\sigma\geq 0$, then by Theorem \ref{ThmEffCondEffOp} and Theorem \ref{ThmMainResultEffCondEqsEffOpLatticeZProb} it suffices to prove an effective conductivity $\sigma_{eff}$ of this lattice electrical network (see Def.\ \ref{def:perohmlaw}) exists. Thus, our proof will be complete if we can prove the claim that the identity operator $I_{\mathcal{U}}$ on $\mathcal{U}$ is such an effective conductivity. Suppose that $I$ and $V$ satisfy $I=\sigma V$ with $I\in\ker D^\bullet \cap \Hone, V\in\operatorname{ran}D\cap \Hone.$ Then $I=V$ so that $\langle I \rangle_{\Hone}=\langle V \rangle_{\Hone}=I_{\mathcal{U}}(\langle V \rangle_{\Hone})$, which proves the claim. In particular, we have shown that $\sigma_{eff}=\sigma_*=I_{\mathcal{U}}$.
\end{example}
	
The final example shows that an effective conductivity $\sigma_{eff}$ of lattice electric network $(\mathbb{Z}^d,E,\sigma)$ as defined in Def.\ \ref{def:perohmlaw} may not exist even though a unique effective operator $\sigma_*$ of the lattice $Z$-problem $(\mathcal{H},\mathcal{U},\mathcal{E},\mathcal{J},\sigma)$ (in Def.\ \ref{def:latzprob}) does. \footnote{In particular, this provides a counterexample to a claim (i.e., Theorem 65) in Ref.\ \onlinecite{22KB}.}
\begin{example}\label{ExEffCondLatticeZProbIsContrived}
 Consider any periodic digraph $G=(\mathbb{Z}^d,E)$ with lattice $L$ in (\ref{ref:latticevecperiodicitytau})-(\ref{ref:edgeset}) such that $\mathcal{J}\cap \operatorname{ran}D\cap \Hone\not=\emptyset$ (Example \ref{ExZProbNoOhm} gives one such example). Choose any unit vectors $I_0\in \mathcal{U}$ and $I\in\mathcal{J}\cap \operatorname{ran}D\cap \Hone$. Then there is an $\sigma\in \mathcal{L}({\Hone})$ with $\sigma^*=\sigma\geq 0$ such that
 \begin{align}
     \sigma(I_0)=\sigma(I)=I_0+I,\;\sigma(\operatorname{span}\{I_0,I\}^{\perp})=\{0\}.
 \end{align}
 Now consider the corresponding lattice electric network $(\mathbb{Z}^d,E,\sigma)$. We claim there is no effective conductivity for it as defined in Def.\ \ref{def:perohmlaw}. For suppose there was, i.e., suppose there exists a function $\sigma_{eff}:\mathcal{U}\rightarrow \mathcal{U}$ such that 
 \begin{align}
     \sigma_{eff}(\langle V\rangle_{\Hone})=\langle J\rangle_{\Hone} ,
 \end{align}
 whenever
 \begin{align}
     \sigma(V)=J,\;	J\in\ker D^\bullet \cap \Hone, V\in\operatorname{ran}D\cap \Hone.
 \end{align}
 Then since $\sigma(0)=0$ and with $V=0\in\operatorname{ran}D\cap \Hone$ and $J=0\in\ker D^\bullet \cap \Hone$, we must have
 \begin{align}
    \sigma_{eff}(0)=\sigma_{eff}(\langle 0\rangle_{\Hone})=\langle 0\rangle_{\Hone}=0.
 \end{align}
 But we also have $\sigma(I)=I_0+I$ and so with $V=I\in\operatorname{ran}D\cap \Hone$ and $J=I_0+I\in\ker D^\bullet \cap \Hone$, we must have
  \begin{align}
    \sigma_{eff}(0)=\sigma_{eff}(\langle I\rangle_{\Hone})=\langle I_0+I\rangle_{\Hone}=I_0.
 \end{align}
 This implies $0=\sigma_{eff}(0)=I_0\not=0,$ a contradiction. This proves the claim. 
 
 Nevertheless, in this example there is a unique effective operator $\sigma_*$ of the lattice $Z$-problem $(\mathcal{H},\mathcal{U},\mathcal{E},\mathcal{J},\sigma)$ (in Def.\ \ref{def:latzprob}) by Theorem \ref{ThmEffCondEffOp}.
\end{example}

\subsubsection{Periodic Dirichlet \textit{Z}-Problem}

The purpose of this section is to clarify the nuances that occurred in Section \ref{SecEffCond} using a new but related $Z$-problem. Our results in this regard are based on the following theorem that has a similar flavor to Theorem \ref{ThmKHodge}. Thus, in a sense, one could consider this section to be a periodic analogy of Section \ref{sec:EffOpReprOfDtNMap} on the Dirichlet $Z$-problem. For our purposes though, we will be satisfied with just the next two theorems.
\begin{theorem}\label{ThmPerDirHodgeDecomp}
    The sets $\mathcal{U}_{\#},\;\mathcal{E}_{\#},\;\mathcal{J}_{\#}$ defined by
	\begin{align}
		\mathcal{U}_{\#} & = \{D u\in \Hone : u \in \mathcal{F}(\mathbb{Z}^d, \mathbb{K}), D^{\bullet}D u=0\} \label{UPoundSpace}, \\
		\mathcal{E}_{\#} & = \{Du: u\in\Hzero\}, \label{EPoundSpace}                           \\
		\mathcal{J}_{\#}  & = (\operatorname{ran}D\cap \Hone)^{\perp} \label{JPoundSpace}
	\end{align}
	are mutually orthogonal subspaces in the Hilbert space $ \Hone$. Furthermore, 
	\begin{gather} \Hone=\mathcal{U}_{\#} \overset{\perp}{\oplus} \mathcal{E}_{\#} \ho \mathcal{J}_{\#}, \label{orthotriplePerDirZproblem}\\
	\operatorname{ran}D\cap \Hone= \mathcal{U}_{\#} \ho \mathcal{E}_{\#},\label{ReprUPoundPlusEPound}\\
	\ker D^\bullet \cap \Hone=\ker (D^\bullet)_\#=\mathcal{U}_{\#}\ho\mathcal{J}_{\#},\label{ReprUPoundPlusJPound}\\
	\mathcal{U}_{\#}=\mathcal{U}\ho (\mathcal{J}\cap \operatorname{ran}D),\label{ReprOfUPound}\\
	\mathcal{E}_{\#}=\mathcal{E}=\operatorname{ran}D_{\#},\label{ReprOfEPound}\\
		\mathcal{J}_{\#}=\mathcal{J}\overset{\perp}{\ominus}(\mathcal{J}\cap \operatorname{ran}D).\label{ReprOfJPound}
	\end{gather}
\end{theorem}
\begin{proof}
    First, (\ref{ReprOfEPound}) follows immediately from the definition of $D_{\#}$ and by Theorem \ref{ThmLatDecomp}. Next, we will prove (\ref{ReprOfUPound}). First, let $J\in\mathcal{J}\cap \operatorname{ran}D$. Then $J\in \mathcal{J}$ so that $(D^\bullet)_{\#} J=0, \langle J \rangle_{\Hone}=0,$ and there exists $u \in \mathcal{F}(\mathbb{Z}^d, \mathbb{K})$ such that $J=Du$ so that $0=(D^\bullet)_{\#} J=D^\bullet J=D^\bullet Du$ implying $J\in \mathcal{U}_{\#}$. This proves that $\mathcal{J}\cap \operatorname{ran}D\subseteq \mathcal{U}_{\#}.$ Next, by Lemma \ref{prop:PeriodicVoltCurrentSpBreakdown} we know that $\mathcal{U}\subseteq \ran D\cap \Hone$ and by Theorem \ref{ThmLatDecomp} we know that $\mathcal{U}\subseteq \ker (D^\bullet)$ which implies $\mathcal{U}\subseteq \mathcal{U}_{\#}.$ This together with $\mathcal{U}\perp \mathcal{J}$ (by Theorem \ref{ThmLatDecomp}) proves that $\mathcal{U}\ho (\mathcal{J}\cap \operatorname{ran}D)\subseteq \mathcal{U}_{\#}.$ To prove the reverse inclusion, let $V\in \mathcal{U}_{\#}.$ Then $V=Du$ for some $u \in \mathcal{F}(\mathbb{Z}^d, \mathbb{K})$ and $(D^{\bullet})_{\#}V=D^{\bullet}D u=0$. By Theorem \ref{ThmLatDecomp} we know that $\ker(D^\bullet)=\mathcal{U}\ho\mathcal{J}$ so that $V=I_0+I$ for some $I_0\in \mathcal{U}, I\in \mathcal{J}$. As $\mathcal{U}\subseteq \mathcal{U}_{\#}\subseteq \operatorname{ran}D$, this implies that $I=V-I_0\in \operatorname{ran}D$ and hence $I\in \mathcal{J}\cap \operatorname{ran}D$. Thus, $V=I_0+I\in \mathcal{U}\ho (\mathcal{J}\cap \operatorname{ran}D)$. This proves that $\mathcal{U}_{\#}\subseteq \mathcal{U}\ho (\mathcal{J}\cap \operatorname{ran}D)$, which proves (\ref{ReprOfUPound}). Next, we prove (\ref{ReprUPoundPlusEPound}). First, it is clear from the definitions of $\mathcal{U}_{\#}$ and $\mathcal{E}_{\#}$ that $\mathcal{U}_{\#}, \mathcal{E}_{\#}\subseteq \operatorname{ran}D\cap \Hone$. We also know by Theorem \ref{ThmLatDecomp} that $\mathcal{U}\perp \mathcal{E}$ and $\mathcal{J}\perp \mathcal{E}$ so by (\ref{ReprOfUPound}) and (\ref{ReprOfEPound}) it follows that $\mathcal{U}_{\#}\perp \mathcal{E}_{\#}$ and hence $\mathcal{U}_{\#}\ho\mathcal{E}_{\#}\subseteq \operatorname{ran}D\cap \Hone$. To prove the reverse inclusion, let $u \in \mathcal{F}(\mathbb{Z}^d, \mathbb{K})$ such that $Du\in \Hone.$ Then by Theorem \ref{ThmLatDecomp} there exists an $V_0\in \mathcal{U}, V\in \mathcal{E}, I\in \mathcal{J}$ such that $Du=V_0+V+I.$ Now, by Lemma \ref{prop:PeriodicVoltCurrentSpBreakdown} we know that $V_0+V\in \operatorname{ran}D\cap \Hone$ so that there exists $v \in \mathcal{F}(\mathbb{Z}^d, \mathbb{K})$ such that $V_0+V=Dv\in \Hone.$ This implies that $I=D(u-v)=Du-Dv\in \operatorname{ran}D\cap \Hone$ and since $I\in \mathcal{J}$ then $0=D^\bullet I=D^\bullet D(u-v)$ with $u-v\in \mathcal{F}(\mathbb{Z}^d, \mathbb{K})$ this implies $I\in \mathcal{U}_{\#}$. It follows from this, (\ref{ReprOfUPound}), and (\ref{ReprOfEPound}) that $Du=(V_0+I)+V\in \mathcal{U}_{\#}\ho\mathcal{E}_{\#}$. Hence, we conclude that $\operatorname{ran}D\cap \Hone\subseteq \mathcal{U}_{\#}\ho\mathcal{E}_{\#}$, which proves (\ref{ReprUPoundPlusEPound}). Next, (\ref{orthotriplePerDirZproblem}) follows immediately from (\ref{ReprUPoundPlusEPound}) and the definition of $\mathcal{J}_{\#}$. Now, (\ref{ReprOfJPound}) follows from (\ref{ReprUPoundPlusEPound}), (\ref{ReprOfUPound}), and (\ref{ReprOfEPound}) since (by Theorem \ref{ThmLatDecomp})
    \begin{gather*}
        \Hone= \mathcal{U}\ho \mathcal{E}\ho \mathcal{J}\\
        =[\mathcal{U}\ho (\mathcal{J}\cap \operatorname{ran}D)]\ho \mathcal{E}\ho[\mathcal{J}\overset{\perp}{\ominus}(\mathcal{J}\cap \operatorname{ran}D)]\\
        =\mathcal{U}_{\#}\ho\mathcal{E}_{\#}\ho[\mathcal{J}\overset{\perp}{\ominus}(\mathcal{J}\cap \operatorname{ran}D)]\\
        =[\operatorname{ran}D\cap \Hone]\ho [\mathcal{J}\overset{\perp}{\ominus}(\mathcal{J}\cap \operatorname{ran}D)].
    \end{gather*}
    Finally, (\ref{ReprUPoundPlusJPound}) follows (\ref{ReprOfUPound}), (\ref{ReprOfJPound}), and (\ref{ref:JUeqaul}) from Lemma \ref{prop:PeriodicVoltCurrentSpBreakdown}. This completes the proof.
\end{proof}

The next theorem and it's proof is inspired by standard operations on effective operators and subspace collections in the abstract theory of composites [see Eqs.\ (3.17) and (3.18) in Ref.\ \onlinecite{87bGM}, Sec.\ 29.1 in Ref.\ \onlinecite{02GM}, and Chap.\ 7, Sec.\ 9 in Ref.\  \onlinecite{16GM}].
\begin{theorem}\label{ThmEffOpInTermsOfEffOpPound}
    Let $(\mathbb{Z}^d,E,\sigma)$ be a lattice electrical network with conductivity $\sigma\in \mathcal{L}(\Hone)$ satisfying $\sigma^*=\sigma\geq 0$. Let $\sigma_*$ be the effective operator of the lattice $Z$-problem $(\mathcal{H},\mathcal{U},\mathcal{E},\mathcal{J},\sigma)$ (in Def.\ \ref{def:latzprob}) with Hilbert space $\mathcal{H}=\Hone$ and orthogonal triple decomposition of $\mathcal{H}$ in (\ref{LatticeHodgeDecomp}). Let $\sigma_{*_{\#}}$ be the effective operator of the $Z$-problem $(\mathcal{H}_{\#},\mathcal{U}_{\#},\mathcal{E}_{\#},\mathcal{J}_{\#},\sigma)$ with Hilbert space $\mathcal{H}_{\#}=\Hone$ and orthogonal triple decomposition of $\mathcal{H}_{\#}$ in (\ref{orthotriplePerDirZproblem}).  Then
    \begin{align}
        \sigma_*=\Gamma_0\sigma_{*_{\#}}\Gamma_0|_{\mathcal{U}},\label{RelEffOpInTermsOfEffOpPound}
    \end{align}
    where $\Gamma_0$ is the (edge cell) average operator (in Def.\ \ref{DefAvgEdgeUnitCell}), i.e., $\sigma_*$ is the compression of $\sigma_{*_{\#}}$ to $\mathcal{U}$.
\end{theorem}
\begin{proof}
    Assume the hypotheses. Then $\sigma^*=\sigma\geq 0$ and hence by Corollary \ref{CorSuffCondForTrueThmGenClassicalDiriMinPrin} we have existence and uniqueness of both the effective operators $\sigma_*$ and $\sigma_{*_{\#}}$. We will now prove (\ref{RelEffOpInTermsOfEffOpPound}).
    
    Let $V_0\in \mathcal{U}$. Then, by Corollary \ref{CorSuffCondForTrueThmGenClassicalDiriMinPrin}, there exists a $(I_0,V, I)\in \mathcal{U}\times \mathcal{E}\times \mathcal{J}$ such that
\begin{align*}
    I_0+I=\sigma(V_0+V)
\end{align*}
and
\begin{align*}
   I_0= \sigma_*V_0.
\end{align*}
On the other hand, $V\in \mathcal{E}=\mathcal{E}_{\#}$ and, by Theorem \ref{ThmPerDirHodgeDecomp}, $V_0,I_0\in \mathcal{U}_{\#}$ and $(I_{\mathcal{H}}-\Gamma_2^{\#})I\in \mathcal{U}_{\#}\overset{\perp}{\ominus}\mathcal{U}$, where $\Gamma_2^{\#}$ is the orthogonal projection of $\mathcal{H}_{\#}=\mathcal{H}$ onto $\mathcal{J}_{\#}$, so that
\begin{align*}
    [I_0+(I_{\mathcal{H}}-\Gamma_2^{\#})I]+\Gamma_2^{\#}I=\sigma(V_0+V),
\end{align*}
with $V_0\in \mathcal{U}_{\#}$ and so $(I_0+(I_{\mathcal{H}}-\Gamma_2^{\#})I,V,\Gamma_2^{\#}I)\in \mathcal{U}_{\#}\times\mathcal{E}_{\#}\times \mathcal{J}_{\#}$ is a solution of the $Z$-problem $(\mathcal{H}_{\#},\mathcal{U}_{\#},\mathcal{E}_{\#},\mathcal{J}_{\#},\sigma)$ at $V_0$ implying by Corollary \ref{CorSuffCondForTrueThmGenClassicalDiriMinPrin} that
\begin{align*}
    I_0+(I_{\mathcal{H}}-\Gamma_2^{\#})I=\sigma_{*_{\#}}(V_0).
\end{align*}
It follows from this and the facts
\begin{align*}
    \Gamma_0V_0=V_0,\;\Gamma_0[I_0+(I_{\mathcal{H}}-\Gamma_2^{\#})I]=I_0
\end{align*}
[the latter a consequence of $(I_{\mathcal{H}}-\Gamma_2^{\#})I\in \mathcal{U}_{\#}\overset{\perp}{\ominus}\mathcal{U}$],
that we have
\begin{align*}
    \Gamma_0\sigma_{*_{\#}}\Gamma_0|_{\mathcal{U}}(V_0)=I_0=\sigma_*(V_0).
\end{align*}
As this is true for every $V_0\in \mathcal{U}$, then we have proven (\ref{RelEffOpInTermsOfEffOpPound}). This completes the proof.
\end{proof}

\appendix

\section{\label{SectAbsTheoryCompositesVecSpFramework}Abstract Theory of Composites for Vector Spaces}

In this section, we give a concise and self-contained presentation of the most fundamental aspects of the abstract theory of composites on vector spaces that we need in this paper. In particular, we do not assume here that the vector spaces are inner product spaces and as such need not be Hilbert spaces. This level of generality was first considered in Chapter 7 in Ref.\ \onlinecite{16GM} and this appendix can be considered supplementary to it.

We begin by introducing some notation. Let $V, W$ be two vector spaces over a field $\mathbb{K}$. Let $L(V,W)$ denote the vector space (over $\mathbb{K}$) of all the linear functions from $V$ to $W$ [when $V=W$, we abbreviate it by $L(V),$ i.e., $L(V)=L(V,V)$]. Now suppose $V$ has a direct sum decomposition
\begin{align}
    V=V_0\oplus V_1 \oplus V_2,\label{VDecompAbsTheoryCompositesVecSpFramework}
\end{align}
for some subspaces $V_0,V_1,V_2\subseteq V$ and $W_0,W_1,W_2\subseteq W$. For each $i=0,1,2,$ let $\Gamma_{V_i}\in L(V)$ denote the projection of $V$ onto $V_i$ along $V_j\oplus V_k,$ where $j,k\in \{0,1,2\}\setminus\{i\}$ with $j<k$. More precisely, since $V=V_i\oplus (V_j \oplus V_k)$, then $\Gamma_{V_i}$ is uniquely defined by $\Gamma_{V_i}(v)=v$ if $v\in V_i$ and $\Gamma_{V_i}(v)=v$ if $v\in V_j \oplus V_k.$ Suppose that $W$ also has a direct sum decomposition
\begin{align}
    W=W_0\oplus W_1 \oplus W_2.\label{WDecompAbsTheoryCompositesVecSpFramework}
\end{align}
Then every $\sigma\in L(V,W)$ can be written as a $3\times 3$ block operator matrix, i.e.,
\begin{align}
\sigma=[\sigma_{ij}]_{i,j=0,1,2}=\begin{bmatrix}
    \sigma_{00} & \sigma_{01}&\sigma_{02}\\
     \sigma_{10}& \sigma_{11}&\sigma_{12} \\
     \sigma_{20}&\sigma_{21} &\sigma_{22}
\end{bmatrix},    
\end{align}
with respect to the decompositions (\ref{VDecompAbsTheoryCompositesVecSpFramework}) and (\ref{WDecompAbsTheoryCompositesVecSpFramework}), where
\begin{align}
    \sigma_{ij}\in \mathcal{L}(V_j,W_i),\;\sigma_{ij}=\Gamma_{W_i}\sigma\Gamma_{V_j}:V_j\rightarrow W_i,\label{DefOfSigmaSubblocksAppendix}
\end{align}
for $i,j=0,1,2$.

\begin{definition}[$Z$-problem and effective operator]\label{DefZProbMainAppendix}
	The $Z$-problem
	\begin{align}
	    (V,V_0,V_1,V_2, W, W_0,W_1,W_2, \sigma),\label{DefZProbAppendix}
	\end{align}
	is the following problem associated with two vector spaces $V,W$ (over a field $\mathbb{K}$) having decompositions
	\begin{equation}
	V=V_0\oplus V_1 \oplus V_2,\;W=W_0\oplus W_1 \oplus W_2\label{DefZProbTriDecomppAppendix}
	\end{equation}
	and an $\sigma\in L(V,W)$:
	given $v_{0}\in V_0$, find $v_1\in V_1,w_0\in W_0, w_2\in W_2$ satisfying 
	\begin{equation}
		w_{0}+w_2=\sigma \left(  v_0+v_1\right),
		\label{DefZProbEqAppendix} 
	\end{equation}
    such a tuple $(w_0,v_1, w_2)\in W_0\times V_1\times W_2$ called a solution of the $Z$-problem at $v_0$.
	If there exists a function $\sigma_*$ such that 
	\begin{equation}
		w_{0}=\sigma_{\ast}(v_0), \label{DefZProbEffOpAppendix}
	\end{equation}
	whenever $v_0\in V_0$ and $(w_0,v_1, w_2)$ is a solution of the $Z$-problem at $v_0$, then $\sigma_*$ is called an effective operator of the $Z$-problem.
\end{definition}

\begin{remark}
In our definition and notation, the equation $(\ref{DefZProbEqAppendix})$ is equivalent to the system
\begin{gather}
    \sigma_{00}v_0+\sigma_{01}v_1=w_0,\label{ZProbEquivFormPart1Appendix}\\
        \sigma_{10}v_0+\sigma_{11}v_1=0,\label{ZProbEquivFormPart2Appendix}\\
        \sigma_{20}v_0+\sigma_{21}v_1=w_2.\label{ZProbEquivFormPart3Appendix}
\end{gather}    
\end{remark}

In particular, it follows that if $\sigma_{11}$ is invertible, then we have the following formulas for the solution of the $Z$-problem at each $v_0\in V_0$ and the effective operator as a Schur complement:
\begin{gather}
    w_0=\sigma_*v_0,\;v_1=-\sigma_{11}^{-1}\sigma_{10}v_0,\; w_2=\sigma_{20}v_0+\sigma_{21}v_1,\label{ClassicSolnZProbAppendix}\\
    \sigma_*=\begin{bmatrix}
        \sigma_{00}&\sigma_{01}\\
        \sigma_{10}&\sigma_{11}
    \end{bmatrix}/\sigma_{11}=\sigma_{00}-\sigma_{01}\sigma_{11}^{-1}\sigma_{10}.\label{ClassicEffOperFormulaAppendix}
\end{gather}
Thus, the following theorem follows.
\begin{theorem}\label{thm:AppendixThmMainClassicalZProbEffOp}
If $(V,V_0,V_1,V_2, W, W_0,W_1,W_2, \sigma)$ is a $Z$-problem (as in Def.\ \ref{DefZProbMainAppendix}) and $\sigma_{11}$ [as defined by (\ref{DefOfSigmaSubblocksAppendix})] is invertible then the $Z$-problem (\ref{DefZProbEqAppendix}) has a unique solution for each $v_0\in V_0$ and it is given by the formulas (\ref{ClassicSolnZProbAppendix}), (\ref{ClassicEffOperFormulaAppendix}). Moreover, the effective operator of the $Z$-problem exists, is unique, and is given by the Schur complement formula (\ref{ClassicEffOperFormulaAppendix}).
\end{theorem}

The next theorem and corollary provides the necessary and sufficient conditions for an effective operator of a $Z$-problem to exist along with some of the fundamental properties of it when it does exist.

\begin{theorem}[Fundamental theorem of effective operators]\label{thm:FundThmExistenceUniquenessSolvabilityZProbEffOpVecSp}
    Let $(V,V_0,V_1,V_2, W, W_0,W_1,W_2, \sigma)$ be a $Z$-problem. Define the set $D_*(\sigma)\subseteq  V_0$ by
    \begin{align}
        D_*(\sigma)=\{v_0\in V_0:\text{a solution to the $Z$-problem at $v_0$ exists}\}.\label{ThmPartialAnswerQuestioniiPart1Appendix}
    \end{align}
    Then $D_*(\sigma)$ is the inverse image of $\sigma_{10}$ on the range of $\sigma_{11},$ i.e.,
    \begin{align}
        D_*(\sigma)=\sigma_{10}^{-1}(\ran \sigma_{11}),\label{FundFormulaForDStarSigma}
    \end{align}
    and, in particular, is a subspace of $V_0$.
    Furthermore, the following statements are equivalent:
    \begin{itemize}
        \item[(a)] $\ker \sigma_{11}\subseteq \ker \sigma_{01}.$
        \item[(b)] There exists a function $f:D_*(\sigma)\rightarrow W_0$ such that
    \begin{align}
        w_0=f(v_0)\label{PropPartialAnswerQuestioniiPart3Appendix}
    \end{align}
    whenever $v_0\in V_0$ and $(w_0,v_1,w_2)$ is a solution of the $Z$-problem at $v_0$.
        \item[(c)] If $(w_0,v_1,w_2)$ is a solution of the $Z$-problem at $v_0=0$ then $w_0=0$.
    \end{itemize}
    Moreover, if any of the statements $(a)$, $(b)$, or $(c)$ is true then an effective operator $\sigma_*$ of the $Z$-problem exists, is unique, equals the function $f:D_*(\sigma)\rightarrow W_0$ in $(b)$, i.e.,
    \begin{align}
        \sigma_*=f,
    \end{align}
    and $\sigma_*$ is linear, that is,
    \begin{align}
        \sigma_*\in L(D_*(\sigma),W_0).
    \end{align}
\end{theorem}
\begin{proof}
    Assume the hypotheses. First, we prove the equality (\ref{FundFormulaForDStarSigma}). Let $v_0\in D_*(\sigma)$. Then there exists a triple $(w_0,v_1,w_2)\in W_0\times V_1\times W_2$ which is a solution of the $Z$-problem at $v_0$. This implies that $\sigma_{10}v_0+\sigma_{11}v_1=0$ so that $v_0\in \sigma_{10}^{-1}(\{\sigma_{11}(-v_1)\})\in \sigma_{10}^{-1}(\ran\sigma_{11})$. This proves $D_*(\sigma)\subseteq \sigma_{10}^{-1}(\ran\sigma_{11})$. Conversely, suppose $v_0\in \sigma_{10}^{-1}(\ran\sigma_{11})$. Then there exists $v_1\in V_1$ such that $\sigma_{10}v_0=\sigma_{11}(-v_1),$ that is, $\sigma_{10}v_0+\sigma_{11}v_1=0$. Then defining $w_0, w_2$ by $w_0=\sigma_{00}v_0+\sigma_{01}v_1\in W_0$ and $w_2=\sigma_{20}v_0+\sigma_{21}v_1\in W_2$ it follows that $v_0\in V_0, (w_0,v_1,w_2)\in W_0\times V_1\times W_2$ and $\sigma(v_0+v_1)=w_0+w_2,$ i.e., $(w_0,v_1,w_2)$ is a solution to the $Z$-problem at $v_0$. Hence, $v_0\in D_*(\sigma).$ This proves $\sigma_{10}^{-1}(\ran\sigma_{11})\subseteq D_*(\sigma)$. Thus, we have proven equality (\ref{FundFormulaForDStarSigma}). Next, it follows immediately from this and the linearity of $\sigma_{10}$ and $\sigma_{11}$ that $D_*(\sigma)$ is a subspace of $V_0$.

     We will now prove $(a)\Rightarrow (b) \Rightarrow (c)\Rightarrow (a)$. Suppose that $(a)$ is true. Let us prove $(b)$. Let $v_0\in D_*(\sigma)$ and suppose there exists $(w_0,v_1,w_2),(w_0',v_1',w_2')\in W_0\times V_1\times W_2$ such that
     \begin{align*}
         \sigma(v_0+v_1)=w_0+w_2,\;\sigma(v_0+v_1')=w_0'+w_2'.
     \end{align*}
     Then by linearity of $\sigma$ we have
     \begin{align*}
         \sigma(0+v_1-v_1')=w_0-w_0'+w_2-w_2'
     \end{align*}
     so that $(w_0-w_0',v_1-v_1',w_2-w_2')\in W_0\times V_1\times W_2$ is a solution of the $Z$-problem at $0$. From this, it follows that
     \begin{align*}
         \sigma_{11}(v_1-v_1')=0,\;\sigma_{01}(v_1-v_1')=w_0-w_0'.
     \end{align*}
     By $(a)$, this implies that $0=\sigma_{01}(v_1-v_1')=w_0-w_0'$ so that $w_0=w_0'$. This proves that the function $f:D_*(\sigma)\rightarrow W_0$ defined by $w_0=f(v_0)$ whenever $(w_0,v_1,w_2)$ is a solution of the $Z$-problem at $v_0$, is well-defined. This proves $(b)$.
     
     Now we will prove $(b)$ implies $(c)$. Suppose $(b)$ is true. If $(w_0,v_1,w_2)$ is a solution of the $Z$-problem at $v_0=0$ then by $(b)$ we must have $f(0)=w_0,$ but we also know that $(0,0,0)$ is a solution of the $Z$-problem at $0$ so by $(b)$ we must have $f(0)=0$ and hence $w_0=0$. This proves $(c)$.
     
     Next, we will prove $(c)$ implies $(a)$. Suppose $(c)$ is true. Let $v_1\in \ker \sigma_{11}$. Then $0=\sigma_{11}(v_1)=\Gamma_{W_1}\sigma\Gamma_{V_1}(v_1)=\Gamma_{W_1}\sigma(v_1)$. This implies $\sigma(v_1)=w_0+w_2$ for some $w_0\in  W_0, w_2\in W_2$ and hence $(w_0,v_1,w_2)$ is a solution of the $Z$-problem at $v_0=0$. By $(c)$ we must have $w_0=0$ and hence $\sigma(v_1)=w_2$ so that $\sigma_{01}(v_1)=\Gamma_{W_0}\sigma\Gamma_{V_1}(v_1)=\Gamma_{W_0}w_2=0$. In particular, $v_1\in \ker \sigma_{01}$. This proves $(a)$.
     
     Therefore, we have proven $(a)\Rightarrow (b)\Rightarrow (c) \Rightarrow (a)$, i.e., statements $(a)$, $(b)$, and $(c)$ are equivalent.
     
     To prove the remaining statements, assume that $(a)$, $(b)$, or $(c)$ is true. Then they are all true as we just showed. Let $f$ be any function with the properties described in $(b)$. Then by Definition \ref{DefZProbMainAppendix} and the definition of $D_*(\sigma)$ it follows immediately that $f$ is an effective operator of the $Z$-problem and if $g$ is another effective operator of the $Z$-problem then it is also a function with the properties described in $(b)$. Let us now prove that $f=g$. If $v_0\in D_*(\sigma)$ then by definition $v_0\in V_0$ and there exists a solution $(w_0,v_1,w_2)$ of the $Z$-problem at $v_0$ implying by $(b)$ that $f(v_0)=w_0=g(v_0)$. This proves $f=g$. Hence, we can denote $f$ by $\sigma_*,$ i.e., $\sigma_*=f$. We now prove that $\sigma_*:D_*(\sigma)\rightarrow W_0$ is linear, i.e., $\sigma_*\in L(D_*(\sigma),W_0)$. Let $v_0,v_0'\in D_*(\sigma)$ and $c\in \mathbb{K}$. Then there exists $(w_0,v_1,w_2),(v_0',v_1',w_2')\in W_0\times V_1\times W_2$ such that
     \begin{align*}
         \sigma(v_0+v_1)=w_0+w_2,\;\sigma(v_0'+v_1')=w_0'+w_2',
     \end{align*}
     and hence $w_0=\sigma_*(v_0)$ and $w_0'=\sigma_*(v_0')$. Thus, by linearity of $\sigma$, it follows that
     \begin{align*}
         \sigma((cv_0+v_0')+(cv_1+v_1'))=(cw_0+w_0')+(cw_2+w_2')
     \end{align*}
     and so $(cw_0+w_0',cv_1+v_1',cw_2+w_2')\in W_0\times V_1\times W_2$ is a solution of the $Z$-problem at $cv_0+v_0'\in V_0$ which implies $cv_0+v_0'\in D_*(\sigma)$ and
     \begin{align*}
         \sigma_*(cv_0+v_0')=cw_0+w_0'=c\sigma_*(v_0)+\sigma_*(v_0').
     \end{align*}
     This proves that $\sigma_*:D_*(\sigma)\rightarrow W_0$ is linear, which completes the proof of the theorem.
\end{proof}

\begin{definition}[Hilbert space orthogonal $Z$-problem]\label{def:HilbertSpaceOrthoZProb}
    A $Z$-problem $(V,V_0,V_1,V_2, W, W_0,W_1,W_2, \sigma),$ such that $\mathbb{K}=\mathbb{R}$ or $\mathbb{K}=\mathbb{C}$,
    $V=W$ is a Hilbert space (over $\mathbb{K}$), $V_i=W_i$ for $i=0,1,2$,
    \begin{align}
        V=V_0\ho V_1\ho V_2,
    \end{align}
    and $\sigma$ is a bounded linear operator, i.e., $\sigma\in\mathcal{L}(V)$, is called a Hilbert space orthogonal $Z$-problem and denoted by
    \begin{align}
        (V,V_0,V_1,V_2, \sigma).
    \end{align}
\end{definition}

\begin{corollary}\label{cor:KeyResultAppendixComposites}
    Let $(V,V_0,V_1,V_2, \sigma)$ be a Hilbert space orthogonal $Z$-problem such that $\dim V<\infty$ and
    	\begin{gather}
		\begin{bmatrix}
		\sigma_{00} & \sigma_{01}\\
		\sigma_{10} & \sigma_{11}
		\end{bmatrix}=\begin{bmatrix}
		\sigma_{00} & \sigma_{01}\\
		\sigma_{10} & \sigma_{11}
		\end{bmatrix}^*.\label{WeakAssumpKeyResultAppendixComposites}
	\end{gather}
    Then statements $(a)$, $(b)$, and $(c)$ in Theorem \ref{thm:FundThmExistenceUniquenessSolvabilityZProbEffOpVecSp} are equivalent to the statement:  \begin{itemize}
        \item[(d)] $\operatorname{ran}\sigma_{10}\subseteq \operatorname{ran}\sigma_{11}$.
    \end{itemize}
    Furthermore, if any of the statements $(a)$, $(b)$, and $(c)$ in Theorem \ref{thm:FundThmExistenceUniquenessSolvabilityZProbEffOpVecSp} or statement $(d)$ is true then an effective operator $\sigma_*$ of the $Z$-problem exists, is unique, equals the function $f:D_*(\sigma)\rightarrow V_0$ in $(b)$, i.e.,
    \begin{align}
        \sigma_*=f,
    \end{align}
    and $\sigma_*$ is a bounded linear operator, that is,
    \begin{align}
        \sigma_*\in \mathcal{L}(D_*(\sigma),V_0).
    \end{align}
    Moreover,
    \begin{align}
        D_*(\sigma)=V_0.
    \end{align}
\end{corollary}
\begin{proof}
    Assume the hypotheses. Then $V=V_0\ho V_1\ho V_2$ so that $\Gamma_{V_i}$ is the orthogonal projection of $V$ onto $V_i$ for each $i=0,1,2$ and $\sigma^*=\sigma$. From this and the fact that $V$ is a finite-dimensional Hilbert space, it follows that $\sigma_{ij}^*=\sigma_{ji}$ for all $i,j=0,1$ and so
    \begin{align}
        (\ker \sigma_{ij})^{\perp}=\operatorname{ran} \sigma_{ji},\;(\operatorname{ran} \sigma_{ji})^{\perp}=(\ker \sigma_{ij})^{\perp}
    \end{align}
    for all $i,j=0,1$. From this we conclude that statement $(d)$ is equivalent to statement $(a)$ in Theorem \ref{thm:FundThmExistenceUniquenessSolvabilityZProbEffOpVecSp}. The rest of the proof of this theorem now follows immediately from Theorem \ref{thm:FundThmExistenceUniquenessSolvabilityZProbEffOpVecSp} for the Hilbert space orthogonal $Z$-problem $(V,V_0,V_1,V_2,\sigma)$ and the hypotheses that $V$ is a finite-dimensional Hilbert space [so that $L(D_*(\sigma),V_0)=\mathcal{L}(D_*(\sigma),V_0)$]. In particular, if statement $(d)$ is true then
    \begin{align}
        V_0\supseteq D_*(\sigma)=\sigma_{10}^{-1}(\operatorname{ran}\sigma_{11})\supseteq \sigma_{10}^{-1}(\operatorname{ran}\sigma_{10})=V_0
    \end{align}
    implies $D_*(\sigma)=V_0$. This completes the proof of the corollary.
\end{proof}

\begin{remark}
Definitions \ref{def:HilbertSpaceOrthoZProb} and \ref{DefZProbMainAppendix} taken together are very similar to Definition \ref{DefZProbMain}. In fact, the only difference is in the definition of an effective operator $\sigma_*$. For in Definition \ref{DefZProbMain}, it is required that any effective operator $\sigma_*$ be a bounded linear operator, whereas in Definition \ref{DefZProbMainAppendix} it is just required to be a function. But there is no difference in the case of finite-dimensional Hilbert spaces under the assumption $\sigma^*=\sigma$ (or the weaker assumption \ref{WeakAssumpKeyResultAppendixComposites}) due to Corollary \ref{cor:KeyResultAppendixComposites}. This is one of the main reasons we begin this paper with Definition \ref{DefZProbMain}. The other reason is that in the classical case, i.e., $\sigma_{11}$ is invertible, it is also true that there is no difference which follows from Theorem \ref{thm:FundThmExistenceUniquenessSolvabilityZProbEffOpVecSp} and Theorem \ref{ThmMainClassicalZProbEffOp}.
\end{remark}

\section{\label{SectAbsHodgDecomp}Abstract Hodge Decompositions}

In this section, we adapt, simplify, and extend results of Ref.\ \onlinecite{20LL} from matrices to linear operators. These results can be utilized when identifying the orthogonal triple decomposition needed for $Z$-problems and their effective operators. We only consider finite-dimensional inner product spaces over the same field $\mathbb{K}$, where $\mathbb{K}=\mathbb{R}$ or  $\mathbb{K}=\mathbb{C}$ (hence real or complex Hilbert spaces, respectively).

To begin, we need the following well known results \cite{19FIS}, which we state without proofs.
\begin{lemma}\label{LemHodgePre1}
	Let $\mathcal{A},\mathcal{B},\mathcal{C}$ be finite-dimensional inner product spaces, $U\in\mathcal{L}(\mathcal{A},\mathcal{B}),$ and $T\in\mathcal{L}(\mathcal{B},\mathcal{C})$. Then
		\begin{enumerate}[(a)]
			\item $\mathcal{B}=\ker U^*\ho \ran U$,
			\item $\mathcal{B}=\ran T^*\ho \ker T$.
		\end{enumerate}
\end{lemma}

Using the next lemma we will easily be able prove the next theorem.
\begin{lemma}\label{LemHodgePre2}
	Let $\mathcal{A},\mathcal{B},\mathcal{C}$ be finite-dimensional inner product spaces. If $U\in\mathcal{L}(\mathcal{A},\mathcal{B})$ and $T\in\mathcal{L}(\mathcal{B},\mathcal{C})$ 
	satisfy
	\begin{gather}
		TU=0\; (i.e.,\; U^*T^*=0)\label{DefHodgeCond}
	\end{gather}
	then
	\begin{enumerate}[(a)]
		\item $\ker (T^*T+UU^*)=\ker T\cap \ker U^*$,
		\item $\ker T=\ker (T^*T+UU^*)\ho \ran U$.
	\end{enumerate}
\end{lemma}
\begin{proof}
	Assume the hypotheses.
	\begin{enumerate}[(a)]
		\item Clearly, $\ker T\cap\ker U^*\subseteq \ker(T^*T+UU^*)$. It remains to show $\ker (T^*T+UU^*)\subseteq \ker T\cap \ker U^*$. Let $x\in\ker (T^*T+UU^*)$. Then
		\begin{gather*}
    0 = (T^*Tx + UU^*x, x) = (T^*Tx, x) + (UU^*x,x) 
    \\= \|Tx\|^2_{\mathcal{C}} + \|U^* x \|^2_{\mathcal{A}}
\end{gather*}
    which implies $x\in \ker T\cap \ker U^*$. This proves $\ker (T^*T+UU^*)\subseteq \ker T\cap \ker U^*$, which proves $(a)$.
		
		\item Next, we claim that
\begin{align*}
     \ker T = (\ker T\cap \ker U^*) \ho \ran (U).
\end{align*}
To prove this, notice first that it is obvious that $$(\ker T\cap \ker U^*) \ho [\ran (U) \cap \ker (T)] \subseteq\ker T.$$ To prove the reverse inclusion, let $x \in \ker T$.  Then by Lemma (\ref{LemHodgePre1}).$(a)$, $x = a + b$ for some $a \in \ker(U^*) $ and $b \in \ran(U)$.   Since $\ran(U) \subseteq \ker(T)$ [by (\ref{DefHodgeCond})], we have that $b \in \ker(T)$.  But then $a=x-b \in \ker(T)$ as well, so $a \in \ker T \cap \ker U^*$ and $b \in \ran(U) \cap \ker(T)$.  This proves
    \begin{align*}
     \ker T & \subseteq (\ker T\cap \ker U^*) \ho [\ran (U) \cap \ker (T)].
\end{align*}
From this and since $\ran(U) \subseteq \ker(T)$ [by (\ref{DefHodgeCond})], the proof of the claim follows. The proof of $(b)$ now follows immediately from this and $(a)$.
	\end{enumerate}
	This completes the proof.
\end{proof}

We now have all the necessary tools to prove the abstract Hodge decomposition theorem.
\begin{theorem}[Abstract Hodge decomposition]\label{ThmHodgeDecomp}
	Let $\mathcal{A},\mathcal{B},\mathcal{C}$ be finite-dimensional inner product spaces. If $U\in\mathcal{L}(\mathcal{A},\mathcal{B})$ and $T\in\mathcal{L}(\mathcal{B},\mathcal{C})$ 
	satisfy
	\begin{gather}
		TU=0\;\; (i.e.,\; U^*T^*=0)
	\end{gather}
	then
	\begin{gather}
		\mathcal{B}=\operatorname{ran} T^*\ho\ker(T^*T+UU^*)\ho\operatorname{ran}U.
		\label{DefHodgeDecomp}
	\end{gather}
	Furthermore,
	\begin{gather}
		\ran(T^*T+UU^*)=\ran T^*\ho\ran U,\\
		\ker (T^*T+UU^*)=\ker T\cap \ker U^*.\label{LemLaplKer}
	\end{gather}
\end{theorem}
\begin{proof}
	Assume the hypotheses. By Lemma \ref{LemHodgePre1}.$(b)$ together with  Lemma \ref{LemHodgePre2}.$(b)$ we have
	\begin{gather*}
		\mathcal{B}=\ran T^*\ho\ker T=\ran T^*\ho\ker(T^*T+UU^*)\ho\ran U.
	\end{gather*}
	Furthermore, by Lemma \ref{LemHodgePre2}.$(a)$ we have
	\begin{gather*}
		\ran (T^*T+UU^*)=[\ker(T^*T+UU^*)]^\perp\\=(\ker T\cap \ker U^*)^\perp=\ran T^*\ho\ran U.
	\end{gather*}
	This completes the proof.
\end{proof}

\begin{definition}
	We call the decomposition (\ref{DefHodgeDecomp}) an (abstract) Hodge decomposition.
\end{definition}
Orthogonal projections associated with abstract Hodge decompositions play a key role in this paper and so the following corollary is complementary to Theorem \ref{ThmHodgeDecomp}.
\begin{corollary}
    Suppose $\mathcal{B}$ has an abstract Hodge decomposition (\ref{DefHodgeDecomp}). Let $\Gamma_{\ran T^*}, \Gamma_{\ran U},$ and $\Gamma_{\ker (T^*T+UU^*)}$ denote the orthogonal projections of $\mathcal{B}$ onto $\ran T^*, \ran U,$ and  $\ker (T^*T+UU^*),$ respectively. Then
    \begin{gather}
        \Gamma_{\ran T^*}=T^+T,\; \Gamma_{\ran U}=UU^+,\\
        \Gamma_{\ker (T^*T+UU^*)}=I_{\mathcal{B}}-\Gamma_{\ran T^*}-\Gamma_{\ran U},
    \end{gather}
    where $(\cdot)^+$ denotes the Moore-Penrose pseudoinverse and $I_{\mathcal{B}}$ is the identity operator on $\mathcal{B}$.
\end{corollary}
\begin{proof}
    This follows immediately from the abstract Hodge decomposition (\ref{DefHodgeDecomp}) of $\mathcal{B}$ and the fundamental properties of the Moore-Penrose pseudoinverse [see statements $(4)$ and $(5)$ in Lemma \ref{LemMPProp}].
\end{proof}

\nocite{*}
\bibliography{aipsamp}

\end{document}